\documentclass[10pt, twocolumn]{article}
\usepackage{fullpage,amsmath,amssymb,amsfonts,amsthm}

\usepackage[numbers,sort&compress]{natbib}
\usepackage{graphicx}
\usepackage{grffile}
\usepackage{xcolor}
\usepackage[margin=1cm, font=footnotesize]{caption}
\usepackage{hyperref}
\usepackage[capitalise,noabbrev]{cleveref}
\crefname{equation}{Eq.}{Eqs.}

\usepackage{algorithm, algorithmic}
\usepackage{authblk}

\usepackage{tabularx}
\usepackage{multirow}

\usepackage{siunitx}
\usepackage{xparse}
\usepackage{physics}
\usepackage{tikz}
\usetikzlibrary{quantikz}
\usetikzlibrary{decorations.markings}

\usepackage{commath}

\newtheorem{theo}{Theorem}[section]
\newtheorem{defi}[theo]{Definition}

\newtheorem{coro}[theo]{Corollary}
\newtheorem{lemma}[theo]{Lemma}

\DeclareMathOperator{\geo}{geo}

\DeclareMathOperator{\diam}{diam}

\newcommand{\R}{\mathbb{R}}

\newcommand{\Z}{\mathbb{Z}}

\newcommand{\Prob}{{\mathbb P}}

\newcommand{\Vb}{{V_\bullet}}
\newcommand{\Vw}{{V_\circ}}

\newcommand{\pdf}{{f}}
\newcommand{\pdfAD}{{f}}

\newcommand{\hardoutcome}{\hat{m}}
\newcommand{\idealoutcome}{\overline{m}}
\newcommand{\softoutcome}{m}

\DeclareMathOperator{\T1m}{\tau_A}

\DeclareMathOperator{\soft}{soft}

\usepackage{soul}

\begin{document}

\title{Improved quantum error correction using soft information}

\author[1]{Christopher A. Pattison}
\author[2]{Michael E. Beverland}
\author[2]{Marcus P. da Silva}
\author[2]{Nicolas Delfosse}

\affil[1]{Caltech, Institute for Quantum Information and Matter, Pasadena, USA}
\affil[2]{Microsoft Quantum and Microsoft Research, Redmond, USA}

\twocolumn[
  \begin{@twocolumnfalse}
    \maketitle
\begin{abstract}
The typical model for measurement noise in quantum error correction is
to randomly flip the binary measurement outcome. In experiments,
measurements yield much richer information---e.g., continuous current
values, discrete photon counts---which is then mapped into binary
outcomes by discarding some of this information. In this work, we
consider methods to incorporate {\em all} of this richer information,
typically called {\em soft information}, into the decoding of quantum
error correction codes, and in particular the surface code.  We
describe how to modify both the Minimum Weight Perfect Matching and
Union-Find decoders to leverage soft information, and demonstrate
these soft decoders outperform the standard {\em (hard)} decoders that
can only access the binary measurement outcomes.  Moreover, we observe
that the soft decoder achieves a threshold $25\%$ higher than any hard
decoder for phenomenological noise with Gaussian soft measurement
outcomes.  We also introduce a soft measurement error model with
amplitude damping, in which measurement time leads to a trade-off
between measurement resolution and additional disturbance of the
qubits.  Under this model we observe that the performance of the
surface code is very sensitive to the choice of the measurement
time---for a distance-19 surface code, a five-fold increase in
measurement time can lead to a thousand-fold increase in logical error
rate.  Moreover, the measurement time that minimizes the physical
error rate is distinct from the one that minimizes the logical
performance, pointing to the benefits of jointly optimizing the
physical and quantum error correction layers.
\end{abstract}

\bigskip
  \end{@twocolumnfalse}
]

Due to the presence of errors and noise in quantum hardware, quantum
error correction is essential for the scalability and usefulness of
quantum computers~\cite{shor1996}. The most promissing approaches to
fault-tolerant quantum error correction revolve around the surface
code~\cite{kitaev2003top_codes, raussendorf2007fault}, which tolerates
high error rates and is naturally implemented one a square grid of
qubits using only local gates. Numerical simulations show that the
surface code performs well for a variety of noise
models~\cite{dennis2002tqm}, including gates corrupted by
stochastic~\cite{raussendorf2007fault, fowler2009high} and
coherent~\cite{bravyi2018surface_code_coherent_noise} errors.  The
error models that are typically considered include a very simple
representation of noisy measurements where binary measurement
outcomes during computations are flipped with some probability.

However, measurements in physical realizations of quantum computers
are rarely, if ever, binary. Measurement outcomes are physically
represented by much richer quantities, such as continuous currents or
voltage, and are converted into binary outcomes by additional
processing that discards some information. This is no different from
the classical setting, where information is similarly represented by
currents and voltages (including non-binary information, such as in
the case of flash memories~\cite{wang2011softLDPC}). In classical
error correction, this richer physical representation, often referred
to as {\em soft information} (in contrast to the {\em hard} or {\em
sharp} features of of binary information), is exploited to obtain
greater tolerance to errors, or to reduce the noise required to
achieve some logical error rate~\cite{costello2007channel}. This
naturally raises the question of how to extract similar benefits from
soft information in the quantum setting. However, the classical
approach of directly measuring the soft values of all bits and
computing parities in noiseless post-processing cannot be applied to
the quantum setting, as it would destroy the quantum superpositions
being protected. Moreover, parities computed with quantum circuits can
propagate errors, since these quantum circuits are noisy. 
Fault-tolerant quantum error correction with soft information requires
non-trivial modifications.

This challenge has been partially addressed in systems with continuous
variable encodings of
qubits~\cite{gottesman2001encoding,fukui2017analog,vuillotquantum,
noh2020fault, noh2021low, chamberland2020building}.  Here we consider
a more general setting, independent of the physical qubit encoding and
grounded on the conditional distributions of the measurement
outcomes. Our main result is the design of generic {\em soft decoders}
for the surface code: a soft minimum weight perfect matching (MWPM)
decoder (which we show identifies a most likely fault set for soft
information), and a soft Union-Find (UF) decoder (which is an
approximation to the MWPM decoder, but with comparable performance and
low computational complexity). Both these decoders are modifications
of existing surface code decoders~\cite{dennis2002tqm,delfosse2017UF},
with similar computational overhead, but may also be applied to other
codes~\cite{bombin2007homological,delfosse2014projection,chao2020optimization,delfosse2021unionfind,hastings2021dynamically}.

We evaluate the performance of these soft decoders numerically by
introducing concrete measurement error models with soft information:
one where the measurement outcome is correpted by white Gaussian noise
(what we call {\em Gaussian soft noise}), and another where the
measurement outcome is corrupted by amplitude damping and Gaussian
noise (what we call {\em Gaussian soft noise with amplitude damping}).
We observe several remakable properties of the soft UF decoder in the
presence of soft Gaussian noise: the soft UF decoder outperforms the
hard UF decoder (both in terms of threshold and peformance below
threshold), and it also achieves a threshold that is 25\% higher than
the optimal threshold achievable by any hard
decoder~\cite{wang2003confinement_higgs}.  Moreover, we find that
logical performance is highly sensitive to the measurement time. In
the case of soft Gaussian noise with amplitude damping, minimizing
physical measurement error rate can lead to logical error rates that
are nearly 1000 times larger over what is achievable, despite such
minimization leading to a measurement time that is only 5 times larger
than the optimum.

We briefly review surface codes and provide some introductory
definitions in \cref{sec:background}.  This section also describes our
general framework to describe soft measurements, and gives specific
examples including a model of the amplitude damping channel during
measurement.  In \cref{sec:graphical_model}, we define a general noise
model with soft measurement noise and we discuss the generalization of
the traditional phenomenological model and circuit models.  The soft
MWPM decoder and the soft UF decoder are presented
in \cref{sec:soft_decoding}.  Then, we prove
in \cref{sec:decoder_success_proof} that the soft MWPM decoder returns
a most likely fault set. We also provide a sufficient condition for
the success of the soft UF decoder, proving that these two soft
decoders perform well.  In \cref{sec:numerics}, we evaluate the
performance of the soft UF decoder for the phenomenological model and
the circuit model with soft noise.  Finally,
in \cref{sec:measure_time_tradeoff}, we use a toy model to illustrate
the benefits of optimizing measurement with awareness of the
fault-tolerant error correction protocol.

\section{Background and definitions} \label{sec:background}

\subsection{Graph and hypergraphs}

In this section, we review the language of graph and hypergraph theory~\cite{berge1973graphs}.

A {\em graph} is a pair $G = (V, E)$ where $V$ is the vertex set and $E$ is the edge set. 
An edge $e \in E$ is a pair of vertices $e = \{u, v\}$. In this work, we consider only finite graphs. 
We will allow graphs to contain multiple copies of the same edge $\{u, v\}$.

To describe subsets of edges and vertices, it is convenient to introduce the language of chain complexes. 
A {\em 1-chain} in a graph is defined to be a formal sum of edges $x = \sum_{e \in E} x_e e$ where $x_e \in \Z_2$.
The 1-chain $x$ can be interpreted as the subset of edges $e$ such that $x_e = 1$.
Conversely, any subset of $E$ defines a $1$-chain.
Similarly, a {\em 0-chain}, which represents a subset of $V$, is defined to be a formal sum of vertices $y = \sum_{v \in V} y_v v$ where $y_v \in \Z_2$.
Given the correspondence between subsets and chains, we use the notation $e \in x$ (respectively $v \in y$) to refer to the fact that $e$ (respectively $v$) belongs to the subset of edges corresponding to $x$ (respectively $y$).

For $i=0, 1$, the set of $i$-chains, denoted $C_i$, is a $\Z_2$-linear space equipped with the component-wise binary addition.
The {\em boundary map} is a $\Z_2$-linear map $\partial$ from $C_1$ to $C_0$.
The boundary of an edge $e = \{u, v\}$ is defined to be 
$
\partial(e) = u + v.
$
and by linearity we have 
$
\partial(x) = \sum_{e \in E} x_e \partial(e)
$
for all $x \in C_1$.
In other words, the boundary of a set of edges $A \subset E$ is the set of vertices that are incident with an odd number of edges of $A$.
Note that if $e_1$ and $e_2$ are two edges in $E$ which correspond to the same pair of vertices, we have that $\partial(e_1+e_2)=0$ by addition in $\Z_2$.

We also consider the restriction of the boundary to a subset of vertices.
Given a subset $U \subset V$, the {\em restricted boundary map} 
$
\partial_U: C_1 \rightarrow C_0
$
is defined by
$
\partial_U(x) = \sum_{v \in \partial(x) \cap U} v.
$

In this work, we also consider finite hypergraphs. 
A {\em hypergraph} is defined as a pair $H = (V, E)$ where $V$ is the vertex set and $E$ is the set of hyperedges.
The difference from a graph is that a hyperedge can be an arbitrary subset of $V$ that may contain more than two vertices.
Like in the case of graphs, a hypergraph can contain multiple copies of the same hyperedge.
When the context is clear, we will often use the term edges to refer to the hyperedges of a hypergraph.

The chain complex language extends to hypergraphs.
A 1-chain in a hypergraph is a formal sum of hyperedges and a 0-chain is a formal sum of vertices.
The boundary of a hyperedge $e = \{v_1, \dots, v_m\}$ is defined to be 
$
\partial(e) = v_1 + \dots + v_m
$
and the boundary map is extended to all 1-chains by linearity as for graphs.

\subsection{Surface codes}
\label{subsec:surface_codes}

In this section, we briefly review\footnote{Refs.~\cite{dennis2002tqm, fowler2012surface, litinski2019game} provide more complete overviews of quantum error correction and quantum computation with surface codes.} the implementation of quantum error correction with the surface code~\cite{kitaev2003fault}.
We focus on the rotated surface code represented in \cref{fig:rotated_code}(a). 
However, all the results of this article generalize immediately to any surface code with or without boundary, and also to hyperbolic surface codes~\cite{breuckmann2017hyperbolic}.

The rotated surface code with minimum distance $d$ encodes one logical qubit in a grid of $d \times d$ data qubits.
Error correction with the surface code is based on so-called {\em plaquette measurements} which each return one outcome bit.
Based on the outcomes extracted, a classical algorithm known as a decoder is used to identify errors which corrupt the data qubits.

A plaquette measurement is a weight-four Pauli measurement $X_a X_b X_c X_d$ or $Z_a Z_b Z_c Z_d$ acting on the four qubits $a, b, c, d$ of a face of the grid of qubits.
The boundary plaquettes represent weight-two measurements $X_e X_f$ or $Z_e Z_f$ acting on the two qubits of a face along the boundary of the grid. 

A Pauli operator which is a product of plaquette operators is called a {\em stabilizer} of the surface code. This is because it acts trivially on encoded states.
We sometimes refer to the measured plaquette operators as {\em stabilizer generators}.

By symmetry, $X$ errors and $Z$ errors can be corrected with the same strategy.
In this section, we describe the correction of $X$ errors based on the outcome of the measurement of $Z$ plaquettes. 

\begin{figure}
\centering
\begin{tikzpicture}[
    scale=1,
    ]

    \foreach \i in {0, ..., 1} {
        \foreach \j in {0, ..., 1} {
            \draw[fill=gray!60] ($(2*\i,1+2*\j)$) +(0,0) -- +(1,0) -- +(1,1) -- +(0,1) -- cycle;
            \draw[fill=gray!60] ($(1+2*\i,2*\j)$) +(0,0) -- +(1,0) -- +(1,1) -- +(0,1) -- cycle;
        }
    }
    \foreach \i in {0, ..., 1} {
        \draw[fill=gray!60] ($(2*\i-1,0)$) +(1,0) -- +(2,0) arc(0:-180:0.5) -- cycle;
        \draw[fill=gray!60] ($(2*\i+1,2)$) +(0,2) -- +(1,2) arc(0:+180:0.5) -- cycle;
    }

    \foreach \i in {0, ..., 1} {
        \foreach \j in {0, ..., 1} {
            \draw[fill=gray!20] ($(  2*\i,  2*\j)$) +(0,0) -- +(1,0) -- +(1,1) -- +(0,1) -- cycle;
            \draw[fill=gray!20] ($(1+2*\i,1+2*\j)$) +(0,0) -- +(1,0) -- +(1,1) -- +(0,1) -- cycle;
        }
    }
    \foreach \i in {0, ..., 1} {
        \draw[fill=gray!20] ($(0,2*\i+1)$) +(0,0) -- +(0,1) arc(90:270:0.5) -- cycle;;
        \draw[fill=gray!20] ($(2,2*\i-1)$) +(2,1) -- +(2,2) arc(90:-90:0.5) -- cycle;;
    }

    \foreach \i in {0, ..., 4} {
        \foreach \j in {0, ..., 4} {
            \draw[fill=gray] (\i, \j) circle (0.1);
        }
    }
\end{tikzpicture}

(a) 

\includegraphics[scale=.4]{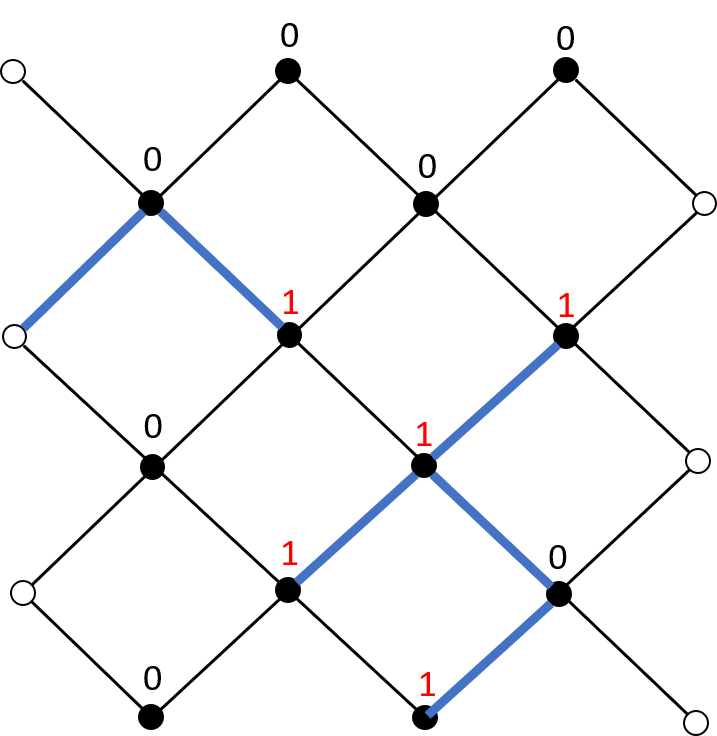}

(b)

\caption{
(a) The distance-five rotated surface code, which encodes one logical qubit into a grid of $5\times 5$ data qubits (grey vertices). 
Blue and white plaquettes support $Z$ and $X$ measurements respectively.
(b) A graphical representation of $X$ errors in the distance-five surface code. 
Each of the 25 data qubits corresponds to an edge of the graph.
Black vertices support $Z$ plaquette measurements.
White vertices are added to provide a second endpoint for edges associated with a qubit that is incident to a single $Z$ measurement.
During a stabilizer measurement round, each black vertex provides an outcome bit for a $Z$ plaquette. 
The set of blue edges represents an $X$ error, 
and the outcome in a black vertex is given by the parity of the number of incident blue edges.
}
\label{fig:rotated_code}
\end{figure}

\medskip
First let's assume that there is some error on the qubits but that the plaquette measurements are performed perfectly.
The value of the outcome bit returned by a plaquette measurement is then given by the parity of the number of qubits in the plaquette supporting the error. 
As a result, errors and measurement outcomes exhibit a natural graph structure as one can see in \cref{fig:rotated_code}(b).
An $X$ error, which is supported on the edges of the decoding graph $G_{X}$, can be represented by the 1-chain 
$
x = \sum_{e \in E} x_e e
$
such that $x_e = 1$ iff the edge $e$ suffers from an $X$ error.
The set of outcomes produced by all the plaquette measurements is represented by a 0-chain
$
m = \sum_{v \in V} m_v v
$
such that $m_v = 1$ iff the outcome bit returned in vertex $v$ is non-trivial.

In practice, the measurement of a plaquette is not perfect.
It is typically implemented using a small circuit that uses an additional ancilla qubit placed in the center of the plaquette.
The ancilla qubit is prepared in an initial state $\ket 0$ or $\ket +$ and a sequence of CNOT gates is performed between the ancilla qubit and the plaquette's qubits.
Then, the outcome of the plaquette measurement is extracted by measuring the ancilla qubit in the $Z$ or $X$ basis.
It is possible to simultaneously extract the outcome bit of all plaquettes in 6 time steps of operations by interleaving the CNOT gates in a particular order (see \cref{subsec:circuit_model}).
We refer to the extraction of bits for the full set of a plaquettes of the surface code a {\em stabilizer measurement round}.

A single stabilizer measurement round is not enough to achieve good performance, because the measurement circuits are implemented with imperfect gates and the outcome measured is unreliable.
To account for measurement errors, the standard solution is to perform error correction based on multiple consecutive rounds of measurement data, as we describe in more detail later.

\subsection{Soft measurement}
\label{subsec:measurement}

Measurements are subject to noise, as are other operations in quantum hardware.
Three standard models of noise are typically considered in the field of quantum error correction. These all assume the depolarizing channel during operations on the qubits (including idling operations), but the models differ in how they treat plaquette measurements. 
The simplest is the \emph{ideal measurement model} where plaquette measurements are assumed to be perfect.
In the \emph{phenomenological noise model}, this is amended by assuming that following each perfect plaquette measurement the binary outcome can be flipped before being reported.
The \emph{circuit noise model} is the most realistic of the standard models.
In this case, plaquette measurements are implemented using explicit circuits built from more basic operations, including single-qubit measurements with a possible flip of the binary outcome. 

In all three of these models, the outcomes of the stabilizer
measurements are discrete, and therefore decoding algorithms are
typically designed to operate using binary measurement outcomes.
However, due to physical details of the measurement apparatus, the
output from a physical system will not be a simple binary value, and
may range from photon counts to continuous voltage waveform. This
output may be processed into a binary outcome, but this processing may
lead to information loss.

Here we seek to go beyond this paradigm. 
To do so, we introduce a simple model of measurement which produces a more general output; see \cref{fig:measurement_model}.
More specifically, in our model of measurement there is first a
projection of the state into a subspace $\pi_{\bar{\mu}}$ as in the
standard models, but the \emph{ideal outcome} $\bar{\mu} \in \{0, 1\}$
cannot be observed directly.  What is observed is instead a \emph{soft
  outcome} $\mu$, whose distribution is given by the probability
density function $\pdf^{(\bar{\mu})}(\mu)$, which is conditioned on
the ideal outcome $\bar{\mu} = 0$ or $1$. Our formalism can easily
handle both discrete and continuous measurement outcomes with minimal
modifications. For brevity, we focus on continuous-valued soft
measurement outcomes in this article.

Given the soft outcome, we can try to infer the ideal outcome by
mapping the soft outcome onto a binary outcome that we refer to as the
{\em hard outcome} $\hat{\mu}$.  One particular choice for this {\em
  hardening map} corresponds to a maximum likelihood assignment, where
\begin{align} \label{eq:def_hard_outcome}
\hat \mu = 
\begin{cases}
0 & \text{ if } \pdf^{(0)}(\mu) \geq \pdf^{(1)}(\mu),\\
1 & \text{ otherwise.}
\end{cases}
\end{align}

\begin{figure}
\centering
\includegraphics[width=0.45\textwidth]{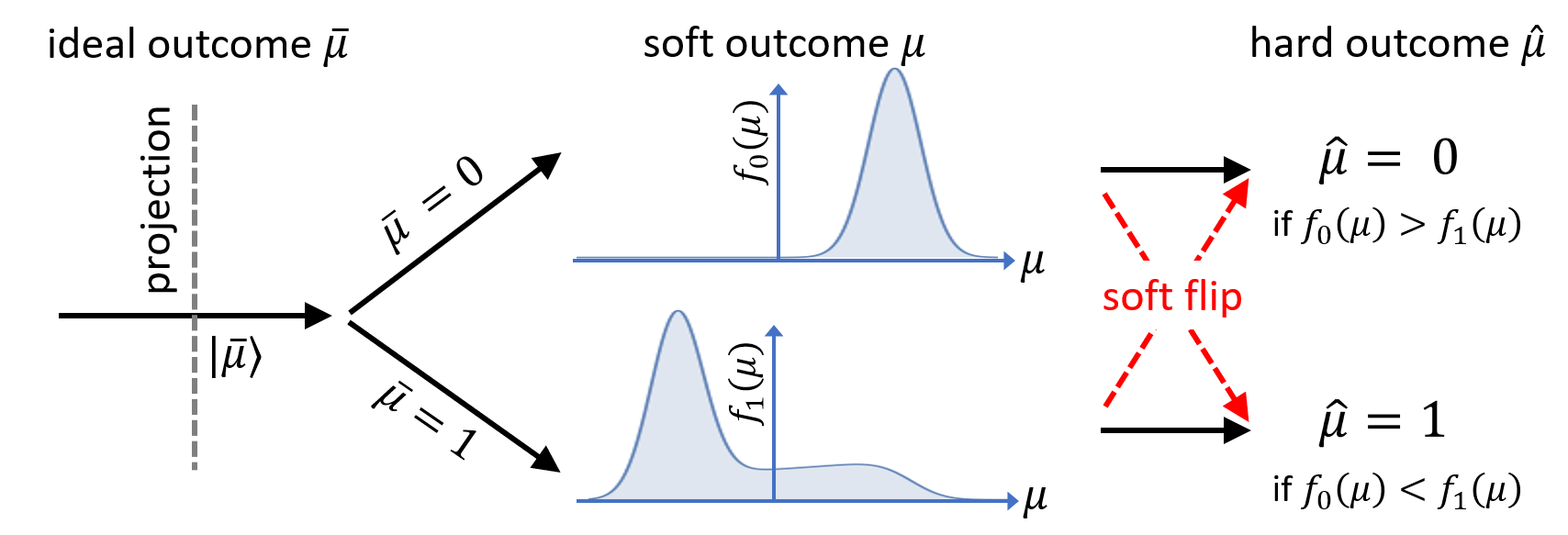}
\caption{
We model measurements as a perfect projection with ideal discrete outcome $\bar{\mu}$, followed by a noisy process which outputs a continuous soft outcome $\mu$.
The soft outcome can be processed to produce a discrete guess $\hat{\mu}$ of the ideal outcome, for example by comparing the probability density functions for $\mu$ conditioned on different cases of $\bar{\mu}$.
There can be some loss of information in going from a soft to a hard outcome.
When the hard outcome does not match the ideal outcome, we say a soft flip has occurred.
Standard decoders use the hard outcome as their input, whereas we propose using the soft outcome directly for decoding.
}
\label{fig:measurement_model}
\end{figure}

We say that a {\em soft flip} occurs during a measurement if the hard outcome inferred is not equal to the ideal outcome.
If the ideal outcome is $\bar \mu$, the probability of a soft flip  is given by
\begin{align} \label{eqn:soft_flip_prob}
\Prob(\text{soft flip} | \bar \mu) 
= 
\int_{f^{(\bar \mu')}(\mu) > f^{(\bar \mu)}(\mu)} \pdf^{(\bar \mu)}(\mu) \text{d}\mu,
\end{align}
where $\bar \mu' = \bar \mu + 1 \pmod 2$.
In general, this probability depends on the value of the ideal outcome $\bar \mu$.

We say that a measurement is {\em symmetric} if the probability of a soft flip is the same for both values of the ideal outcome $\bar \mu = 0$ or 1.
We then define the probability $p_{M, \soft}=\Prob(\text{soft flip} | 0) =\Prob(\text{soft flip} | 1) $ that a hard decoder interprets the soft outcome as a flip after hardening.
In this work we often consider scenarios in which there can first be a flip of the ideal measurement outcome with probability $p_M$ preceding symmetric soft noise.
Then the probability $p_{M,\text{hardened}}$ of a flip overall after hardening is
\begin{align}
p_{M,\text{hardened}} = p_{M} + p_{M, \soft} - p_{M}p_{M, \soft}.\label{eq:net_hardened_flip_prob}
\end{align}
A simple example of a symmetric soft measurement is the {\it Gaussian soft noise} in which a $\pdf^{(0)} = {\cal N}(+1, \sigma^2)$ and $\pdf^{(1)} = {\cal N}(-1, \sigma^2)$.
Some physical processes may lead to asymmetric measurements, as we
discuss next in \cref{subsec:amplitude_damping_model}.

\subsection{Example: Amplitude damping}
\label{subsec:amplitude_damping_model}

Here we provide a toy model~\cite{Gambetta2007measurement} for soft
measurement noise based on {\em amplitude
  damping}~\cite{nielsenChuang} which is common in solid state quantum
devices such as superconducting qubits and quantum
dots~\cite{Blais2004cqed,Colless2013dotdispersive}---we call this
model {\em Gaussian soft noise with amplitude damping}. In this case
the probability density functions $\pdf^{(0)}$ and $\pdf^{(1)}$ can be
tuned by setting the relative duration of three timescales $\tau_M$,
$\T1m$ and $\tau_F$; see \cref{fig:simple-ad-intuition}.  We therefore
write these as $\pdfAD^{(0)}(\mu;\tau_M,\T1m,\tau_F)$ and
$\pdfAD^{(1)}(\mu;\tau_M,\T1m,\tau_F)$, making their dependence on
$\tau_M,\T1m,\tau_F$ explicit.

Under this noise model, at any
infinitesimal time interval the qubit may decay from $\ket{1}$ to
$\ket{0}$ with a fixed time-independent rate
$1/\T1m$, but no transitions from $\ket{0}$ can occur.  
Here we refer to $\T1m$ as the {\it amplitude damping time}, although in the physics
literature it is often referred to as the $T_1$ time.
This asymmetry
arises because in many systems the $\ket{1}$ state has higher energy
than the $\ket{0}$ state, and unintentional interactions with a
low-temperature bath lead to energy transfer from the qubit to the
bath, resulting in this decay (for this reason, amplitude damping is
also known as {\em energy relaxation}). 

Since measurement relies on the accumulation of information about the
state over the {\it measurement time} $\tau_M$, any qubit decays during that
time will increase the probability that a $\ket{1}$ at the beginning
of the measurement is mistaken for a $\ket{0}$, but not the other
way around. 

In our amplitude damping model model, there is a signal $S(t)$ which 
builds up over time $t$ starting from $S(0)=0$.  
The signal's behavior is mathematically
equivalent to a one-dimensional continuous random walk with drift, where the sign
of the drift depends on whether the system is in the state $\ket{0}$
or $\ket{1}$.  Specifically, if the state of the system is $\ket{0}$
during an infinitesimal time interval from $t$ to $t+\text{d}t$, the
mean of $S$ is increased by an amount proportional to $\text{d}t$.  On
the other hand if the state of the system is $\ket{1}$ during this
interval, the mean of $S$ is decreased by the same amount.
Irrespective of the state of the system, the variance of $S$ increases
by an amount proportional to $\text{d}t$ during this time interval.
The ratio between the variance and the square of the mean of the
increment over a small time interval is $\tau_F/\text{d}t$, where the
constant $\tau_F$ is refered to as {\em the fluctuation time}.  The
signal is read out at a time $\tau_M$ giving $S(\tau_M)$ which we take
to be the continuous measurement outcome $\mu$ (after scaling such
that the mean of $\mu$ is $+1$ when $\bar{\mu}=0$).

Although the derivation and precise forms of the conditional distributions 
under this toy model are somewhat involved (see \cref{app:toy-model-measurement}), their qualitative behavior is rather simple, as illustrated by \cref{fig:simple-ad-intuition}.
When $\tau_M\ll\T1m$, the distributions are well approximated by
Gaussians.  These Gaussians have an overlap that decreases
monotonically with $\tau_M / \tau_F$, so and increasing $\tau_M /
\tau_F$ lowers ${\mathbb P}(\text{soft flip}|\bar{\mu})$.  As $\tau_M$
approaches $\T1m$, $\pdfAD^{(0)}$ remains unchanged but $\pdfAD^{(1)}$
is distorted and shifts towards $\pdfAD^{(0)}$, so that as
$\frac{\tau_M}{\T1m}$ increases ${\mathbb P}(\text{soft flip}|1)$
approaches $\frac{1}{2}$ (the distributions overlap completely).  A
regime of particular interest in this model is when
$\tau_F\ll\tau_M\ll\T1m$, so that $\pdfAD^{(1)}$ is not significantly
distorted ($\tau_M\ll\T1m$), while the overlap between the conditional
distributions is also small ($\tau_F\ll\tau_M$).

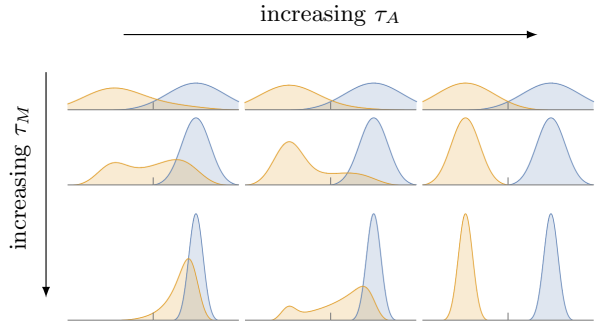
\begin{figure}
  \begin{center}
  \begin{tikzpicture}
    \tikzstyle{flecha}=[draw=gray]
    \tikzstyle{guide}=[->,draw=black,>=latex,line width=0.15mm]

    \node[above, inner sep=0] (image11) at (-\textwidth/7,0) {
      \includegraphics[width=\textwidth/7]{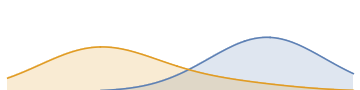}
    };
    \node[above, inner sep=0] (image12) at (0,0) {
        \includegraphics[width=\textwidth/7]{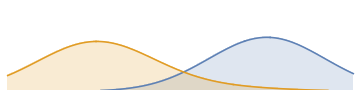}
    };
    \node[above, inner sep=0] (image13) at (\textwidth/7,0) {
        \includegraphics[width=\textwidth/7]{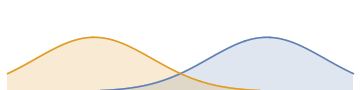}
    };

    \node[above, inner sep=0] (image21) at (-\textwidth/7,-1) {
      \includegraphics[width=\textwidth/7]{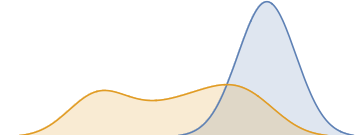}
    };
    \node[above, inner sep=0] (image22) at (0,-1) {
        \includegraphics[width=\textwidth/7]{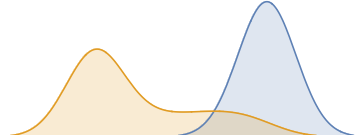}
    };
    \node[above, inner sep=0]  (image23) at (\textwidth/7,-1) {
        \includegraphics[width=\textwidth/7]{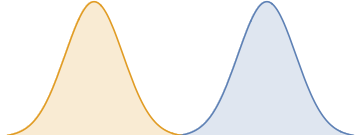}
    };

    \node[above, inner sep=0] (image31) at (-\textwidth/7,-2.8) {
      \includegraphics[width=\textwidth/7]{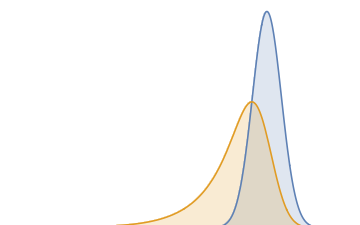}
    };
    \node[above, inner sep=0] (image32) at (0,-2.8) {
        \includegraphics[width=\textwidth/7]{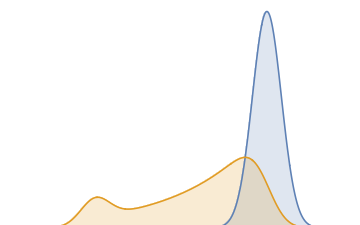}
    };
    \node[above, inner sep=0] (image33) at (\textwidth/7,-2.8) {
        \includegraphics[width=\textwidth/7]{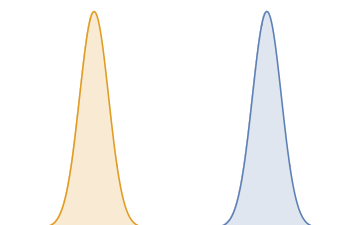}
    };

    \draw [flecha] (-\textwidth/7,0) -- ++ (0,0.1);
    \draw [flecha] (0,0) -- ++ (0,0.1);
    \draw [flecha] (\textwidth/7,0) -- ++ (0,0.1);

    \draw [flecha] (-\textwidth/7,0) -- ++ (-\textwidth/14.5,0);
    \draw [flecha] (-\textwidth/7,0) -- ++ (\textwidth/14.5,0);

    \draw [flecha] (0,0) -- ++ (-\textwidth/14.5,0);
    \draw [flecha] (0,0) -- ++ (\textwidth/14.5,0);

    \draw [flecha] (\textwidth/7,0) -- ++ (-\textwidth/14.5,0);
    \draw [flecha] (\textwidth/7,0) -- ++ (\textwidth/14.5,0);

    \draw [flecha] (-\textwidth/7,-1) -- ++ (0,0.1);
    \draw [flecha] (0,-1) -- ++ (0,0.1);
    \draw [flecha] (\textwidth/7,-1) -- ++ (0,0.1);

    \draw [flecha] (-\textwidth/7,-1) -- ++ (-\textwidth/14.5,0);
    \draw [flecha] (-\textwidth/7,-1) -- ++ (\textwidth/14.5,0);

    \draw [flecha] (0,-1) -- ++ (-\textwidth/14.5,0);
    \draw [flecha] (0,-1) -- ++ (\textwidth/14.5,0);

    \draw [flecha] (\textwidth/7,-1) -- ++ (-\textwidth/14.5,0);
    \draw [flecha] (\textwidth/7,-1) -- ++ (\textwidth/14.5,0);

    \draw [flecha] (-\textwidth/7,-2.8) -- ++ (0,0.1);
    \draw [flecha] (0,-2.8) -- ++ (0,0.1);
    \draw [flecha] (\textwidth/7,-2.8) -- ++ (0,0.1);

    \draw [flecha] (-\textwidth/7,-2.8) -- ++ (-\textwidth/14.5,0);
    \draw [flecha] (-\textwidth/7,-2.8) -- ++ (\textwidth/14.5,0);

    \draw [flecha] (0,-2.8) -- ++ (-\textwidth/14.5,0);
    \draw [flecha] (0,-2.8) -- ++ (\textwidth/14.5,0);

    \draw [flecha] (\textwidth/7,-2.8) -- ++ (-\textwidth/14.5,0);
    \draw [flecha] (\textwidth/7,-2.8) -- ++ (\textwidth/14.5,0);

    \draw [guide] (-\textwidth/6,1) -- ++ (\textwidth/3,0);

    \node[above,font=\small] at (0,1) {increasing $\tau_A$};

    \draw [guide] (-2.75\textwidth/12,0.5) -- ++ (0,-3);

    \node[left,font=\small,rotate=90] at (-3\textwidth/12,0.2) {increasing $\tau_M$};

  \end{tikzpicture}
  \end{center}
  \caption{Conditional probability density functions $\pdfAD^{(0)}$
    (blue) and $\pdfAD^{(1)}$ (yellow) for $\tau_M/\tau_F$ increasing from
    top to bottom, and $\T1m/\tau_F$ increasing from left to right. The width of
    $\pdf^{(0)}$ is proportional to $\sqrt{\tau_F/\tau_M}$, and the scales are
    chosen so that for $\tau_M\ll\tau_A$ the Gaussians are
    centered at $\pm1$. Intuitively, increasing measurement time
    $\tau_M$ narrows the conditional distributions (making
    them more distinguishable), while decreasing amplitude damping
    lifetime $\tau_A$ shifts the $\bar\mu=1$ distribution towards the
    $\bar\mu=0$ distribution (making them less distinguishable).
  \label{fig:simple-ad-intuition}}
\end{figure}

\section{Graphical models for soft noise in the surface code}
\label{sec:graphical_model}

In this section, we propose a general graphical noise model to describe errors and soft measurement outcomes for surface codes, including correlated errors and repeated measurements.
This formalism includes the standard noise models mentioned in \cref{subsec:surface_codes} namely ideal measurement noise, phenomenological noise and circuit noise~\cite{dennis2002tqm} as special cases. 
Later in \cref{sec:soft_decoding} we define decoders which correct the errors in graphical noise models, and in particular which make use of the soft data to outperform standard decoders which ignore this soft information.

\subsection{General graphical model}
\label{subsec:graphical_model}

Here we present a general graphical model for arbitrary Pauli noise with soft measurement outcomes. 
Faults can affect the outcomes of plaquette measurements of the surface code and can leave residual errors on the data qubits.
The following definition captures this notion.

\begin{defi} \label{def:graphical_model_hard}
A {\em graphical model} for $T$ rounds of measurements with the surface code is a quadruple $(G_T, p, f, \pi)$ defined by
\begin{itemize}
\item
{\bf Fault hypergraph:} 
A hypergraph $G_T = (V_T, E_T)$ such that $V_T$ contains a vertex $(a, t)$ for each plaquette $a$ and for each round $t = 1, \dots, T$.
\item 
{\bf Fault probability:}
A value $p_e \in [0, 1]$ for each hyperedge $e \in E_T$.
\item
{\bf Measurement noise:}
A pair of probability density functions $(\pdf^{(0)}_{a, t}, \pdf^{(1)}_{a, t})$ for each plaquette $a$ and for each round $t = 1, \dots, T$ describing the measurement outcome in this location.
\item 
{\bf Residual error:}
A Pauli error $\pi_e$ on the data qubits of the surface code for each hyperedge $e \in E_T$.
\end{itemize}
\end{defi}

We call the graph $G_T$ the \emph{fault hypergraph} or simply fault graph when it is a proper graph.
To construct the graph $G_T$ in this case, we begin by including all the vertices from $G_X$ and $G_Z$ described in \cref{subsec:surface_codes} for each of the $T+1$ layers. 
The set of {\em measurement vertices} of $G_T$, denoted $\Vb$, is the set of vertices $(a, t)$ corresponding to the spacetime locations of the plaquettes, where $a$ corresponds to the plaquette and $t$ labels the measurement round.
Each hyperedge of $G_T$ represents a potential fault and the vertices contained in a hyperedge correspond to the spacetime locations of measurements that detect this fault.
In practice, a fault may correspond to a Pauli error $X, Y$ or $Z$ or the flip of an outcome bit. 
It could also be a combination of multiple Pauli errors and outcome bit flips.
The fault $e$ induces a {\em residual error} $\pi_e$ on the data qubits at the end of the $T$ rounds of measurement. 
The hypergraph $G_T$ may contain multiple copies of the same edge $e$ with different residual errors $\pi_e$.
The graph $G_T$ contains additional vertices which we call {\em boundary vertices}, denoted $\Vw = V \backslash \Vb$.
\cref{fig:hyperedges} shows an example of a graphical noise model which corresponds to $T=3$ rounds of measurement for the distance $d=5$ surface code.

\begin{figure}
\centering
\includegraphics[width=0.37\textwidth]{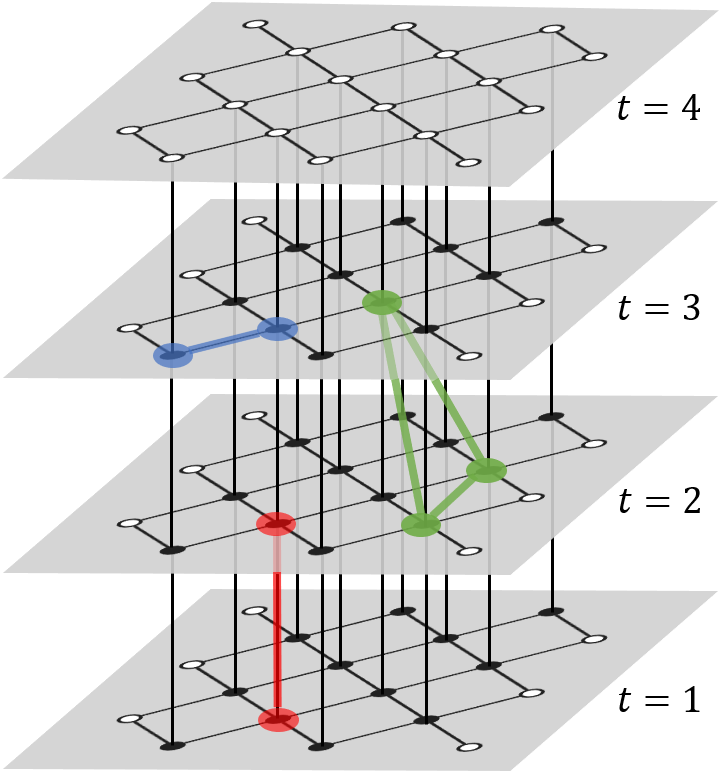}
\caption{
Examples of hyperedges corresponding to faults. 
A single-qubit $X$ error corresponds to a singe horizontal edge (blue).
A flipped measurement outcome corresponds to a singe vertical edge (red).
More general faults are possible, and correspond to hyperedges. 
Here we show an example of a weight-two Pauli error and the resulting syndrome that could arise from a single fault in a stabilizer measurement circuit (green).
}
\label{fig:hyperedges}
\end{figure}

In some cases, it is useful to define independent noise models for $X$ type errors and $Z$ type errors in which case we define a separate hypergraph for each.
The $X$-type hypergraph contains vertices $(a, t)$ for $Z$ plaquettes only.

\subsection{Sampling from graphical models}
\label{subsec:syndrome}

In our graphical noise models, faults occurs independently of one another.
Correlated errors acting on different qubits at potentially different times can occur since each fault has no restriction on the support of the residual error or on the measurement vertices which it flips.

A {\em fault set} for the graphical model $G_T$ is defined to be a 1-chain $x$ of $G_T$, that is a formal sum of hyperedges $x = \sum_{e \in E_T} x_e e$.
The {\em probability of a fault set} $x$ is
\begin{align} \label{eq:def_error_proba}
\Prob(x) = \prod_{e \in x} p_e \prod_{e \notin x} (1 - p_e),
\end{align}
where each fault $e$ occurs with probability $p_e$.

For a trivial fault set, $x = 0$, the {\em ideal outcome} in any measurement vertex $(a, t)$ is $\idealoutcome_{a, t} = 0$.
A fault $e$, containing a measurement vertex $(a, t)$, induces a change of the value of the ideal outcome $\idealoutcome_{a, t'}$ for all $t' \geq t$.
The ideal outcome associated with a general fault set $x$ is obtained by combining the effects of all its faults $e \in x$.

The soft outcome $\softoutcome_{a, t}$ and the hard outcome $\hardoutcome_{a, t}$ are generated from the ideal outcome using the probability density functions $\pdf^{(0)}_{a, t}$ and $\pdf^{(1)}_{a, t}$ as explained in \cref{subsec:measurement}.
Note that although we use the term `ideal outcome', this outcome depends on faults that have occurred at earlier times.

For the decoders we define later it is also convenient to introduce the syndrome which measures the changes between two consecutive rounds of measurements.
The {\em syndrome}, denoted $\hat s(\softoutcome)$, is a 0-chain in $G_T$ which can be calculated from the observed soft measurement outcomes $\softoutcome$ by first using \cref{eq:def_hard_outcome} to identify the hard outcomes $\hardoutcome$, and then
\begin{align} \label{eq:syndrome_def}
\hat s_{a, t} = \hardoutcome_{a, t} + \hardoutcome_{a, t-1} \pmod 2
\end{align}
for all measurement vertices $(a, t)$, with the convention $\hardoutcome_{a, 0} = 0$.
The value of the syndrome in boundary vertices is 0.
The syndrome provides the same information as the hard outcome but it is more convenient for the decoder to use the syndrome as an input instead of the hard outcome.

\subsection{Example: Soft phenomenological noise}
\label{subsec:phenom_model}

Here we define a generalization of the standard phenomenological noise model introduced in~\cite{dennis2002tqm} to include soft Gaussian noise affecting measurement outcomes. 

First we define the noise model's $X$-type \emph{fault graph} $G_T$, by stacking $T+1$ copies of the graph $G_X$ which was described in \cref{subsec:surface_codes} and represented in \cref{fig:rotated_code}(b), and connecting consecutive layers with vertical edges; see \cref{fig:phenom_graph}.
The graph $G_T$ has two types of vertices.
Specifically, the set of {\em measurement vertices} of $G_T$, denoted $\Vb$, is the set of vertices $(a, t)$ corresponding to the space time locations of the plaquettes, where $a$ corresponds to a plaquette and $t=1, \dots, T$ labels the measurement round. The remaining vertices are {\em boundary vertices}, denoted $\Vw = V \backslash \Vb$.
The vertices of the $T+1$ layer are treated differently from other vertices in $G_T$, in that they correspond to perfect measurements.

\begin{figure}
\centering
\includegraphics[width=0.37\textwidth]{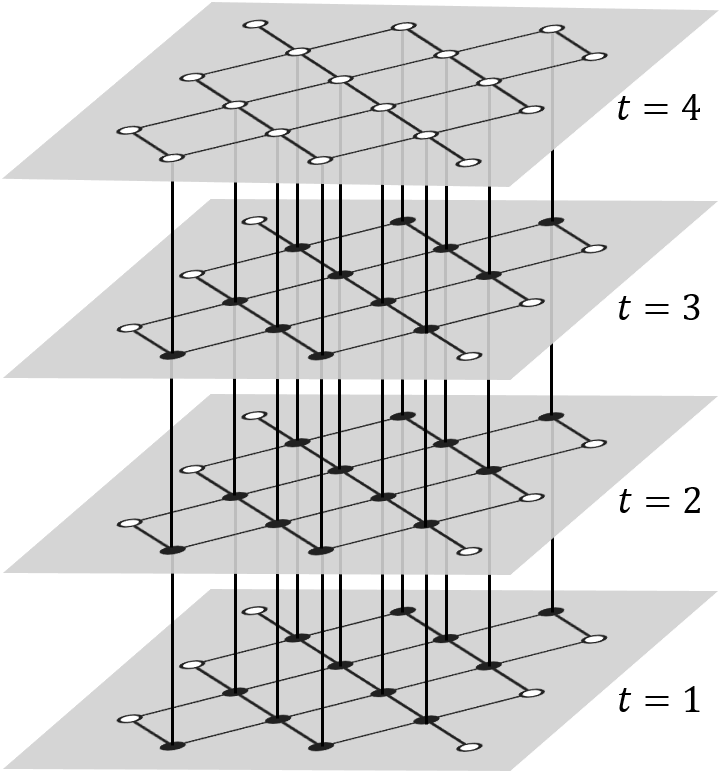}
\caption{
Fault graph for the soft phenomenological noise model with $T=3$ rounds of measurements for the surface code with distance five. 
The graph contains $T+1=4$ copies of the graph represented in \cref{fig:rotated_code}(b).
There is a copy of the graph for each time step and an additional copy corresponding to a round of perfect measurement is added for a proper termination. 
Horizontal edges correspond to qubit errors and vertical edges encode ideal outcome flips.
}
\label{fig:phenom_graph}
\end{figure}

Furthermore, a probability $p_D$ is assigned to each horizontal edge, and a probability $p_M$ is assigned to each vertical edge in $G_T$.
This corresponds to qubits being affected by independent $X$ errors with probability $p_D$ and independent ideal measurement outcome flips with probability $p_M$.
The pair of probability density functions $\pdf^{(0)}$ and $\pdf^{(1)}$ are assigned to each measurement vertex in $G_T$.

In this graphical description of the noise model, edges represent potential faults and the vertices contained in an edge correspond to the measurements whose ideal value will \emph{change} due to the presence of this fault.
For instance, the fault corresponding to the horizontal edge $\{(a,t),(b,t)\}$ in $G_T$ is an $X$ error being applied just before round $t$ of measurements to the qubit associated with the edge $(a,b)$ in $G_X$.
This fault occurs with the probability assigned to the edge, namely $p_D$, and it results in a change of the ideal outcome associated with the plaquettes $a$ and $b$ at time $t$.
In other words, it results in a change of the ideal outcomes of vertices $(a,t)$ and $(b,t)$ when compared with $(a,t-1)$ and $(b,t-1)$ respectively. 
In the absence of additional faults, these vertices then retain their new ideal outcome for all rounds $t' \geq t$.
The fault corresponding to the vertical edge $\{(a,t),(b,t+1)\}$ in $G_T$ is a flip of the ideal outcome of the plaquette associated with $a$ in $G_X$ only at round $t$.
This fault occurs with probability $p_M$, and corresponds to a change in the ideal outcome of the vertex $(a,t)$ with respect to the vertex $(a,t-1)$ and then a further change of the ideal outcome of the vertex $(a,t+1)$ with respect to the vertex $(a,t)$. 

If the ideal outcomes were reported, this would precisely reproduce the standard phenomenological noise model introduced in Ref.~\cite{dennis2002tqm}. 
Instead, the reported output of each plaquette in each round is the soft outcome obtained by sampling from the probability density functions assigned to the vertices of $G_T$, namely $\pdf^{(0)}$ for plaquettes with ideal outcome $0$ and $\pdf^{(1)}$ for plaquettes with ideal outcome $1$.

\subsection{Example: Soft circuit noise}
\label{subsec:circuit_model}

Here we provide a generalization of the standard circuit noise model to include soft measurement. 
Our model involves $T$ consecutive stabilizer measurement rounds, where 
each measurement round is implemented by the circuit shown in \cref{fig:rotated_sc_measurement_order}.

\begin{figure}
    \center{
        \begin{tikzpicture}[
    scale=2,
    CNOT/.pic = {
        \begin{scope}[xshift=0.28cm]
            \draw[line width = 0.6] (-0.2,0) -- (1,0);
            \draw[line width = 0.6] (0,0) circle (0.2);
            \draw[line width = 0.6] (0,-0.2) -- (0,0.2);
            \draw[fill = black] (1,0) circle (0.08);
        \end{scope}
    },
    CZ/.pic = {
        \begin{scope}[xshift=0.28cm]
            \draw[line width = 0.6] (-0.2,0) -- (1,0);
            \draw[line width = 0.2, fill = white] (-0.2,-0.2) rectangle (0.2, 0.2) (0,0) node {\rotatebox{\numexpr #1 + 90 \relax}{\tiny $Z$}};
            \draw[fill = black] (1,0) circle (0.08);
        \end{scope}
    }
]
\draw[fill=gray!60] (0,0) -- (1,0) -- (1,1) -- (0,1) -- cycle;
\draw[fill=gray!60] (1,1) -- (2,1) -- (2,2) -- (1,2) -- cycle;
\draw[fill=gray!60] (1,0) -- (2,0) arc(0:-180:0.5) -- cycle;
\draw[fill=gray!60] (0,2) -- (1,2) arc(0:+180:0.5) -- cycle;

\draw[fill=gray!20] (1,0) -- (2,0) -- (2,1) -- (1,1) -- cycle;
\draw[fill=gray!20] (0,1) -- (1,1) -- (1,2) -- (0,2) -- cycle;
\draw[fill=gray!20] (0,0) -- (0,1) arc(90:270:0.5) -- cycle;
\draw[fill=gray!20] (2,1) -- (2,2) arc(90:-90:0.5) -- cycle;

\draw[fill=gray] (0,0) node(V00){} circle (0.1);
\draw[fill=gray] (1,0) node(V10){} circle (0.1);
\draw[fill=gray] (2,0) node(V20){} circle (0.1);
\draw[fill=gray] (0,1) node(V01){} circle (0.1);
\draw[fill=gray] (1,1) node(V11){} circle (0.1);
\draw[fill=gray] (2,1) node(V21){} circle (0.1);
\draw[fill=gray] (0,2) node(V02){} circle (0.1);
\draw[fill=gray] (1,2) node(V12){} circle (0.1);
\draw[fill=gray] (2,2) node(V22){} circle (0.1);

\draw (V00.north east) pic[scale=0.8, rotate= 45] {CNOT} +( 0.23, 0.07) node {\small $3$};
\draw (V01.south east) pic[scale=0.8, rotate=315] {CNOT} +( 0.23,-0.07) node {\small $1$};
\draw (V10.north west) pic[scale=0.8, rotate=135] {CNOT} +(-0.23, 0.07) node {\small $4$};
\draw (V11.south west) pic[scale=0.8, rotate=225] {CNOT} +(-0.23,-0.07) node {\small $2$};

\draw (V11.north east) pic[scale=0.8, rotate= 45] {CNOT} +( 0.23, 0.07) node {\small $3$};
\draw (V12.south east) pic[scale=0.8, rotate=315] {CNOT} +( 0.23,-0.07) node {\small $1$};
\draw (V21.north west) pic[scale=0.8, rotate=135] {CNOT} +(-0.23, 0.07) node {\small $4$};
\draw (V22.south west) pic[scale=0.8, rotate=225] {CNOT} +(-0.23,-0.07) node {\small $2$};

\draw (V02.north east) pic[scale=0.8, rotate= 45] {CNOT} +( 0.23, 0.07) node {\small $3$};
\draw (V12.north west) pic[scale=0.8, rotate=135] {CNOT} +(-0.23, 0.07) node {\small $4$};

\draw (V10.south east) pic[scale=0.8, rotate=315] {CNOT} +( 0.23,-0.07) node {\small $1$};
\draw (V20.south west) pic[scale=0.8, rotate=225] {CNOT} +(-0.23,-0.07) node {\small $2$};

\draw (V01.north east) pic[scale=0.8, rotate= 45] {CZ={ 45}} +( 0.23, 0.07) node {\small $2$};
\draw (V11.north west) pic[scale=0.8, rotate=135] {CZ={135}} +(-0.23, 0.07) node {\small $4$};
\draw (V02.south east) pic[scale=0.8, rotate=315] {CZ={315}} +( 0.23,-0.07) node {\small $1$};
\draw (V12.south west) pic[scale=0.8, rotate=225] {CZ={225}} +(-0.23,-0.07) node {\small $3$};

\draw (V10.north east) pic[scale=0.8, rotate= 45] {CZ={ 45}} +( 0.23, 0.07) node {\small $2$};
\draw (V20.north west) pic[scale=0.8, rotate=135] {CZ={135}} +(-0.23, 0.07) node {\small $4$};
\draw (V11.south east) pic[scale=0.8, rotate=315] {CZ={315}} +( 0.23,-0.07) node {\small $1$};
\draw (V21.south west) pic[scale=0.8, rotate=225] {CZ={225}} +(-0.23,-0.07) node {\small $3$};

\draw (V00.north west) pic[scale=0.8, rotate=135] {CZ={135}} +(-0.07, 0.27) node {\small $4$};
\draw (V01.south west) pic[scale=0.8, rotate=225] {CZ={225}} +(-0.07,-0.27) node {\small $3$};

\draw (V21.north east) pic[scale=0.8, rotate= 45] {CZ={ 45}} +( 0.07, 0.27) node {\small $2$};
\draw (V22.south east) pic[scale=0.8, rotate=315] {CZ={315}} +( 0.07,-0.27) node {\small $1$};

\end{tikzpicture}
    }
    \caption{\label{fig:rotated_sc_measurement_order}
        Standard stabilizer extraction circuit for a distance 3 rotated surface code (see for instance \cite{tomita2014low}).
        The circuit for an arbitrary distance code is obtained by translating this circuit for all stabilizers.
        Note that there are some idle circuit locations during which CNOT gates are being applied and others during which measurement qubits are being applied, which we treat separately in our model since measurement and gate times can be different.
    }
\end{figure}
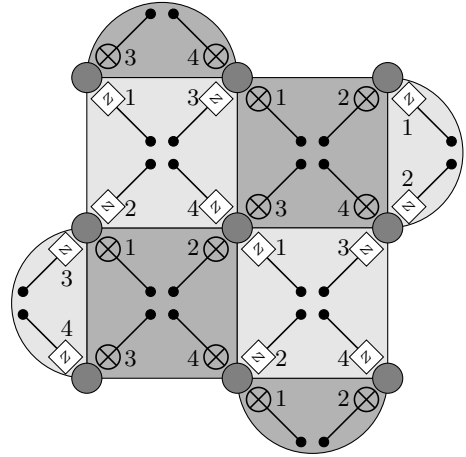

Each of the components of the circuit can fail in specific ways, and for each of these possible failures we introduce a hyperedge to the graph $G_T$. 
In particular, the following failures can occur for the circuit components:
\begin{itemize}
	\item Idle qubits waiting for CNOT gates: $X$, $Y$, $Z$ each occurs independently with probability $p_{IG}/3$. 
	\item Idle qubits waiting for measurements: $X$, $Y$, $Z$ each occurs independently with probability $p_{IM}/3$. 
	\item CNOT gates: $XI$, $YI$, $ZI$, $IX$, $XX$, $YX$, $ZX$, $IY$, $XY$, $YY$, $ZY$, $IZ$, $XZ$, $YZ$, $ZZ$ each occurs independently with probability $p_{\text{CNOT}}/15$.
	\item Measurement outcomes (ideal): have their ideal outcome flipped with probability $p_M$. 
	\item Measurement outcomes (soft): given the ideal outcome $\bar{\mu} \in \{0,1\}$, the continuous outcome $\mu$ is sampled from the probability density function $\pdf^{(\bar{\mu})}(\mu)$ as described in \cref{subsec:measurement}.
\end{itemize}
A few remarks are called for before describing how the graphical noise model is constructed. 
First note that the circuit in \cref{fig:rotated_sc_measurement_order} only has time steps during which either CNOTs or measurements are performed (but not both), making it natural to consider the two above scenarios for idle qubits since a CNOT gate may have a very different duration than a measurement. 
Secondly, here each fault (even faults associated with the same operation in the circuit) are applied independently.
At first this may seem slightly different from the standard circuit noise considered in the literature in different faults that occur on the same operation are exclusive.  
However, the standard exclusive circuit noise model and the inclusive circuit noise model assumed here are actually exactly equivalent as proven in Appendix~E of~\cite{chao2020optimization}. 
When mapping between the inclusive and exclusive noise models, the probability with which each fault occurs can change, although the change is very small in the regimes of interest. 
We discuss this in more detail in \cref{app:inclusive_exclusive}.

To construct the graph $G_T$ in this case, we begin by including all the vertices from $G_X$ and $G_Z$ for $T+1$ layers from \cref{subsec:surface_codes}. 
For each of the possible faults listed in the first four bullet points (i.e. excluding the soft measurement) that can occur in the $T$ rounds of stabilizer measurements implemented by the circuit, a hyperedge is added to the graph $G_T$.
The vertices included in the hyperedge are those vertices whose measurement outcomes will flip if that fault occurs.
The probability assigned to the hyperedge is simply the probability that the fault occurs.
The residual error associated with the hyperedge is identified by observing what Pauli operator is left on the data qubits at the end of the measurement round during which the fault occurs.
Lastly, each measurement vertex has the pair $\{ \pdf^{(0)}(\mu), \pdf^{(1)}(\mu) \}$ of probability density functions assigned.

\section{Soft decoding}
\label{sec:soft_decoding}

In this section, we propose two efficient decoders for surface codes with soft measurements.
First, we propose a soft version of the minimum weight perfect matching (MWPM) decoder~\cite{dennis2002tqm}.
Different variants of the MWPM decoder exploiting some soft information were considered previously for the special case of GKP surface codes~\cite{vuillotquantum, noh2020fault, noh2021low, chamberland2020building}.
Although its complexity is polynomial, the soft MWPM decoder may be too slow for some applications.
Next, we propose a soft version of the Union-Find (UF) decoder~\cite{delfosse2017UF} which has lower complexity and therefore could be more practical.

\subsection{Decodability constraints}
\label{subsec:dec_constraints}

The ultimate goal of a decoder is to correct the data qubits, that is to cancel the effect of the residual errors of the faults which occur.
The two decoders proposed in this article return a fault set $x$ for a given set of soft outcomes.
The correction to apply to the data qubits is the product of the residual errors $\pi_e$ corresponding to the faults $e$ in $x$.

We will restrict ourselves to graphical models $(G_T, p, f, \pi)$ that satisfy the following {\em decodability conditions}:

\medskip
(C1) All the edges in $G_T$ have rank two.

\medskip
(C2) For all $e \in E_T$, we have $p_e < 0.5$.

\medskip
Condition (C1) guarantees that $G_T$ is a graph instead of a general hypergraph.
We require this because we expect the decoding problem for a general hypergraph to be too difficult to solve.
For example, the problem of identifying a least likely fault set in that case includes the three-dimensional matching problem which is NP-hard~\cite{garey1979NP}.
Condition (C2) excludes pathological cases where faults occur with excessively high probability.

\subsection{Soft decoding graph}
\label{subsec:decoding_graph}

In \cref{subsec:graphical_model} we defined the fault hypergraph. 
Here, we introduce another graph which will be used for decoding.
The {\em decoding graph}, denoted $\tilde G_T = (\tilde V_T, \tilde E_T)$, associated with the graphical model $(G_T, q, f, \pi)$ is the graph obtained from $G_T$ by adding an edge (which we call a soft vertical edge) connecting $(a, t)$ with $(a, t+1)$ when $(a, t)$ is a measurement vertex.
The graph $G_T$ may already contain a vertical edge connecting these two vertices.
In that case, the decoding graph $\tilde G_T$ contains two edges connecting $(a, t)$ and $(a, t+1)$ that we refer to as the soft vertical edge and the hard vertical edge.
Moreover, all the boundary vertices of $\tilde G_T$ are identified to single vertex, denoted $v_g$, that we refer to as the {\em ghost vertex}.
By definition, we have $\tilde V_T = \Vb \cup \{v_g\}$ and $\tilde E_T$ is the set $E_T$ augmented with soft vertical edges.

Consider a measurement vertex $(a, t)$. let $\mu$ be the soft outcome observed in this vertex and $\hat \mu$ be the corresponding hard outcome. 
Denote by $\hat \mu' = \hat \mu + 1 \pmod 2$ the other value of the hard outcome.
To define edge weights in the decoding graph, it is convenient to introduce the {\em likelihood ratio} 
\begin{align} \label{eq:def_likelihood_ratio}
L_{a, t}(\mu)
=
\frac{\pdf^{(\hat \mu')}_{a, t}(\mu)}{\pdf^{(\hat \mu)}_{a, t}(\mu)} \cdot
\end{align}
This ratio can be computed from the knowledge of the probability density functions and the value of the soft outcome $\mu$ because $\hat \mu$ and $\hat \mu'$ are derived from $\mu$.
By definition of the hard outcome, we have $L_{a, t}(\mu) \in [0, 1]$.

Given the model parameters $p_e$ and the set of soft outcomes $\softoutcome$ observed over measurement vertices, we define edge weights for the decoding graph.
Hard and soft edges have different weights.
The weight of the soft edge $e$ between $(a, t)$ and $(a, t+1)$ is defined as a function of the soft outcome, whereas the hard edge weight depends on the parameter $p_e$.
Precisely, the edge weight is defined by
\begin{align} \label{eq:def_edge_weight}
w(e) = 
\begin{cases}
- \log L_{a, t}(m_{a, t}) & \text{if $e$ is soft,}\\
- \log(p_e /(1-p_e)) & \text{otherwise.}
\end{cases}
\end{align}
These edge weights are non-negative numbers.
In some cases, we denote this weight by $w_{\tilde G_T}(e)$ to avoid any ambiguity about which graph the weight is defined with respect to.

\subsection{Soft Minimum Weight Perfect Matching decoder}
\label{subsec:soft_mwpm}

Our soft MWPM decoder is specified in \cref{algo:soft_mwpm}.
The key ingredient is the distance graph that we introduce below.

\begin{algorithm}
\caption{Soft MWPM decoder}
\label{algo:soft_mwpm}
\begin{algorithmic}[1]
\REQUIRE The decoding graph $\tilde G_T$ for $T$ rounds of measurements. A set of soft outcomes $\softoutcome$.
\ENSURE A fault set, {\em i.e.} a 1-chain $x$ in $G_T$.
\STATE For each soft vertical edge $e$, compute $w_{\tilde G_T}(e)$ as a function of $\softoutcome$ using \cref{eq:def_edge_weight}.
\STATE Compute the syndrome $\hat s$ from $\softoutcome$ using \cref{eq:syndrome_def}. Let $v_1, \dots, v_k$ be the set of vertices of $\tilde G_T$ with non-trivial syndrome.
\STATE Construct the distance graph $K(\hat s)$ with vertex set $V_K = \{v_1, \dots, v_k\}$ and compute the edge weights $w_K(\{v_i, v_j\}) = d_{\tilde G_T}(v_i, v_j)$ using Dijkstra's algorithm.
\STATE Compute a minimum weight perfect matching $M$ in $K(\hat s)$.
\STATE For each edge $\{u, v\} \in M$, compute a geodesic $\geo(u, v)$ in the graph $\tilde G_T$.
\RETURN the 1-chain $\tilde x = \sum_{\{u, v\} \in M} \geo(u, v)$ restricted to the graph $G_T$.
\end{algorithmic}
\end{algorithm}

We consider the distance in the decoding graph $\tilde G_T$ induced by the edge weights $w_{\tilde G_T}(e)$.
The length of a path in $\tilde G_T$ is defined to be the sum of the weights of its edges.
The {\em distance} between two vertices $u$ and $v$, denoted $d_{\tilde G_T}(u, v)$, is defined to be the minimum length of a path connecting $u$ and $v$.
A minimum length path joining $u$ and $v$ is called a {\em geodesic}.
Given $u, v \in V_T$, let $\geo(u, v)$ be the set of edges supporting a geodesic connecting $u$ and $v$. 
If the decoding graph contains multiple geodesics between $u$ and $v$, $\geo(u, v)$ is chosen arbitrarily among them.
We will consider $\geo(u, v)$ as a 1-chain in the decoding graph.

Given the decoding graph $\tilde G_T$ and a syndrome $\hat s$, the {\em distance graph} $K(\hat s) = (V_K, E_K)$ is defined to be the fully connected graph with vertex set $V_K$ given by the set of measurement vertices supporting a non-trivial syndrome.
If $K(\hat s)$ contains an odd number of vertices, we add the ghost vertex to $V_K$.
We define edge weights in $K$ based on the distance in the decoding graph, that is $w_K(\{u, v\}) = d_{\tilde G_T}(u, v)$.

Like in the standard MWPM decoder, the basic idea of \cref{algo:soft_mwpm} is to compute a minimum set of paths that connect pairs of syndrome vertices in the decoding graph.
This is done using a Minimum Weight Perfect Matching (MWPM) algorithm~\cite{edmonds1965mwpm}.
The distance graph contains an even number of vertices and its edge weights are non-negative making the application of the MWPM algorithm straightforward.

The main difference between our soft MWPM decoder and the standard hard MWPM decoder~\cite{dennis2002tqm} is that the weights in the soft decoding graph depend on the set of outcomes $\softoutcome$ whereas they are fixed in the hard case.
As a result, the distances between all pairs of nodes can be precomputed in the standard MWPM but not in the soft MWPM decoder. 
Similarly, one cannot precompute a geodesic for each pair of nodes in the soft decoding graph.
This increases the execution time of the decoder.
However, as we explain below, the worst-case asymptotic complexity is no worse than the standard MWPM decoder.

The worst case complexity of the standard MWPM decoder is dominated by the call of the MWPM algorithm. 
Based on Blossom~V, a popular implementation of a MWPM algorithm designed by Kolmogorov~\cite{kolmogorov2009blossom} which has a worst case complexity in $O(n_v^3 n_e)$ for a graph with $n_v$ vertices and $n_e$ edges, one can achieve a worst case complexity in $O(|V_T|^5)$ for the hard MWPM decoder, that is $O(d^{15})$ if $T=d$ rounds of measurement are used.
Other implementations of the MWPM decoder achieve a better worst-case complexity but they might not run faster in practice (see Table~I in~\cite{cook1999blossom}).
The most favorable asymptotic scaling for a MWPM algorithm~\cite{vazirani1994mwpm} leads a worst-case complexity in $O(\sqrt{n_v} n_e)$ for the matching subroutine and a MWPM decoder complexity in $O(d^{7.5})$ for $T=d$.

\medskip
If Blossom~V is used, the complexity of our soft variant of the MWPM decoder remains the same, $O(d^{15})$, dominated by the cost of this MWPM subroutine.
We can also achieve $O(d^{7.5})$ worst-case complexity using the MWPM algorithm of~\cite{vazirani1994mwpm} despite the fact that distances and geodesic between the syndrome nodes must be computed on the fly in our soft MWPM decoder.
The computation of all the distances $d_{\tilde G_T}(v_i, v_j)$ is performed in step 3 and 4 of the algorithm with worst-case complexity $O(|V_T||E_T| + |V_T|^2 \log |V_T|)$ using Dijkstra's algorithm. Assuming $T=d$ measurement rounds, this subroutine runs in $O(d^6 \log d)$.
The geodesics are computed in step 7 of \cref{algo:soft_mwpm} also thanks to Dijkstra's algorithm with the same worst-case complexity as the distances calculation.

\subsection{Soft Union-Find decoder}
\label{sec:soft_unionfind}

Here we provide a soft version of the UF decoder~\cite{delfosse2017UF, Huang2020UF}, which achieves an error correction performance similar to the MWPM decoder but with better time complexity.
Our soft UF decoder is specified in \cref{algo:soft_uf}.

\begin{figure}
\centering
\includegraphics[scale=.4]{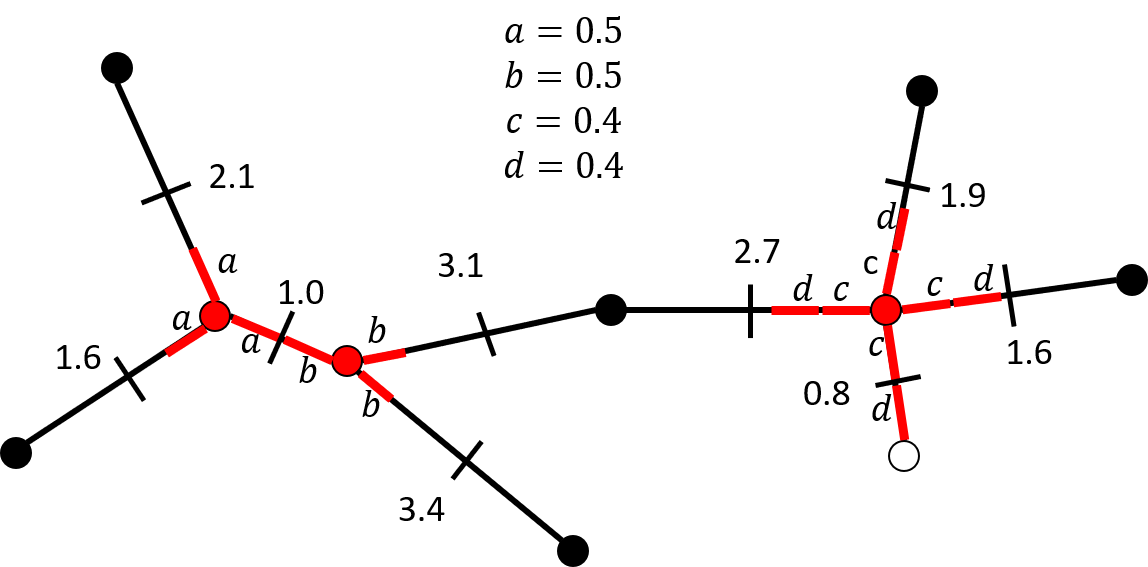}
\caption{
Illustration of the cluster growth subroutine of the UF decoder.
Edges of the decoding graph are split in the middle.
The three red vertices support a non-trivial syndrome.
Initially, there are three odd clusters, one for each syndrome vertex. 
One odd cluster has perimeter four and the other two have perimeter three.
At the first growth step, we select a cluster with minimum perimeter and we grow the three edges $(a)$ by 0.5.
Then, the edges $(b)$, which belong to the least recently grown cluster with perimeter three, are grown.
This connects two odd clusters forming an even cluster.
During the next two steps ($(c)$ and $(d)$), the remaining odd cluster grows until it meets a boundary node and becomes even.
}
\label{fig:uf_growth}
\end{figure}

The soft UF decoder works in two steps:
First, we grow clusters of qubits in the decoding graph until each of the clusters can be corrected separately.
The growth subroutine is shown in \cref{fig:uf_growth}.
Then, the peeling decoder~\cite{delfosse2020peeling} is used to find a correction inside the grown clusters in linear time.
In what follows, we explain the growth procedure used in our soft UF decoder.

To describe our soft UF decoder, it is convenient to introduce the {\em split-edge} graph $\tilde H_T$ obtained from the decoding graph $\tilde G_T$ by adding a vertex in the middle of each edge.
Each edge $e$ of the decoding graph splits into two edges of $\tilde H_T$ whose weights are given by $w_{\tilde G_T}(e)/2$.
The boundary vertices of $\tilde H_T$ are those corresponding to boundary vertices of $\tilde G_T$.

To keep track of the edge growth, we associate a {\em growth state} $\gamma(e) \in \R$ with each edge $e$ of $\tilde H_T$.
An edge $e$ is said to be {\em fully grown} if $\gamma(e) = w_{\tilde H_T}(e)$.
The {\em cluster} of a vertex $v$ of $\tilde H_T$ is defined to be the set of vertices of $\tilde H_T$ that can be reached from $v$ through a path of fully grown edges.

Given a syndrome $\hat s$, a cluster of $\tilde H_T$ is said to be {\em even} if it contains an even number of vertices of $\hat s$ or at least one boundary vertex.
A cluster that is not even is said to be {\em odd}.
If a cluster is even, we can use the peeling decoder to determine a fault set included in the cluster which cancels the syndrome of this cluster~\cite{delfosse2020peeling}.
Therefore, we will grow odd clusters until all clusters are even and can be corrected by peeling.

To determine which cluster to grow first, we consider the set $B(C)$ of edges of $\tilde H_T$ connecting a cluster $C$ and its complementary.
The {\em perimeter} of the cluster $C$ is defined as the number of edges in $B(C)$.
At the beginning of the growth procedure, the growth state of each edge is initialized to $\gamma(e) = 0$.
Then, at each growth step, we select an odd cluster $C$ with minimum perimeter and grow all edges of $B(C)$ by incrementing the growth states by
\begin{align} \label{eq:Gamma_C}
\Gamma_C = \min_{e \in B(C)}(w_{\tilde H_T}(e) - \gamma(e)).
\end{align}
$\Gamma_C$ is the smallest value $\Gamma$ such that growing all edges of $B(C)$ by $\Gamma$ fills at least one edge.
If multiple odd clusters have the same perimeter, we prioritize the least recently grown cluster.

After the growth, we are left with only even clusters, and we will call the peeling decoder to find an estimation of the fault set inside the grown clusters.
We consider the set $\varepsilon \subset \tilde E_T$ of edges of the decoding graph $\tilde G_T$ for which both halves are fully grown in the split-edge graph $\tilde H_T$.
The peeling decoder takes as input the subset $\varepsilon \subset \tilde E_T$ and a syndrome $\hat s$ and computes, in linear time, a 1-chain $\tilde x\subset \varepsilon$ in $\tilde G_T$ such that $\partial_{\Vb}(\tilde x) = \hat s$.
We will use the notation $\tilde x = \textbf{Peeling}(\varepsilon, \hat s)$.
The restriction of $\tilde x$ to $G_T$ is the output the soft UF decoder.

\begin{algorithm}
\caption{Soft UF decoder}
\label{algo:soft_uf}
\begin{algorithmic}[1]
\REQUIRE The graph $\tilde H_T$ for $T$ rounds of measurements. 
A set of soft outcomes $\softoutcome$.
\ENSURE A 1-chain $x$ in $G_T$.
\STATE For each soft vertical edge $e$, compute $w_{\tilde H_T}(e)$ as a function of $\softoutcome$ using \cref{eq:def_edge_weight}.
\STATE Compute the syndrome $\hat s$ from $\softoutcome$ using \cref{eq:syndrome_def}.
\STATE For each edge $e$ of $\tilde H_T$, set $\gamma(e) = 0$.
\STATE While there exists at least one odd cluster do:
\STATE \hspace{1cm} Select an odd cluster $C$ with minimum perimeter. If there are multiple such clusters select the least recently grown.
\STATE \hspace{1cm} Compute $\Gamma_C$ defined in \cref{eq:Gamma_C}.
\STATE \hspace{1cm} For all $e \in B(C)$ do $\gamma(e) \leftarrow \gamma(e) + \Gamma_C$.
\STATE Compute the set $\varepsilon$ of edges $e$ of $\tilde G_T$ such that the two halves of $e$ in $\tilde H_T$ are fully grown.
\STATE Compute the 1-chain $\tilde x = \textbf{Peeling}(\varepsilon, \hat s)$ of $\tilde G_T$.
\RETURN the restriction of $\tilde x$ to the graph $G_T$.
\end{algorithmic}
\end{algorithm}

The growth procedure of our soft UF decoder~\cref{algo:soft_uf} differs from the growth of the variant of the UF decoder proposed in~\cite{Huang2020UF} in two ways.
First, we grow half-edges instead of the original edges of the decoding graph. This is why we perform the growth in the graph $\tilde H_T$ instead of $\tilde G_T$.
Second, we prioritize the least recently grown cluster.
We checked that these two modifications bring a small improvement of the error threshold of hard UF decoder for the standard circuit noise model.

The time complexity of the soft UF decoder depends on the precision required for the edge weights of the soft decoding graph. 
Like in the circuit level study of~\cite{Huang2020UF}, we expect that a finite precision is sufficient to achieve good decoding performance. Our simulations are done with 32-bits of precision for the edge weights.
In that case, the worst case complexity of the soft UF decoder remains $O(d^3 \alpha(d))$, where $\alpha(d)$ is the slowly increasing inverse Ackermann function.
This complexity is significantly better than the soft MWPM decoder. 

\section{Proof of decoding success}
\label{sec:decoder_success_proof}

In this section, we prove that the soft MWPM decoder and the soft UF decoder perform well.
First, we prove that the soft MWPM decoder returns a most likely fault set.
A key technical ingredient to establish this result is Lemma~\ref{lemma:error_proba} which provides the probability of a fault set given a set of soft outcomes as a function of the edge weights in the soft decoding graph.
Then, we propose a sufficient condition for the success of the soft UF decoder.

\subsection{Technical definitions}

Here we provide some definitions which are useful throughout this section.

Given a weighted graph $G$ and a 1-chain $a$, we denote by $|a|_w = \sum_{e \in a} w(e)$ the sum of the edge weights in $a$.
The diameter of a subset $C$ of vertices of the graph, denoted $\diam_w(C)$, is the maximum distance between two vertices of $C$ assuming that edge lengths are given by $w(e)$.
The diameter of a subset of edges or a 1-chain is defined to be the diameter of the set of incident vertices.

A fault set for a graphical model is a 1-chain in the fault graph $G_T$.
Similarly, we can represent the set of soft flips for a given fault set $x$ with a set of soft outcomes $m$ as a 1-chain $x_{\soft}$ in the graph $\tilde G_T$ with support on soft vertical edges.
In what follows, we denote $\tilde x = x + x_{\soft}$ the 1-chain representing a fault set and its soft flips for a given set of soft outcomes.
The value of $x_{\soft}$ depends on $x$ and the set of soft outcomes observed $m$.

The syndrome $\hat s(m)$ can be derived from the value of $\tilde x$ because $\hat m$ can be computed from $x$ and from the location of the soft flips.
As a result, we can talk about a fault set $\tilde x$ with trivial syndrome.

The {\em weighted minimum distance} of a graphical model, denoted $d_w$, is defined to be the minimum value $|x|_w$ for a fault set $x$ such that the residual error is a non-trivial logical error of the surface code.

To guarantee that the soft UF decoder performs well, we assume that the graphical model satisfies the following topological property.

\medskip
(C3) Let $\tilde x$ be a fault set with trivial syndrome.
If the diameter of $\tilde x$ satisfies $\diam_w(\tilde x) < d_w$ then $\pi(x)$ is a stabilizer of the surface code.

\medskip
Both soft circuit noise and soft phenomenological noise satisfy this condition with Gaussian soft measurements and amplitude damping soft measurements, which include the cases we study numerically.

\subsection{Success of the soft MWPM decoder}

The goal of this section is to prove the following result.

\begin{theo}\label{thm:mle_from_mwpm}
For any graphical model satisfying (C1) and (C2), the soft MWPM decoder (\cref{algo:soft_mwpm}) returns a  most likely fault set.
\end{theo}

\begin{proof}
Let $\softoutcome$ be a set of soft outcomes for a graphical model $(G_T, p, f)$ and let $\hat s$ be the corresponding syndrome.
Given $\softoutcome$ as an input, \cref{algo:soft_mwpm} first generates the distance graph $K(\hat s) = (V_K, E_K)$ associated with $\hat s$. 
Then, it computes a minimum weight matching $M \subseteq E_K$ and the corresponding 1-chain
$$
\tilde x_M = \sum_{\{u, v\} \in M} \geo(u, v)
$$
in the decoding graph $\tilde G_T$.
The output of \cref{algo:soft_mwpm} is the fault set $x_M$ which is the restriction of $\tilde x_M$ to $G_T$.
We will prove by contradiction that it is a most likely fault set conditioned on the set of soft outcomes $\softoutcome$.

Assume that there exists a fault set $x'$ (that is a 1-chain in $G_T$) such that $\Prob(x' | \softoutcome) > \Prob(x | \softoutcome)$.
Consider the extension $\tilde x' = x' + x_{\soft}'$ as a 1-chain of $\tilde G_T$ obtained by adding the soft flip of $x'$.
The value of $x'_{\soft}$ can be derived from $x'$ and $\hat s$.
By \cref{lemma:syndrome_is_boundary}, the support of $\tilde x'$ contains a set of edge-disjoint paths in $G_T$ connecting each vertex of $\hat s$ either with another vertex of $\hat s$ or with a boundary vertex.
One can prove the existence of these paths by induction.
If our set of paths contains two paths connecting to the boundary, we merge these two paths into a single paths in $\tilde G_T$.
Then, each path of this set corresponds to an edge in $K(\hat s)$ and the complete set of paths induces a perfect matching $M'$ in $K(\hat s)$.
Moreover, the weight of the matching $M'$, which is a sum of weights of edges included in $\tilde x'$, satisfies $w(M') \leq \sum_{e \in \tilde x} w(e)$. 
Based on \cref{lemma:error_proba}, this weight is equal to
\begin{align*}
w(M') \leq \sum_{e \in \tilde x'} w(e)= C - \log \Prob(x' | \softoutcome) \cdot
\end{align*}
Using the assumption $\Prob(x' | \softoutcome) > \Prob(x_M | \softoutcome)$, this yields
\begin{align*}
w(M') 
& \leq C - \log \Prob(x' | \softoutcome) \\
& < C - \log \Prob(x_M | \softoutcome) \\
& < \sum_{e \in x_M} w_K(e)
= w(M),
\end{align*}
which contradicts the minimality of $M$.
This concludes the proof of the proposition.
\end{proof}

In the proof of \cref{thm:mle_from_mwpm} we relied on the following two lemmas.

\begin{lemma} \label{lemma:syndrome_is_boundary}
Let $x$ be a fault set for the graphical model $(G_T, p, f, \pi)$ and let $m$ be a set of soft outcomes observed in the measurement vertices.
Then, the syndrome of $x$ satisfies
\begin{align} \label{eq:syndrome_equal_boundary}
\hat s(m) = \partial_{\Vb}(\tilde x)
\end{align}
where $\tilde x = x + x_{\soft}$ and $x_{\soft}$ is the 1-chain of soft flips.
\end{lemma}

\begin{proof}
To prove that the syndrome is equal to $\partial_{\Vb}(\tilde x)$, consider the syndrome value at a measurement vertex $(a, t)$.
Assume first that no soft flip occurs in this vertex.
Then, the hard outcome and the ideal outcome are identical and by definition of the syndrome, we have
$
\hat s(a, t) = \idealoutcome(a, t) + \idealoutcome(a, t-1) \pmod 2 \cdot
$
Using the definition of the ideal outcome, we see that this number is the parity of the number of edges of $x$ that are incident to $(a, t)$.
This proves \cref{eq:syndrome_equal_boundary} in the absence of soft flip.
This equation is also satisfied if a soft flip occurs in $(a, t)$ thanks to the additional vertical soft edge incident to $(a, t)$.
\end{proof}

\begin{lemma} \label{lemma:error_proba}
Let $x$ be a fault set for a graphical model $(G_T, p, f, \pi)$ and let $m$ be the set of soft outcomes measured.
The probability $\Prob(x | \softoutcome)$ of the fault set $x$ conditioned on the soft measurement outcome $m$ is given by
\begin{align} \label{eq:error_prob_equal_weight}
\log \Prob(x | \softoutcome) = C + \sum_{e \in \tilde x} w(e)
\end{align}
where $\tilde x = x + x_{\soft}$ and $x_{\soft}$ is the 1-chain of soft flips, for some constant $C$ that is independent of $x$.
\end{lemma}

\begin{proof}
To compute $\Prob(x | \softoutcome)$, we apply Bayes theorem to obtain
\begin{align} \label{eq:lemma_bayes}
\Prob(x | \softoutcome)
\propto
\Prob(\softoutcome | x) \Prob(x)
\cdot
\end{align}
Here, we abuse notation and use $\Prob(\softoutcome | x)$ to represent the probability density function of the continuous random variable $\softoutcome$ conditioned on $x$.
For simplicity, we focus on continuous distributions over sets of soft outcomes $\softoutcome$, but the same argument applies to discrete soft outcome distributions.

We will consider the two terms of the right hand side of \cref{eq:lemma_bayes} separately. 

Based on \cref{eq:def_error_proba}, we have
$
\Prob(x)
= C_1 \prod_{e \in x} \frac{p_e}{1-p_e}
$
where $C_1$ is the product of all the terms $(1-p_e)$ with $e \in E_T$.
Taking the $\log$, this yields
\begin{align} \label{eq:lemma_proba_x}
\log \Prob(x)
= C_1' + \sum_{e \in x} + \log \left( \frac{p_e}{1-p_e} \right),
\end{align}
for some constant $C_1'$.

Consider now $\Prob(\softoutcome | x)$. 
By definition of $\softoutcome$, we have $\Prob(\softoutcome | x) = \Prob(\softoutcome | \idealoutcome)$.
Because the soft outcomes in different measurement vertices $(a, t)$ are independent, we have
\begin{align*}
\Prob(\softoutcome | \idealoutcome)
& = \prod_{(a, t) \in \Vb}
	\pdf^{(\idealoutcome_{a, t})}_{a, t}(\softoutcome_{a, t})
\\
& = 
\prod_{\substack{(a, t) \in \Vb \\ \idealoutcome_{a, t}=\hardoutcome_{a, t}}}
	\pdf^{(\hardoutcome_{a, t})}_{a, t}(\softoutcome_{a, t})
\\
& \times \prod_{\substack{(a, t) \in \Vb \\ \idealoutcome_{a, t}\neq\hardoutcome_{a, t}}} 
	\pdf^{(\hardoutcome_{a, t}+1)}_{a, t}(\softoutcome_{a, t})
\\
& = 
\prod_{(a, t) \in \Vb} 
	\pdf^{(\hardoutcome_{a, t})}_{a, t}(\softoutcome_{a, t})
\\
& \times \prod_{\substack{(a, t) \in \Vb \\ \idealoutcome_{a, t}\neq\hardoutcome_{a, t}}} 
	\frac
	{
		\pdf^{(\hardoutcome_{a, t}+1)}_{a, t}(\softoutcome_{a, t})
	}
	{
		\pdf^{(\hardoutcome_{a, t})}_{a, t}(\softoutcome_{a, t})
	}
\\
& = 
C_2 
\prod_{\substack{(a, t) \in \Vb \\ \idealoutcome_{a, t} \neq \hardoutcome_{a, t}}} L_{a, t}(m_{a, t})
\end{align*}
$ C_2 = \prod_{(a, t) \in \Vb} \pdf_{a, t}^{(\hardoutcome_{a, t})}(\softoutcome_{a, t})$ depends only on the set of soft outcomes $\softoutcome$, since $\hardoutcome$ is fully determined by $\softoutcome$.
The condition $\idealoutcome_{a, t} \neq \hardoutcome_{a, t}$ in the second product is equivalent to the presence of a soft flip in $(a, t)$.
Therefore, taking the $\log$ we obtain
\begin{align} \label{eq:lemma_proba_m_given_ideal}
\log \Prob(\softoutcome | \idealoutcome)
= 
C_2' + \sum_{e \in x_{\soft}} -\log L_{a, t}(m_{a, t}) \cdot
\end{align}
By definition of the edge weights in the decoding graph, combining \cref{eq:lemma_bayes,eq:lemma_proba_x,eq:lemma_proba_m_given_ideal} concludes the proof of the lemma.
\end{proof}

From \cref{thm:uf_success} and \cref{lemma:error_proba} we can infer the following useful corollary.

\begin{coro} \label{coro:soft_mwpm_low_weight_correction}
For any graphical model satisfying (C1) and (C2), the soft MWPM decoder (\cref{algo:soft_mwpm}) successfully corrects all fault sets such that $|\tilde x|_w < d_w/2$.
\end{coro}

\begin{proof}
Let $x$ be the fault set which occurs and let $m$ be the set of soft outcomes observed.
Assume that $|\tilde x|_w < d_w/2$.
By \cref{thm:uf_success}, the soft MWPM decoder returns a fault set $y$ such that
$
\Prob(y | m) \leq \Prob(x | m).
$
By \cref{lemma:error_proba}, this translates into
\begin{align}
\sum_{e \in \tilde y} w(e) \leq \sum_{e \in \tilde x} w(e)
\end{align}
which means $|\tilde y|_w \leq |\tilde x|_w$.
As a result, the residual fault set after correction $\tilde x + \tilde y$
has weight at most $|\tilde x + \tilde y| \leq 2|\tilde x|_w < d_w$ proving the decoder succeeds.
\end{proof}

\subsection{Success of the soft UF decoder}

Above we saw in \cref{coro:soft_mwpm_low_weight_correction} that the soft MWPM decoder successfully corrects all fault sets such that $|\tilde x|_w < d_w/2$.
Here we establish a similar sufficient condition for the soft UF decoder.

\begin{theo}\label{thm:uf_success}
For any graphical model satisfying (C1), (C2) and (C3), the soft UF decoder (\cref{algo:soft_uf}) successfully corrects any fault set such that $|\tilde x|_w < d_w/2$.
\end{theo}

\begin{proof}
Let $x$ be a fault set such that $|\tilde x|_w < d_w/2$.
Consider the graph $\tilde H_T(\gamma)$ represented in \cref{fig:uf_growth} induced by the grown edges at the end of the growth. 
For each edge $e = \{u, v\}$ of $\tilde H_T$ where $v$ is the splitting vertex, if $\gamma(e) = w(e)$, then $\tilde H_T(\gamma)$ contains the edge $e$, otherwise it contains an edge $\{u, u'\}$ with length $\gamma(e)$ which represents the grown part of the edge $e$.

Denote by $\Gamma_1, \Gamma_2, \dots$ the growth rates in step 6 of \cref{algo:soft_uf}.
The connected components $C_1, C_2, \dots$ of $\tilde H_T(\gamma)$ satisfy
\begin{align}
\sum_i \diam_w(C_i) \leq 2 \sum_j \Gamma_j, \\
\sum_{e \in \tilde x} \gamma(e) \geq \sum_j \Gamma_j.
\end{align}
The first equation holds because during the $j$th growth step, the selected cluster grows in all directions of $\Gamma_j$.
Moreover, because the growing cluster is odd, this growth step covers at least a fraction $\Gamma_j$ of an edge of $x$, leading to the second equation.
Combining the two equations above, we obtain
\begin{align} \label{eq:uf_success_proof_diam_bound}
\sum_i \diam_w(C_i) \leq 2 \sum_{e \in \tilde x} \gamma(e) \cdot
\end{align}

Let $\tilde y$ be the fault set returned by the peeling decoder in step 9 and let $\tilde z = \tilde x + \tilde y$.
Denote by $\tilde z_1, \tilde z_2, \dots$ the connected components of $\tilde z$.
By construction $\tilde y$ is included in $\tilde H_T(\gamma)$.
The only edges of $\tilde z$ outside of $\tilde H_T(\gamma)$ are the edges of $\tilde x$ such that $\gamma(e) = 0$.
As a result, we have
\begin{align}
\sum_k \diam_w(z_k) 
& \leq
	\sum_i \diam_w(C_i) 
 	+ \sum_{\substack{e \in \tilde x \\ \gamma(e) = 0}} w(e) \\
& \leq 
	2 \sum_{e \in \tilde x} \gamma(e)
	+ \sum_{\substack{e \in \tilde x \\ \gamma(e) = 0}} w(e) \\
& \leq 
	2 |\tilde x|_w < d_w \cdot
\end{align}
Therein, the second inequality is obtained by applying \cref{eq:uf_success_proof_diam_bound}.
Because the last round of measurement is assumed to be perfect and the decoder returns a fault set $\tilde y$ which has the same syndrome as $\tilde x$, the residual fault set after correction $\tilde z$ has trivial syndrome and we proved that it obeys $|\tilde z| < d_w$.
Therefore, applying (C3), the decoder succeeds.
\end{proof}

\section{Performance of soft decoders}
\label{sec:numerics}

\begin{figure*}[]
    \includegraphics[width=\textwidth]{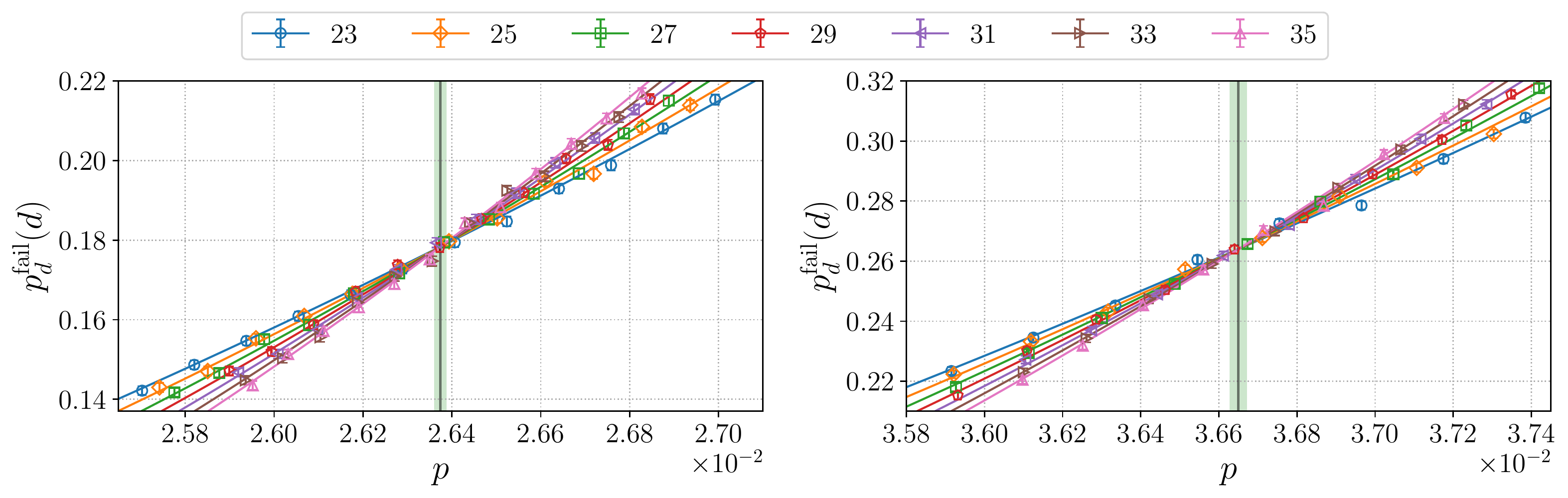}
    \caption{\label{fig:thresh_pheno}
    The failure probability of the protocol for $d$ rounds of error correction with soft phenomenological noise using the hard Union Find decoder (left) and the soft Union Find decoder (right).
    In this noise model the $X$ and $Z$ errors are identical and independent, with the $X$ error probability $p_D=p$, no hard ideal outcome flips $p_{M}=0$, and with Gaussian soft noise with a spread $\sigma$ selected such that $p_{M,\text{hardened}} = p$. 
	The threshold value $p^*$ for each decoder is estimated from the crossings of curves of different distance, and one standard deviation of uncertainty is shaded.
    }
\end{figure*} 

\begin{figure}[]
    \includegraphics[width=\columnwidth]{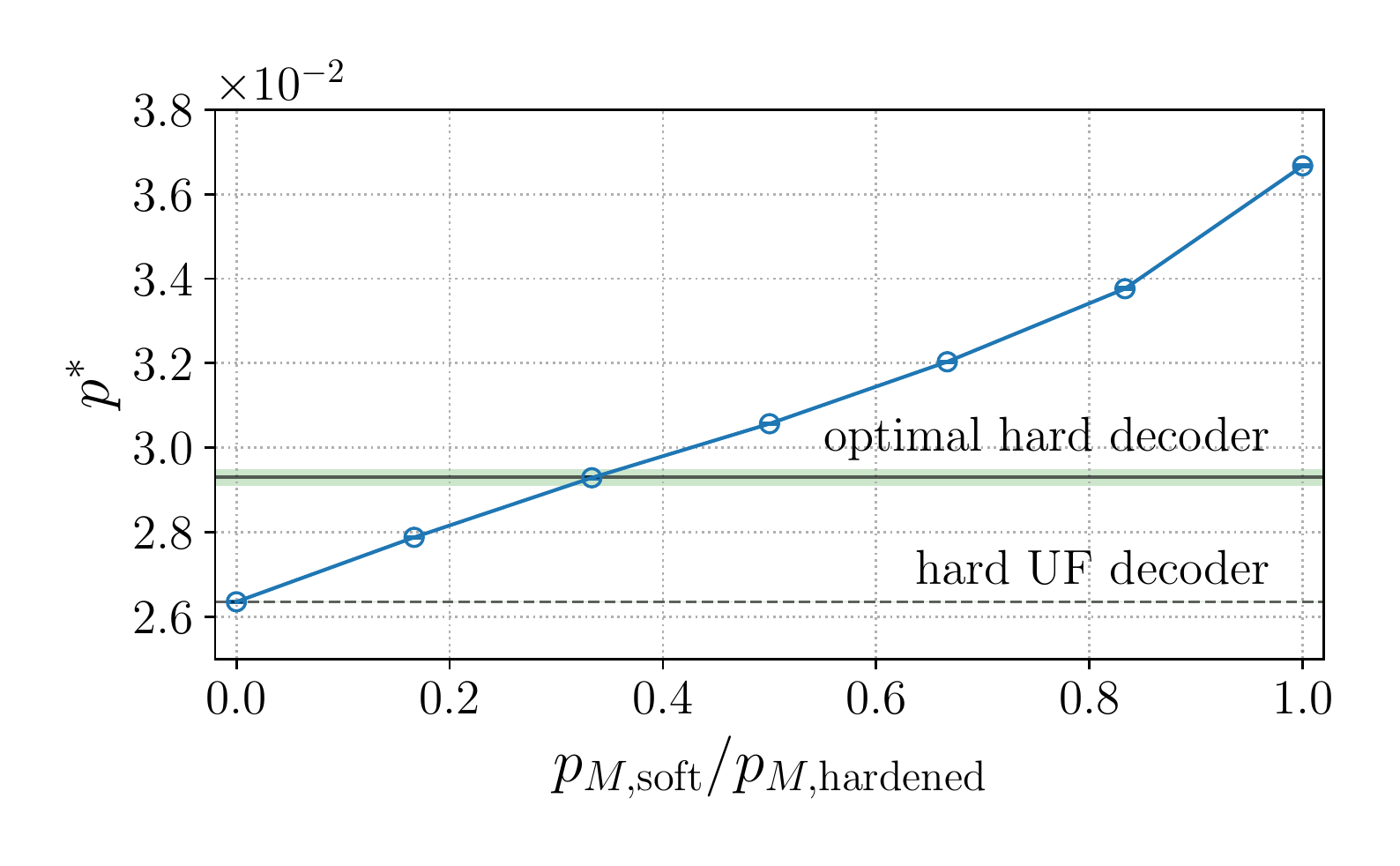}
    \caption{\label{fig:thresh_sweep_pheno}
	The threshold $p^*$ obtained with the soft decoder with soft phenomenological noise with different ratios of soft to hard noise in the measurement.
	In this noise model the $X$ and $Z$ errors are identical and independent, with the $X$ error probability $p_D=p$, with a variable ratio $r=p_{M,\text{soft}}/p_{M,\text{hardened}}$ but with the overall hardened flip probability $p_{M,\text{hardened}}=p$. 
	The threshold of the hard Union Find decoder is the same for all values of $r$ and is indicated with a dashed line.
	Note that $r=1$ corresponds to the setting of \cref{fig:thresh_pheno} in which there is no hard measurement flip, while $r=0$ corresponds to the only measurement noise being a hard flip.
    The optimal threshold with a hard decoder was estimated in Ref.~\cite{wang2003confinement_higgs} and is indicated with a solid horizontal line with numerical uncertainty indicated with green shading.
    }
\end{figure}

\begin{figure*}[]
    \includegraphics[width=\textwidth]{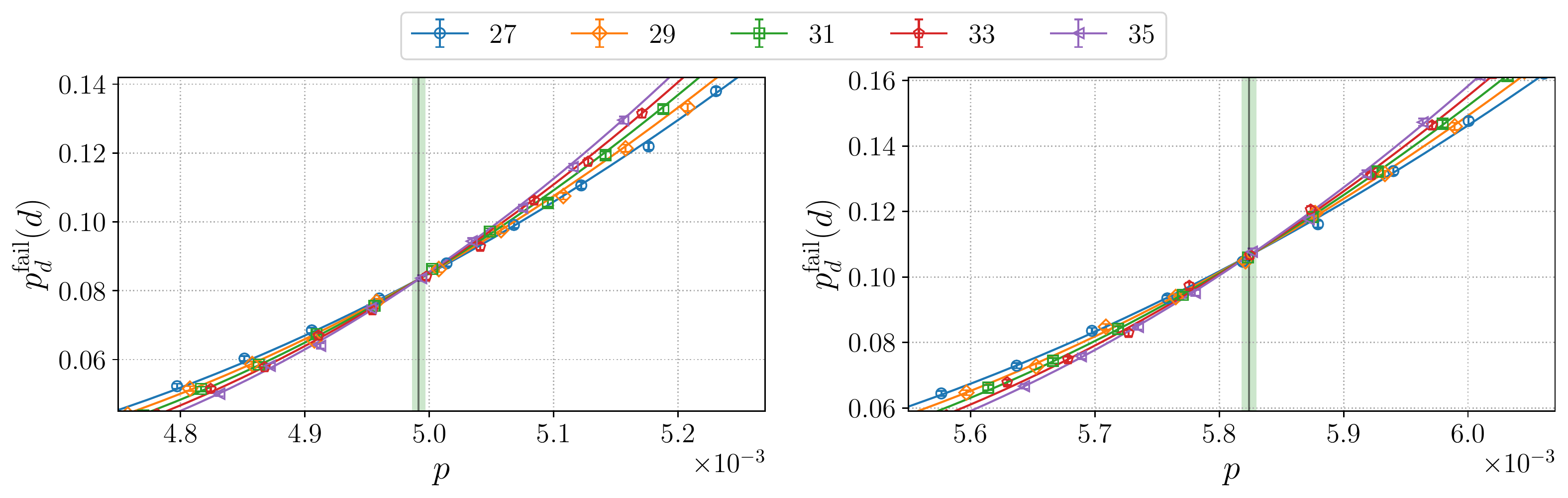}
    \caption{\label{fig:thresh_circuit_noise_transmon}
    The failure probability of the protocol for $d$ rounds of error correction with soft circuit noise using the hard Union Find decoder (left) and the soft Union Find decoder (right).
    We set $p_{IG}$, $p_{IM}$ and $p_{\text{CNOT}}$ equal to $p$, set $p_M = 0$ and use $\pdf^{(0)} = {\cal N}(+1, \sigma^2)$ and $\pdf^{(1)} = {\cal N}(-1, \sigma^2)$ with $\sigma$ chosen such that $p_{M,\text{hardened}}=10p$.
    	The threshold value $p^*$ for each decoder is estimated from the crossings of curves of different distance, and one standard deviation of uncertainty is shaded.
    }
\end{figure*}

In this section we analyze the performance of some of the decoders defined in \cref{sec:soft_decoding} under the soft noise models defined in \cref{sec:graphical_model}.
In particular we compare the soft Union Find decoder that we have introduced with the standard hard Union Find decoder that has access only to the hardened outcomes produced by discretizing the soft measurement outcomes.
In \cref{subsec:logical_error_rate} we first specify the protocol we use to evaluate the error correction performance, along with other definitions used in the numerical analysis.
Then in \cref{subsec:performance_phenomenological} and \cref{subsec:performance_circuit} we present our numerical results for the performance of the soft Union Find decoder under soft phenomenological and soft circuit noise respectively.

\subsection{Error correction protocol}
\label{subsec:logical_error_rate}

It can be desirable to define a notion of {\em logical error rate} as the probability that a logical operation fails given a specific noise model and decoding scheme.
In this section we describe how we estimate the logical error rate of a noise model equipped with a decoder using Monte Carlo simulations.
Our simulations use a protocol starting with a perfectly prepared code state which undergoes $T$ rounds of noisy syndrome measurements before a final round of perfect syndrome measurement, which specifies a graphical noise model $(G_T, p, f, \pi)$.

In a Monte Carlo run of this protocol, first a fault set $x$ is sampled according to the distribution $\Prob$ associated with the noise model.
For each fault $e$ in $x$, the corresponding residual error $\pi_e$ is applied to the data qubits.
Then, for each measurement vertex $(a, t)$ (including the final round of perfect measurement), the ideal outcome $\idealoutcome_{a,t}$ is computed and a soft outcome $\softoutcome_{a,t}$ is sampled according to the distribution $f_{a, t}^{(\idealoutcome_{a,t})}$.
Based on the set of soft outcomes, the decoder returns a fault set $y$ and a correction
\begin{align} \label{eqn:fault_set_residual_error}
\pi(y) = \prod_{e \in y} \pi_e
\end{align}
is applied to the data qubits intending to cancel the effect of $x$.
The residual error on the data qubits after correction is then
$\pi(x) + \pi(y)$.

Since the protocol ends with a perfect round of measurement, the decoders studied here are guaranteed to produce a correction with a residual error which is either a stabilizer or a non-trivial logical operator for the surface code.
We say the decoder {\em succeeds} if the residual error $\pi(x) + \pi(y)$ is a stabilizer.
Otherwise, we say the decoder {\em fails}.
With a distance-$d$ surface code, the failure probability of the protocol $p^{\text{fail}}_d(T)$ is the expected value that the protocol fails, and in practice we estimate it from the fraction of failures among a finite number of samples.
The decoders we study act separately on $X$ and $Z$ errors. 
We often find it convenient to separately consider the probability $p^{X \text{fail}}_d(T)$ of a logical $X$ failure of the protocol, irrespective of whether or not there is a logical $Z$ failure.

There are a number of approaches in the literature to estimate a logical error rate from the failure probability of this kind of protocol. 
The simplest and most standard approach is to simply set $T=d$ and take $p^{\text{fail}}_d(d)$ itself to be an estimate of the logical error rate. 
The reason for setting $T = d$ in this approach is that typically logical error operations (such as lattice surgery) with the surface code are assumed to take $d$ rounds to ensure fault tolerance. 
However, the fact that the protocol ends with a perfect measurement round can lead this approach to significantly underestimate the true logical failure rate of a logical operation which takes $d$ rounds~\cite{chao2020optimization}. 
Another drawback of this approach is that some logical operations may not take $d$ rounds.

Here we instead estimate the logical failure rate per round, which we denote by $\bar{p}_d$, which can then be multiplied by the number of rounds in a logical operation to indicate its logical error rate.
To avoid dependence on the perfect preparation and final stabilizer measurement round in the protocol on our estimate, we use a model of the protocol which applies for long times $T$.
Our model is based on the assumption that during each round a logical $X$ error occurs with probability $\bar{p}/2$ independently of other rounds~\cite{fowler2021repetition}. 
After $T$ rounds, there is a logical $X$ failure in this model if an odd number of rounds have experienced logical $X$ errors, which occurs with probability
\begin{align}
\hat{p}^{X~\text{fail}}(T) = \frac{1 - (1-\bar{p})^T}{2},\label{eq:failure_rate_per_round}
\end{align}
which can be inverted to give $\bar{p} = 1-(1 -2\hat{p}^{X~\text{fail}}(T))^{1/T}$.
We use a hat to differentiate the failure probability $\hat{p}^{X~\text{fail}}(T)$ of our model from the failure probability $p^{X~\text{fail}}(T)$ of the protocol itself. 
Assuming the validity of this model becomes accurate for large $T$, we take the logical failure rate per round $\bar{p}_d$ to be calculated from the protocol's failure probability according to
\begin{align}
\bar{p}_d(T)&= 1-(1 -2p_d^{X~\text{fail}}(T))^{1/T},\label{eq:failure_rate_per_round_data_T}\\
\bar{p}_d&= \lim_{T \to \infty}\left( \bar{p}_d(T) \right).\label{eq:failure_rate_per_round_data}
\end{align}
In practice we estimate $\bar{p}_d$ by evaluating $\bar{p}_d(T)$ for a large value of $T$.
More details on this model's validity, and on the numerical convergence of \cref{eq:failure_rate_per_round_data} can be found in \cref{app:bc_indep}.

The {\em error correction threshold} is the value of a single-parameter noise model below which error correction can be performed to achieve arbitrarily good protection by increasing the code distance. 
For noise models which have a number of parameters such as the soft noise models we have introduced, the concept of a threshold generalizes to a threshold surface~\cite{svore2007architecture2d}.
Alternatively, one can introduce constraints on the parameters of a nosie model so that only one free parameter remains, and then consider the threshold value of the resulting single-parameter noise model.
There are many approaches to numerically estimate the value of the threshold, all of which can be expected to agree in the limit of large code distances~\cite{beverland2021codeswitching} because they are asymptotic properties of a phase boundary~\cite{dennis2002tqm}.
Here we estimate the threshold $p^*$ of single-parameter noise model using the standard approach of identifying the intersections of the curves $p^{\text{fail}}(d)$ for different distances as a function of the parameter $p$.

\subsection{Performance under soft phenomenological noise}
\label{subsec:performance_phenomenological}

Here we consider the soft phenomenological noise model which was defined in \cref{subsec:phenom_model} and use the notation therein.
In this model, the $X$ and $Z$ errors are generated and corrected independently. 
For simplicity we focus on the $X$ errors here -- the $Z$ errors are handled in the same way.

Recall that in this model, probabilities $p_D$ and $p_{M}$ are assigned to horizontal and vertical edges in the $X$-type fault graph $G_T$ respectively, and the pair of Gaussian probability density functions $\pdf^{(0)} = {\cal N}(+1, \sigma^2)$ and $\pdf^{(1)} = {\cal N}(-1, \sigma^2)$ are assigned to each measurement vertex in $G_T$.

In \cref{fig:thresh_pheno} we obtain a single-parameter noise model from the phenomenological noise model by setting $p_D=p$, $p_{M}=0$ and setting $\sigma$ such that $p_{M,\text{hardened}} = p$, where $p_{M,\text{hardened}}$ is calculated according to \cref{eq:net_hardened_flip_prob}. 
We plot the failure probability of the protocol for $d$ rounds using the soft UF decoder and the hard UF decoder, and estimate the threshold $p^*$ by considering the crossings of the resulting curves for different values of $d$.
Note that the failure probabilities $p^{\text{fail}}_d(T)$ reported include failures of both $X$ and $Z$ type.
We observe that the soft UF decoder significantly outperforms the hard UF decoder, exhibiting a threshold of \num{3.665(2)e-2} versus \num{2.637(1)e-2}. 
Moreover, the optimal threshold of a decoder using hard information only is \num{2.93(2)e-2}~\cite{wang2003confinement_higgs} which is still below the value obtained using the soft UF decoder.

In \cref{fig:thresh_sweep_pheno} we repeat this analysis again with $p_D=p$ and $p_{M,\text{hardened}}=p$, but with different fixed values of $r=p_{M,\text{soft}}/p_{M,\text{hardened}}$. 
For each value of $r$, this specifies a single-parameter noise model for which we can estimate the threshold using the same approach as was used in \cref{fig:thresh_pheno} for the case $r=1$. 
We plot these estimated threshold values for various values of $r$ in \cref{fig:thresh_sweep_pheno}.

Note that by construction the hard noise models which are obtained by hardening the soft phenomenological noise models we study in \cref{fig:thresh_pheno} and \cref{fig:thresh_sweep_pheno} reduce to the standard single-parameter phenomenological noise model.
In particular, the hardened version of the soft noise model for all values of $r$ has $X$ error on each qubit with probability $p_D=p$ and outcome flips of each plaquette measurement with the same probability $p_{M,\text{hardened}}=p$. 

\subsection{Performance under soft circuit noise}
\label{subsec:performance_circuit}

\begin{figure}[]
    \includegraphics[width=\columnwidth]{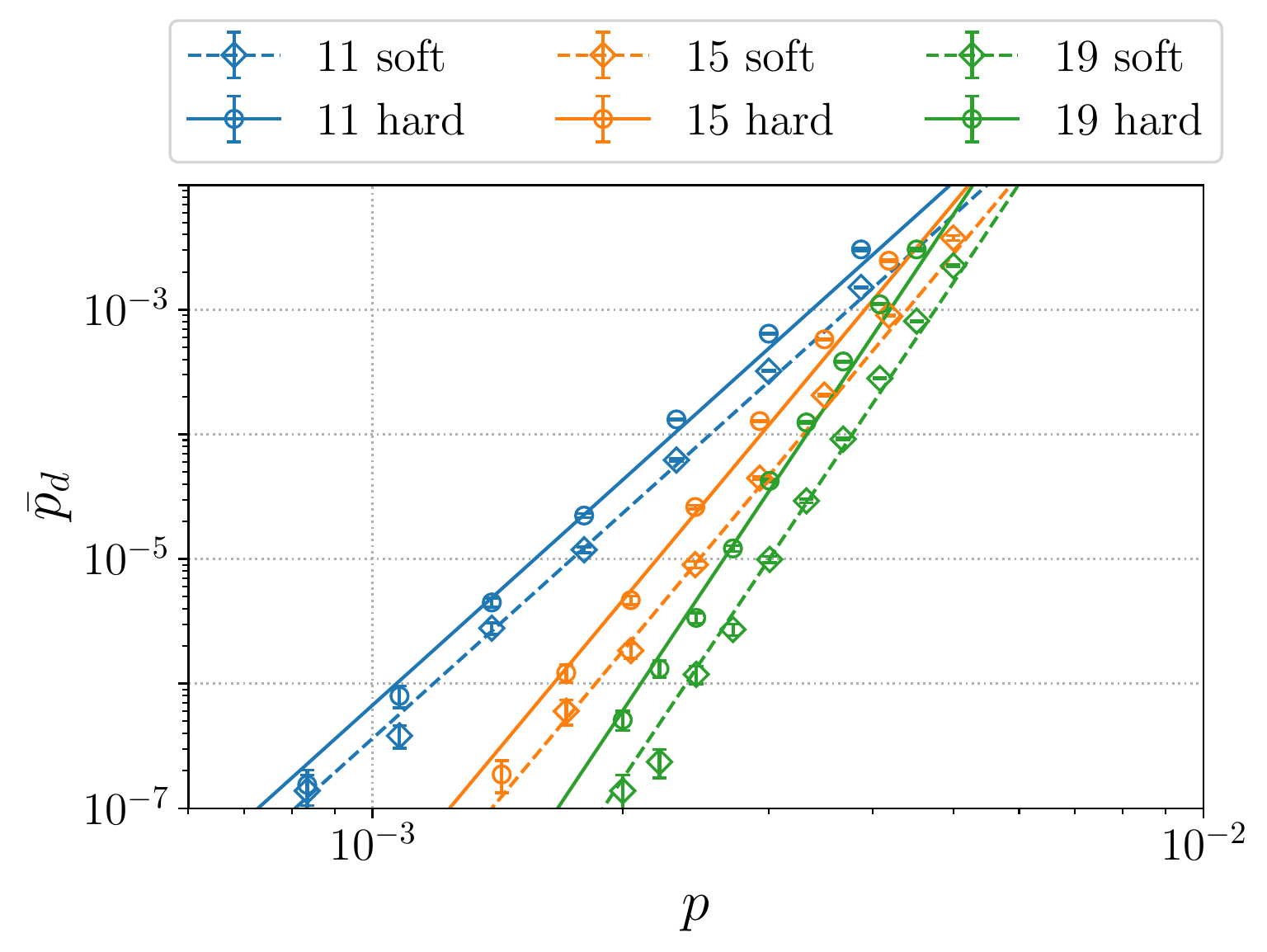}
    \caption{\label{fig:low_p_circuit_noise}
The failure probability per round for low error rates with the same soft circuit noise parameterization as in \cref{fig:thresh_circuit_noise_transmon}.
    	 We estimate the failure rate per round $\bar{p}_d$ from the failure probability of the protocol for $T=1000$ rounds using \cref{eq:failure_rate_per_round_data}. 
     }
\end{figure}

Here we consider the soft circuit noise model which was defined in \cref{subsec:circuit_model}.
Recall that in this model, probabilities $p_{IG}$, $p_{IM}$ and $p_{\text{CNOT}}$ specify the probabilities of each possible fault of idle qubits during CNOT time steps, idle qubits during measurement time steps, and faults of CNOT operations.
Additionally, the probability $p_M$ specifies the flip of ideal outcomes, and probability distribution functions $\{\pdf^{(0)}(\mu), \pdf^{(1)}(\mu) \}$ characterize the soft noise.
For the numerical studies here we use the standard exclusive circuit noise model to generate the faults (see \cref{app:inclusive_exclusive}), while we use the inclusive circuit noise model is \cref{subsec:circuit_model} to specify the edge weights of the decoder as described in \cref{subsec:decoding_graph}.
To improve the performance of the decoder, we merge all edges sharing the same endpoints in the decoding graph and combine their edge weights using the formula given in Appendix E of \cite{chao2020optimization}.

In \cref{fig:thresh_circuit_noise_transmon} we obtain a single-parameter noise model from the soft circuit noise model by setting $p_{IG}=p$, $p_{IM}=p$, $p_{\text{CNOT}} = p$, $p_M = 0$ and setting $\pdf^{(0)} = {\cal N}(+1, \sigma^2)$ and $\pdf^{(1)} = {\cal N}(-1, \sigma^2)$ with $\sigma$ chosen such that $p_{M,\text{hardened}}=10p$. 
We plot the failure probability of the protocol for $d$ rounds using the soft UF decoder and the hard UF decoder, and estimate the threshold value of $p$ by considering the crossings of the resulting curves for different values of $d$.
We observe that the soft UF decoder outperforms the hard UF decoder, exhibiting a threshold of \num{5.824(6)e-3} versus \num{4.991(5)e-3}. 
In \cref{fig:low_p_circuit_noise} we estimate the failure rate per round $\bar{p}_d$ using the soft UF decoder and the hard UF decoder with \cref{eq:failure_rate_per_round_data} for $T=1000$ rounds.

For many hardware approaches, the readout fidelity can be significantly lower than the gate fidelity, which motivated us to consider a measurement flip probability $p_{M,\text{hardened}}$ which was ten times larger than the gate error probability in \cref{fig:thresh_circuit_noise_transmon}.
In \cref{tab:result_table} we list a number of threshold values including those for soft circuit noise but with $p_{M,\text{hardened}}$ equal to the gate error probability, in which case the threshold gap for the hard and soft decoder is less pronounced due to the relatively larger effects of ancilla qubit errors during the syndrome extraction circuit.

\begin{table}[h]
\centering
{
    \newcolumntype{A}{>{\centering\arraybackslash}m{7.8em}} 
    \newcolumntype{B}{>{\centering\arraybackslash}m{5.5em}}
    \renewcommand\arraystretch{1.2}
    \begin{tabular}[t]{| A | B B |} 
        \hline
        soft noise                                   &  soft UF threshold                     & hard UF threshold                       \\\hline\hline
        \vspace{0.15em}phenomenological              & \vspace{0.15em}\SI{3.665(2)}{\percent} & \vspace{0.15em}\SI{2.637(1)}{\percent}  \\[0.5ex]\hline 
        circuit \mbox{$p_{M,\text{hardened}} = p$}   &  \SI{0.727(1)~}{\percent}              &  \SI{0.702(1)~}{\percent}               \\[2ex]\hline   
        circuit \mbox{$p_{M,\text{hardened}} = 10p$} &  \SI{0.5824(6)}{\percent}              &  \SI{0.4991(5)}{\percent}               \\[2ex]         
        \hline
    \end{tabular}
}
\caption{\label{tab:result_table}
    Thresholds observed when using the soft and hard decoders for various soft noise models.
    For soft phenomenological noise we set $p_D=p$, $p_{M}=0$ and $\sigma$ is chosen such that $p_{M,\text{hardened}} = p$.
    We consider two cases for soft circuit noise,  each with $p_{IG} = p_{IM} = p_{\text{CNOT}} =p$ and $p_M = 0$, but with different $p_{M,\text{hardened}}$.
}
\end{table}

\section{Measurement time tradeoff}
\label{sec:measure_time_tradeoff}

\begin{figure}
    \includegraphics[width=\columnwidth]{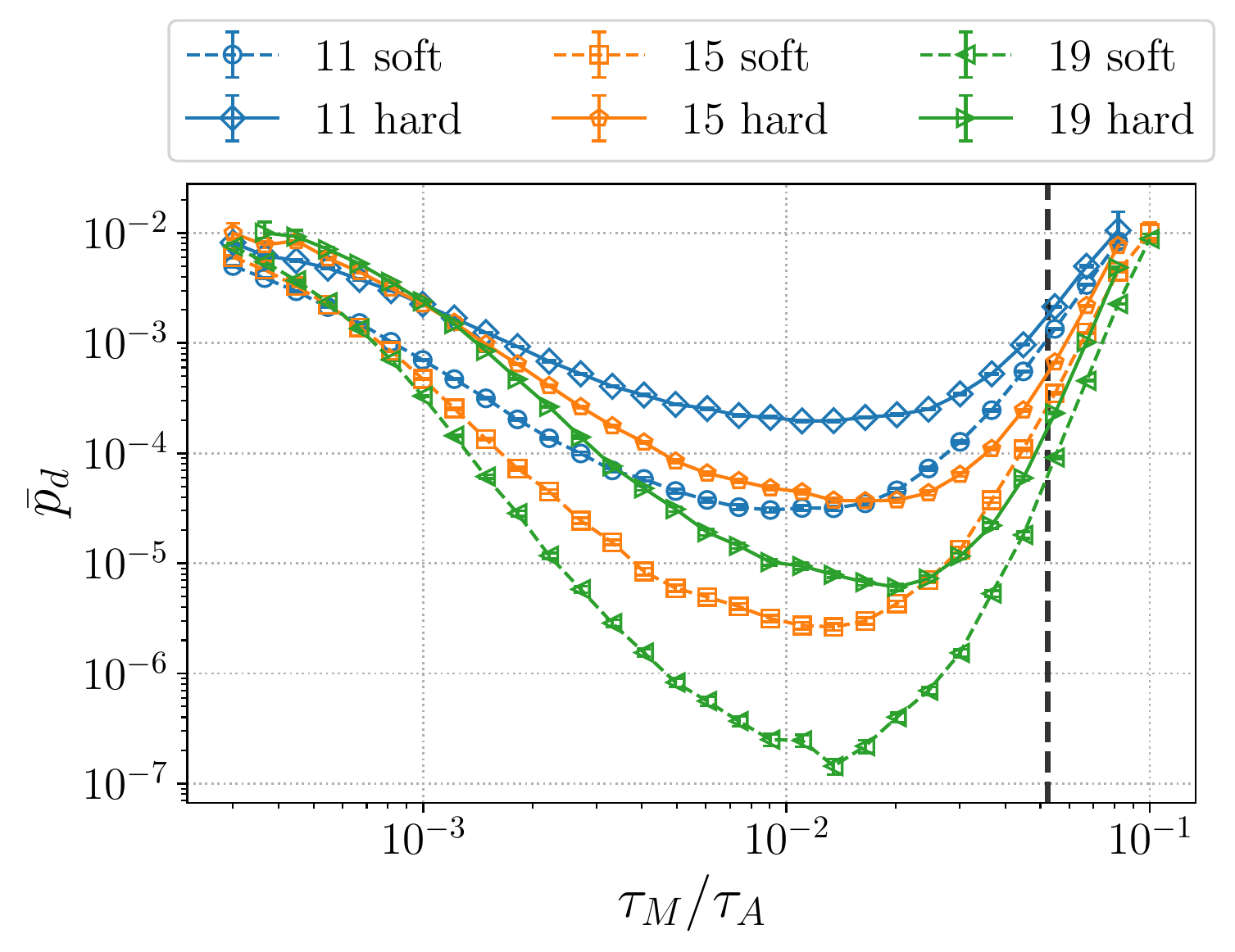}
    \caption{\label{fig:meas_tradeoff} The logical error rate per
      round $\bar{p}_d$ for the parametric circuit noise model defined
      in section \cref{sec:measure_time_tradeoff} as a function of
      measurement time $\tau_M$.  The $\tau_M$ which minimizes
      $\bar{p}_d$ depends on the distance and decoder.  The vertical
      dashed line marks the measurement time that minimizes the
      average soft-flip probability of a single measurement with respect to measurement time
      (instead of the logical error rate), and leads to
      sub-optimal logical error rates.
    }
\end{figure}

In this section, we establish that the use of soft information can provide an advantage in the overall logical cycle time of a quantum computing platform in a model that considers the effects of measurement time.
In real quantum error correction applications, selecting a measurement time for measurements represents a trade-off between measurement noise and data noise:
Using too short a measurement time gives little information about the state of the ancilla.
However, too long a measurement time gives time for additional errors to accumulate on the data.

We define an experimentally motivated variant of the soft circuit noise model where the noise parameter for each fault location is adjusted based on the duration that each circuit operation takes.
Recall that in the soft circuit noise model, probabilities $p_{IG}$, $p_{IM}$ and $p_{\text{CNOT}}$ specify the probabilities of each possible fault of idle qubits during CNOT time steps, idle qubits during measurement time steps, and faults of CNOT operations.
Additionally, $p_M$ specifies the probability of a flip of an ideal outcome, and the probability distribution functions $\{\pdf^{(0)}(\mu), \pdf^{(1)}(\mu) \}$ characterize the soft noise.
In our {\em parametric circuit noise model}, we consider:
\begin{itemize}
\item $p_{IG} =  1 - e^{-\frac{\tau_G}{\tau_D}}$,
\item $p_{IM} = 1 - e^{-\frac{\tau_M}{\tau_D}} $,
\item $p_{\text{CNOT}} =  1 - e^{-\frac{\tau_G}{\tau_D}}$,
\item $p_M=0$ and $\pdf^{(\bar{\mu})}(\mu) = \pdfAD^{(\bar{\mu})}(\mu;\tau_M,\T1m,\tau_F)$.
\end{itemize}

This noise model assumes depolarizing noise, with additional Gaussian soft noise with amplitude damping during measurement. 
As such, our {\em parametric circuit noise model} is expressed in terms of the depolarizing time $\tau_D$, the gate and measurement times $\tau_G$ and $\tau_M$, along with additional parameters which specify the amplitude damping channel including the relaxation time $\T1m$ and the fluctuation time $\tau_F$ as described in \cref{subsec:amplitude_damping_model} and \cref{app:toy-model-measurement}.
As in \cref{subsec:performance_circuit} we use the standard exclusive circuit noise model to generate the faults in our numerical simulations here, while we use the inclusive noise model to specify the edge weights of the decoder.

In \cref{fig:meas_tradeoff} we explore the effect of measurement time on the failure probability.
In this plot, we set $\tau_G =\SI{10}{\nano\second}$, $\tau_D = \SI{30}{\micro\second}$, $\T1m=\SI{15}{\micro\second}$, $\tau_F =\SI{100}{\nano\second}$ and vary $\tau_M$.
These parameters are within the range that can be expected to be achievable using near-term experiments~\cite{arute2019supremacy,nagirneac2019fastcz,crippa99spinsgatesensing,jurcevic2021volume64,xue2021spins99}. 
We plot the failure rate per round $\bar{p}_d$ using the soft UF decoder and the hard UF decoder, by using \cref{eq:failure_rate_per_round_data} for $T=1000$ rounds.

In contrast, we may choose $\tau_M$ to minimize the average soft flip
probability $\frac{1}{2}\Prob(\text{soft flip} | 0) +
\frac{1}{2}\Prob(\text{soft flip} | 1)$ for a single measurement under
the Gaussian soft noise with amplitude damping model (a metric often
considered in experimental demonstrations). This ignores the errors
that accumulate in the data qubits during the measurement of the
ancilla, but minimizes the probability of errors in the hard decision
for each measurement (for details see
\cref{app:toy-model-measurement}).

For a long measurement time $\tau_M$, the data qubit noise rate per
round is high and as a result the logical failure probability
approaches that of a maximally mixed state.  If the measurement time
is too short, the information obtained by the syndrome extraction
circuit is not reliable enough, which translates into a high logical
error rate.  In order to obtain the best possible logical qubit, we
should select the measurement time that minimizes the logical error
rate.  In \cref{fig:meas_tradeoff}, we observe that the optimal
measurement time using the soft decoder is reached around
$\tau_M/\tau_A = \num{1e-2}$ with a logical error rate per round close
to $\num{1e-7}$ for distance 19.  For comparison, if the measurement
time is selected by minimizing the average soft-flip probability, the
measurement time is roughly five times longer and, for distance 19,
the logical error rate per round increases by nearly three orders of
magnitude to $\num{1e-4}$.  In these numerical results, we also
observe that the optimal measurement time for the soft UF decoder is
shorter than the optimal measurement for the hard UF decoder, which
indicates that metric to optimize is not simply a function of the soft
flip probabilities and the error rates on data qubits (as these are
the same for the hard and soft decoders).

\section{Conclusion}

We have described a general framework that incorporates soft
information into fault-tolerant quantum error correction with the
surface code.  This framework is independent of the physical
realization of the qubits, and requires only knowledge of the
distributions of soft outcomes conditioned on the qubit state.  We
give explicit constructions of two soft decoders for the surface code:
a soft Minimum Weight Perfect Matching decoder (which we show
identifies a most likely fault set for soft information), and a soft
Union-Find decoder.  Both of these decoders have the same
computational complexity as their counterparts that handle only hard
information, while the soft Union-Find decoder has lower computational
overhead than the soft Minimum Weight Perfect Matching decoder.

We introduce two error models to evaluate the decoder performance: one
where the syndrome is corrupted by Gaussian noise, and another
where the syndrome is corrupted by Gaussian noise and amplitude
damping. We show that under these error models the soft Union-Find
decoder outperforms its hard counterpart both in terms of error
threshold as well as sub-threshold logical error rate. Notably, the
soft Union-Find decoder has a threshold that is $25\%$ larger than the
best possible hard decoder for the surface code, as estimated by
statistical mechanics models.  These error models were also used to
study the trade-off between measurement time and logical error rate,
where we found a strong dependence between these quantities, and
showed that optimizing the measurement time with awareness of error
correction can reduce the measurement time five-fold while decreasing
the the logical error rate a thousand fold (even for relatively modest
surface code distances).  This highlights the benefits of jointly
optimizing the physical layer and the error correction layer.

It would be valuable to develop theoretical methods to study the
optimal decoding threshold for soft noise in the surface code---for
example, by searching for a mapping to a generalized classical
random-bond Ising model in which subsets of binary spins are be
replaced by continuous variables.

Beyond surface codes, it would be interesting to explore these soft
decoders with other codes such as color
codes~\cite{bombin2006topological}, other surface code
variants~\cite{ataides2021xzzx} and some quantum Low-Density
Parity-Check codes~\cite{delfosse2021toward}.

Given the gains observed with the use of soft
information in the surface code, we expect that the addition of soft
information to the quantum fault tolerance arsenal will shorten the
time horizon for useful scalable quantum computers.

\section{Acknowledgment}

M.E.B, M.P.S, and N.D. thank David Poulin for his friendship,
inspiration, and mentorship, as well as encouragement in the early
stages of this work. C.A.P. acknowledges helpful discussions with Evan
Zalys-Geller and  Microsoft Quantum for support during his internship.
C.A.P. is supported in part by Air Force Office of Scientific Research (AFOSR), FA9550-19-1-0360.

\newpage

\appendix

\section{Gaussian soft noise with amplitude damping}
\label{app:toy-model-measurement}

We consider one of the simple measurement models described in
Ref.~\cite{Gambetta2007measurement}, namely, the measurement of a
qubit with finite $\ket{1}$ lifetime $\tau_A$ and finite signal-to-noise
ratio. This model was originally proposed to describe the dispersive
measurement of superconducting qubits, but it applies to many other
solid-state qubit implementations such as quantum dots. We present the
toy model here in some detail for completeness.

In our toy model we assume the Hilbert space of the quantum system is
spanned by two orthogonal states $\ket{0}$ and $\ket{1}$. Once the
measurement is turned on at time $\tau=0$, the (classical) continuous
time output signal of the measurement apparatus is given by
\begin{align}
  \dif S(\tau) = v~R(\tau)~\dif \tau + \xi~\dif W(\tau)
\end{align}
where $v$ is the signal amplitude, $R(\tau)$ is the measurement response,
$\xi$ is the noise amplitude, and $\dif W(\tau)$ is the Wiener increment (a
white Gaussian noise process with variance $\dif \tau$ and zero mean).

We assume the state of the system collapses instantaneously into one
of the basis states once measurement starts at $\tau=0$, and that $R(\tau)$
depends on the state of the system instantaneously, such that
$R(\tau)=1$ if the system is in state $\ket{0}$, and $R(\tau)=-1$ if the
system is in state $\ket{1}$.

The instantaneous response is important because the quantum system can
spontaneously decay at any moment from the state $\ket{1}$ to the
state $\ket{0}$ (but will remain in the $\ket{0}$ indefinitely). The
decay time $K$ is an exponentially distributed random
variable with rate $1/\tau_A$. 
If we compute the probability that the
system is still in state $\ket{1}$ at time $\tau$, we find that ${\mathbb P}(\ket{1};
\tau) = e^{-\tau/\tau_A}$ as expected from the physics of spontaneous decay.

Throughout this appendix we will use capital letters to denote random
variables, and lower case letters otherwise. The probability density
function for a random variable $H$ will be denoted by $f_H(h)$, while
the corresponding cumulative density function will be denoted by
$F_H(h)$, where $F_H(h) = \int_{-\infty}^{h}~\dif {h'}~f_H(h') =
{\mathbb P}(H<h)$.  Following the convention in the main body of the
text, superscripted random variables indicate conditioning on the
quantum basis state, while conditioning on other random variables is
indicated with the $|$ symbol.

\subsection{Soft-outcome distributions}

In order to determine what was the outcome of the quantum measurement,
we inspect the measurement signal and aim to classify the observed
signal as corresponding to a $\ket{0}$ or a $\ket{1}$ outcome. One
approach is to simply integrate the measurement signal for some
measurement time $\tau_M$~\footnote{The signal may be post-processed
  in other
  ways~\cite{Gambetta2007measurement,gambetta2008measurement2}, but we
  focus on direct integration for simplicity.}. The resulting scalar
is given by the random variable
\begin{align}
  S = \int_0^{\tau_M}~\dif S(\tau), 
\end{align}
and we are free to choose $\tau_M$ in order to improve performance
({\it i.e.}, our ability to distinguish between $S$ when the state is
$\ket{0}$ and when the state is $\ket{1}$). We note that each
realization of $S$ corresponds to the soft outcome $\mu$.

Then, we have that
\begin{align}
  S
    & = \int_0^{\tau_M}~\dif S(\tau), \\
    & = \underbrace{\int_0^{\tau_M}\dif \tau~v~R(\tau)}_{P} + \underbrace{\int_0^{\tau_M}\dif W(\tau)~n}_{Q}.
\end{align}
It should be apparent that $Q \sim {\mathcal N}(0, \xi^2~\tau_M)$ regardless of
the qubit state, while the distribution for $P$ requires some
additional consideration conditioned on the initial state of the
qubit.

If the initial state is $\ket{0}$, we have that $P_0 = v~\tau_M$, so that
$S_0 \sim {\mathcal N}(v~\tau_M, \xi^2~\tau_M)$.

If the initial state is $\ket{1}$, then the state can decay at some
random time $K$.  We have that
\begin{align}
  P_1|K =
  \left\{
  \begin{array}{ll}
    - v ( 2K - \tau_M ), & K<\tau_M\\
    - v \tau_M, & K\ge \tau_M    
  \end{array}
  \right.
\end{align}
and where $K$ has an exponential distribution with mean $\frac{1}{\tau_A}$.

In order to compute the distribution for $P_1$, we first consider the
case where $\tau_M\to\infty$, so that the output of the integral is a
function of the random variable $K$, {\it i.e.},
\begin{align}
g(K) &= v(\tau_M - 2K)\\
g^{-1}(P) &= \frac{\tau_M}{2} - \frac{P}{2v}
\end{align}
so that
\begin{align}
\frac{\dif}{\dif P}g^{-1} = -\frac{1}{2v}
\end{align}
and thus the distribution for $g(K)$ is
\begin{align}
f_{g(K)}(p) = \frac{1}{2v} f_K\left( \frac{\tau_M}{2} - \frac{p}{2v} \right)  
\end{align}
where $f_K$ is the distribution for the decay time (an exponential
distribution).

Now, $g(K)$ differs from $P_1$ because $\tau_M$ is finite, so we
condition the transformation on the decay time. In particular,
\begin{align}
  f_{P_1|K<\tau_M}(p) &= \frac{f_{g(K)}(p)}{F_{K}(\tau_M)},\\
  f_{P_1|K\ge \tau_M}(p) &= \delta(p + v \tau_M),
\end{align}
and if we combine these we obtain
\begin{multline}
  f_{P_1}(p)
  =
  \Theta( p + v ~ \tau_M )~\frac{1}{2v}~f_{K}\left(\frac{\tau_M}{2} - \frac{p}{2v}\right) +\\
  \delta( p + v ~ \tau_M )~\left[1-F_{K}(\tau_M)\right]
\end{multline}
where $\delta$ is the Dirac function, and $\Theta(\tau)$ is the Heaviside
step function.

When $\tau_M\ll \tau_A$, decay events are unlikely, and thus $S_1$ is
approximately ${\mathcal N}(-v~\tau_M, \xi^2~\tau_M)$. When $\tau_M\gg \tau_A$, a
decay is almost certain to happen during the measurement, and $S_1$ is
approximately the sum of a $Q\sim{\mathcal N}(0, \xi^2~\tau_M)$ and an
a linear function of a exponential random variable with mean $\frac{1}{\tau_A}$.

We can obtain the probability density function for $S_1$ by convolving
$f_{P_1}$ with $f_{Q}$, resulting in
\begin{multline}
  \pdf_{S_1}(s) =
  \frac{
    e^{-\frac{\tau_M}{\tau_A}-\frac{(v \tau_M+s)^2}{2 \xi^2 \tau_M}}
  }{
    4 \xi v \tau_A \sqrt{2 \pi \tau_M}
  }\cdot\\
  \left\{
  4 v \tau_A +
  e^\frac{\left[\xi^2 \tau_M + 2 v \tau_A(v \tau_M + s)\right]^2}{8 \xi^2 v^2 \tau_A^2 \tau_M}
  \xi \sqrt{2 \pi \tau_M}\times\right.\\
  \left.
  \left[
    \mathrm{erf}\left(
    \frac{\xi^2 \tau_M + 2 v \tau_A (s + v \tau_M)}{2 \xi v \tau_A \sqrt{2 \tau_M}}
    \right)
    \right.
    \right.
    -\\
    \left.
    \left.
    \mathrm{erf}\left(
    \frac{\xi^2 \tau_M + 2 v \tau_A (s - v \tau_M)}{2 \xi v \tau_A \sqrt{2 \tau_M}}    
    \right)
  \right]
  \right\},
\end{multline}
so that we have analytical expressions for the conditional
distributions of the soft outcomes.

In the main text we use a slightly different convention for the soft
outcomes, where the conditional means are $\pm 1$ in the absence of
decays. This can be achieved by a simple change of variables, so that
\begin{align}
  \pdfAD^{(0)}(\mu) & = v~\tau_M~\pdf_{S_0}\left(v~\tau_M~\mu\right),\\
  \pdfAD^{(1)}(\mu) & = v~\tau_M~\pdf_{S_1}\left(v~\tau_M~\mu\right).
\end{align}
We can obtain a more convenient parameterization by defining
{\em the fluctuation time}
$\tau_F=\frac{\xi^2}{v^2}$, which allows us to write the conditional
distributions in terms of the dimensionless ratios $\frac{\tau_M}{\tau_F}$
and $\frac{\tau_M}{\tau_A}$, resulting in
\begin{align}
  \pdfAD^{(0)}(\mu;\tau_M,\T1m,\tau_F) &= \sqrt{\frac{\tau_M}{2\pi\tau_F}} e^{-\frac{(\mu-1)^2\tau_M}{2\tau_F}}
\end{align}
\begin{multline}
  \pdfAD^{(1)}(\mu;\tau_M,\T1m,\tau_F) =\\
  \sqrt{\frac{\tau_M}{2\pi\tau_F}}
  e^{-\frac{{(\mu+1)}^2\tau_M}{2\tau_F}-\frac{\tau_M}{\tau_A}}
  -\frac{1}{4}\frac{\tau_M}{\tau_A}e^{\frac{\tau_M\tau_F}{8\tau_A^2}+\frac{1}{2}\frac{\tau_M}{\tau_A^2}(\mu-1)}\times\\
  \left\{
  {\rm erf}\left[
    \sqrt{\frac{\tau_M\tau_F}{8\tau_A^2}}+(\mu-1)\sqrt{\frac{\tau_M}{2\tau_F}}
    \right]
  \right.\\
  \left.
  -
  {\rm erf}\left[
    \sqrt{\frac{\tau_M\tau_F}{8\tau_A^2}}+(\mu+1)\sqrt{\frac{\tau_M}{2\tau_F}}
    \right]
  \right\}
\end{multline}
It is natural to take $\tau_A$, $\tau_F$ to be fixed, and vary
$\tau_M$ to optimize the measurement. The justification for this is
that $\tau_A$ is set by qubit properties (materials, design, and
fabrication), while $\tau_F$ is set by amplifier properties (noise
temperature, compression point, etc.) and other device properties
(maximum tolerable signal power on the qubit, coupling strength
between qubit and measurement apparatus, etc.)---all of which are
fixed in an experiment.

The rough intuition is that $\T1m$ sets an upper limit to $\tau_M$. As
long as $\tau_M\ll\T1m$, increasing $\gamma_M$ decreases the soft-flip
error rate. As $\tau_M$ approaches and surpasses $\T1m$, decay events
cause the soft-flip error rate to increase. Similarly, $\tau_F$ sets
the lower limit on $\tau_M$. If $\tau_M$ is comparable to $\tau_F$,
the soft measurement distributions are dominated by Gaussian noise. As
$\frac{\tau_M}{\tau_F}$ increases above 1, the distributions become
distinguishable, as the overlap decays exponentially with
$\frac{\tau_M}{\tau_F}$ (for $\tau_M\ll\T1m$). A low soft flip error
rate can be achieved if $\tau_F\ll\tau_M\ll\tau_A$.

\subsection{Soft-flip probabilities}

We can convert this to conditional probabilities of soft flips by
integrating, which yields
\begin{multline}
  {\mathbb P}({\rm soft~flip}|0) =
  \frac{1}{2}\left[1 + {\rm erf}\left(\frac{\beta}{\sqrt{2}}-\sqrt{\frac{\tau_M}{2\tau_F}}\right)\right],
\end{multline}
\begin{multline}
  {\mathbb P}({\rm soft~flip}|1) =
  \frac{1}{2}-
  \frac{1}{2}{\rm erf}\left(\frac{\beta}{\sqrt{2}}-\sqrt{\frac{\tau_M}{2\tau_F}}\right)+\\
  \frac{1}{2}{\rm exp}\left[
    \beta\sqrt{\frac{\tau_M\tau_F}{4\tau_A^2}} + \frac{\tau_M\tau_F}{8\tau_A^2}-\frac{\tau_M}{2\tau_A}
    \right]\times\\
  \left\{
      {\rm erf}\left[
        \frac{\beta}{\sqrt{2}}-\sqrt{\frac{\tau_M}{2\tau_F}}+\sqrt{\frac{\tau_M\tau_F}{8\tau_A^2}}
        \right]\right.\\
      \left.
           -{\rm erf}\left[
             \frac{\beta}{\sqrt{2}}+\sqrt{\frac{\tau_M}{2\tau_F}}+\sqrt{\frac{\tau_M\tau_F}{8\tau_A^2}}
             \right]
    \right\}
\end{multline}
where $\beta$ is the decision boundary for the hard decisions based on
soft outcomes. The maximum likelihood hardening map corresponds to a
particular foice for $\beta$ as a function of $\tau_M/\tau_F$ and
$\tau_M/\tau_A$. We vary $\beta$ and $\tau_M$ to optimize hard
measurement performance.

\section{Inclusive and exclusive circuit noise models}
\label{app:inclusive_exclusive}

Here we discuss the distinction between two variations of the circuit noise model: an inclusive model in which all faults are independent, and an exclusive model in which some faults are mutually exclusive.
In \cref{subsec:circuit_model} we presented an inclusive circuit model. 
The standard model of circuit noise model is an exclusive model, in which each of the components of a circuit can fail in the following specific ways:
\begin{itemize}
	\item Idle qubits waiting for CNOT gates: fail with probability $\tilde{p}_{IG}$. When a failure occurs, an element of the set $\{ X, Y, Z \}$ is uniformly sampled and applied. 
	\item Idle qubits waiting for  measurements: fail with probability $\tilde{p}_{IM}$. When a failure occurs, an element of the set $\{ X, Y, Z \}$ is uniformly sampled and applied.
	\item CNOT gates: fail with probability $\tilde{p}_{\text{CNOT}}$. When a failure occurs, an element of the set 
	of 15 non-trivial weight-two Pauli operators
	is uniformly sampled and applied.
\end{itemize}
Here faults on different locations are applied independently of one another, although faults which occur on the same operation are exclusive.  
The measurement in the inclusive and exclusive models are treated identically.
The standard exclusive circuit noise model shown here and the inclusive circuit noise model described in \cref{subsec:circuit_model} exactly equivalent as proven in Appendix~E of~\cite{chao2020optimization}. 
The equivalence comes with a map for the probabilities between the two models:
\begin{itemize}
	\item $\tilde{p}_{IG} = p_{IG} + O(p_{IG}^2)$.
	\item $\tilde{p}_{IM}= p_{IM} + O(p_{IM}^2)$.
	\item $\tilde{p}_{\text{CNOT}} = p_{\text{CNOT}} +O(p_{\text{CNOT}}^2)$.
\end{itemize}
For low error rates the second order corrections (which can be calculated using the formulas specified in Ref.~\cite{gidneyDecorrelatedDepBlog}) are small. 
For example, for fault probabilities of one percent and below, the corrections are at most $\num{5e-5}$.
We therefore neglect any corrections between the exclusive and inclusive circuit noise model throughout this work.

\section{Logical error rate per round}
\label{app:bc_indep}

Here we discuss the model described in \cref{subsec:logical_error_rate} that we use to extract an estimate of the logical error rate per round from the probability $p_d^{X~\text{fail}}(T)$ of a logical $X$ failure of the error correction protocol with a distance-$d$ code over $T$ rounds. 
Our model is similar to that in Reference~\cite{fowler2021repetition} and is based on the assumption that during each round a logical $X$ error occurs with probability $\bar{p}/2$ independently of other rounds. 
After $T$ rounds, there is a logical $X$ failure in this model if an odd number of rounds have experienced logical $X$ errors, which occurs with probability
\begin{align}
\hat{p}^{X~\text{fail}}(T) = \frac{1 - (1-\bar{p})^T}{2},\label{eq:failure_rate_per_round_app}
\end{align}
which is a repeat of \cref{eq:failure_rate_per_round} in the main text.
We use hats to distinguish the model from the $X$ failure of the protocol itself.

While the assumptions that define this model may not apply to the protocol itself, for a carefully selected value of $\bar{p}$ and for large $T$ the formula $\hat{p}^{X~\text{fail}}(T)$ fits the data taken for the protocol $p^{X~\text{fail}}_d(T)$ very well; see \cref{fig:rate_convergence}.
The good fit can be understood from the fact that each of the decoders acts somewhat locally in the 3D decoding graph, and the matching of errors which are well separated in this graph occurs almost independently. 
In the regime of low logical error rate, this can be expected to lead to logical errors (which are sparse in this regime) which are also independent of one another. 

\begin{figure}
    \includegraphics[width=\columnwidth]{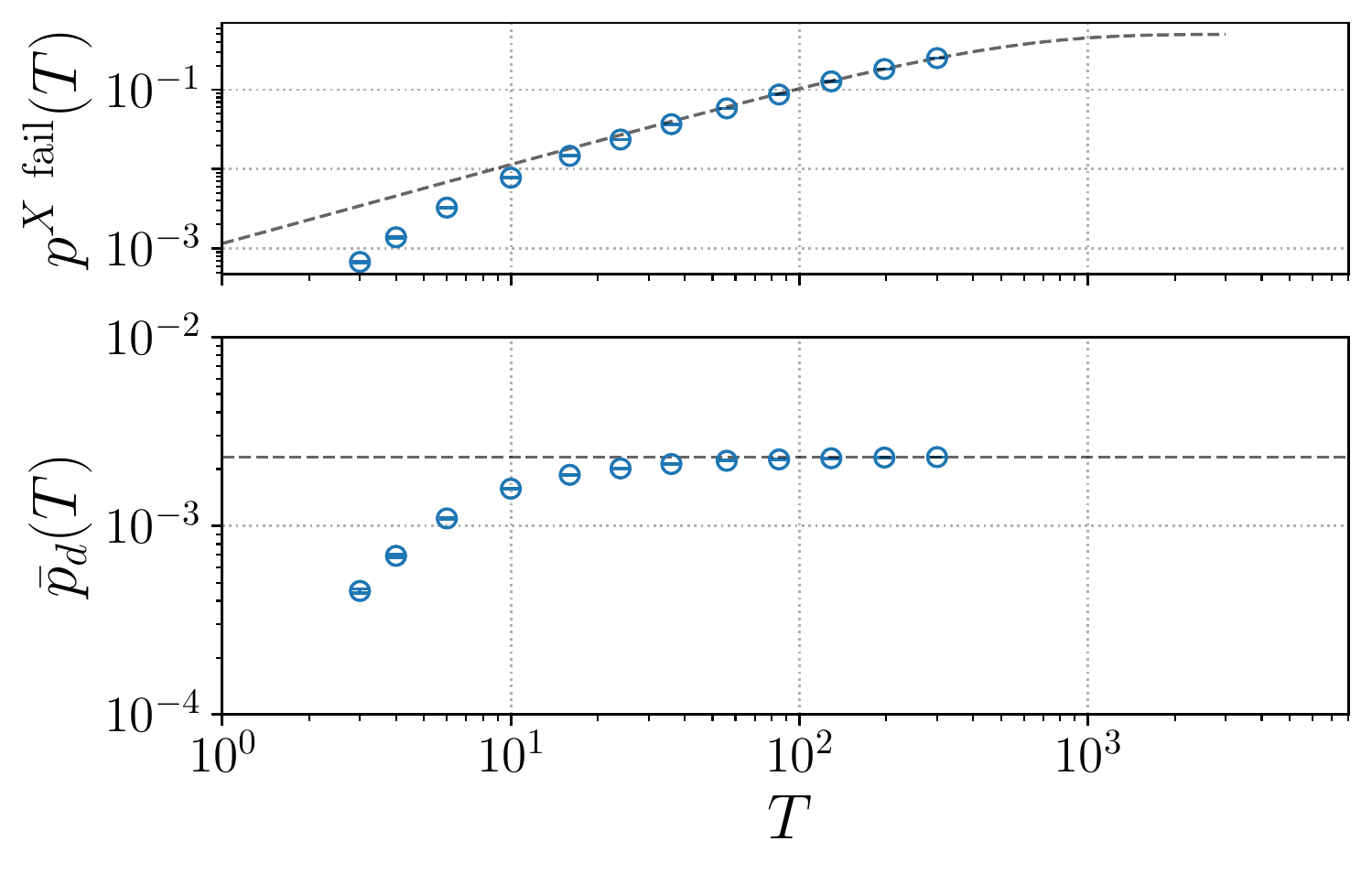}
    \caption{\label{fig:rate_convergence}
    Convergence of $\bar{p}_d(T)$ calculated using \cref{eq:failure_rate_per_round_data_T} with $T$ using the soft phenomenological noise model.
    (a) The protocol's $X$ failure probability $p_d^{X~\text{fail}}(T)$ for a distance $d=15$ surface code for a range of $T$.
    With a dashed line we show the model, namely $\hat{p}^{X~\text{fail}}(T) = (1 - (1-\bar{p})^T)/2$ evaluated at $\bar{p}=\bar{p}_d(300)$.
    (b) For each data point $p_d^{X~\text{fail}}(T)$, we estimate $\bar{p}_d(T)$ according to \cref{eq:failure_rate_per_round_data_T}. 
    We observe a clear convergence of $\bar{p}_d(T)$ toward an asymptotic value by $T=100$. We indicate $\bar{p}_d(300)$ by a dashed line.
    }
\end{figure}

The above model includes no explicit restriction on the correlation between logical $X$ and logical $Z$ failures. 
Let us explore the consequences of some additional assumptions.
We can make the additional assumption that during each round a logical $Z$ error occurs with probability $\bar{p}/2$ independently of other rounds, such that $\hat{p}^{Z~\text{fail}}(T)=\hat{p}^{X~\text{fail}}(T)$.
If furthermore we assume that $X$ and $Z$ failures are uncorrelated, then the overall probability $\hat{p}^{\text{fail,unc}}(T)$ of either a logical $X$ failure, a logical $Z$ failure or both is
\begin{equation}
\hat{p}^{X~\text{fail}}(T)+\hat{p}^{Z~\text{fail}}(T) - \hat{p}^{X~\text{fail}}(T)\hat{p}^{Z~\text{fail}}(T).
\end{equation}
In this case of equal $X$ and $Z$ logical failure probabilities without correlation, $\bar{p}$ approximately corresponds to the overall failure probability of a single round $\hat{p}^{\text{fail,unc}}(1) = \bar{p} - \bar{p}^2/4$. 
As we are primarily interested in the regime of small $\bar{p}$, the $O(\bar{p}^2)$ term can be neglected, which justifies our reference to $\bar{p}$ as the logical failure probability per round.

\section{Simulation Methods}
\label{app:simulation_protocol}
Here we describe some further details of the simulation methods we used to collect and analyze the data presented in the main text.

All of our data involves Monte Carlo sampling to estimate the probability $q$ of success of the error correction protocol described in \cref{subsec:logical_error_rate} given some noise model, code distance, and decoder. 
Given $n$ trials of which $k$ succeed, we use Bayesian inference to estimate the parameter of the corresponding Bernoulli distribution $\mathrm{Bernoulli}(q)$~\cite{gelman2013bayesian}.
We take the Jeffreys prior, so that the posterior distribution of $q$ given the data is $\mathrm{Beta}(\frac{1}{2} + k, \frac{1}{2} + n - k)$.
We report the mean of the posterior distribution $\frac{k+1/2}{n+1/2}$ and 68\% equal-tailed credible intervals calculated directly from the inverse cumulative distribution function of the Beta distribution.
In many data points the intervals indicating statistical uncertainty (indicated by vertical lines with end caps) are smaller than the symbols which mark the data point.

For a noise model parameterized by a single parameter $p$, we often seek to identify a threshold value $p^*$ of $p$ such as in \cref{fig:thresh_sweep_pheno} and \cref{fig:low_p_circuit_noise}.
We obtain a set of data points $p^{\text{fail}}_d(d)$ for a range of distances and values of $p$ in the vicinity of the suspected value of $p^*$.
Following the usual universal scaling ansatz for critical points in the study of phase transitions, we fit a quadratic in the scaled variable $x=(p-p^*) d^{1/\nu}$ to our data in this regime.

To minimize sensitivity to asymmetric measurement errors in models such as the amplitude damping model of \cref{subsec:amplitude_damping_model}, at each measurement round, each ancilla qubit is initialized in a known measurement eigenstate selected uniformly at random.


\begin{thebibliography}{52}%
\makeatletter
\providecommand \@ifxundefined [1]{%
 \@ifx{#1\undefined}
}%
\providecommand \@ifnum [1]{%
 \ifnum #1\expandafter \@firstoftwo
 \else \expandafter \@secondoftwo
 \fi
}%
\providecommand \@ifx [1]{%
 \ifx #1\expandafter \@firstoftwo
 \else \expandafter \@secondoftwo
 \fi
}%
\providecommand \natexlab [1]{#1}%
\providecommand \enquote  [1]{``#1''}%
\providecommand \bibnamefont  [1]{#1}%
\providecommand \bibfnamefont [1]{#1}%
\providecommand \citenamefont [1]{#1}%
\providecommand \href@noop [0]{\@secondoftwo}%
\providecommand \href [0]{\begingroup \@sanitize@url \@href}%
\providecommand \@href[1]{\@@startlink{#1}\@@href}%
\providecommand \@@href[1]{\endgroup#1\@@endlink}%
\providecommand \@sanitize@url [0]{\catcode `\\12\catcode `\$12\catcode
  `\&12\catcode `\#12\catcode `\^12\catcode `\_12\catcode `\%12\relax}%
\providecommand \@@startlink[1]{}%
\providecommand \@@endlink[0]{}%
\providecommand \url  [0]{\begingroup\@sanitize@url \@url }%
\providecommand \@url [1]{\endgroup\@href {#1}{\urlprefix }}%
\providecommand \urlprefix  [0]{URL }%
\providecommand \Eprint [0]{\href }%
\providecommand \doibase [0]{https://doi.org/}%
\providecommand \selectlanguage [0]{\@gobble}%
\providecommand \bibinfo  [0]{\@secondoftwo}%
\providecommand \bibfield  [0]{\@secondoftwo}%
\providecommand \translation [1]{[#1]}%
\providecommand \BibitemOpen [0]{}%
\providecommand \bibitemStop [0]{}%
\providecommand \bibitemNoStop [0]{.\EOS\space}%
\providecommand \EOS [0]{\spacefactor3000\relax}%
\providecommand \BibitemShut  [1]{\csname bibitem#1\endcsname}%
\let\auto@bib@innerbib\@empty
\bibitem [{\citenamefont {Shor}(1996)}]{shor1996}%
  \BibitemOpen
  \bibfield  {author} {\bibinfo {author} {\bibfnamefont {P.}~\bibnamefont
  {Shor}},\ }\bibfield  {title} {\bibinfo {title} {Fault-tolerant quantum
  computation},\ }in\ \href {https://doi.org/10.1109/SFCS.1996.548464} {\emph
  {\bibinfo {booktitle} {Proceedings of 37th Conference on Foundations of
  Computer Science}}}\ (\bibinfo {year} {1996})\ pp.\ \bibinfo {pages}
  {56--65}\BibitemShut {NoStop}%
\bibitem [{\citenamefont {Kitaev}(2003{\natexlab{a}})}]{kitaev2003top_codes}%
  \BibitemOpen
  \bibfield  {author} {\bibinfo {author} {\bibfnamefont {A.~Y.}\ \bibnamefont
  {Kitaev}},\ }\bibfield  {title} {\bibinfo {title} {Fault-tolerant quantum
  computation by anyons},\ }\href@noop {} {\bibfield  {journal} {\bibinfo
  {journal} {Annals of Physics}\ }\textbf {\bibinfo {volume} {303}},\ \bibinfo
  {pages} {2--30} (\bibinfo {year} {2003}{\natexlab{a}})}\BibitemShut {NoStop}%
\bibitem [{\citenamefont {Raussendorf}\ and\ \citenamefont
  {Harrington}(2007)}]{raussendorf2007fault}%
  \BibitemOpen
  \bibfield  {author} {\bibinfo {author} {\bibfnamefont {R.}~\bibnamefont
  {Raussendorf}}\ and\ \bibinfo {author} {\bibfnamefont {J.}~\bibnamefont
  {Harrington}},\ }\bibfield  {title} {\bibinfo {title} {Fault-tolerant quantum
  computation with high threshold in two dimensions},\ }\href@noop {}
  {\bibfield  {journal} {\bibinfo  {journal} {Physical Review Letters}\
  }\textbf {\bibinfo {volume} {98}},\ \bibinfo {pages} {190504} (\bibinfo
  {year} {2007})}\BibitemShut {NoStop}%
\bibitem [{\citenamefont {Dennis}\ \emph {et~al.}(2002)\citenamefont {Dennis},
  \citenamefont {Kitaev}, \citenamefont {Landahl},\ and\ \citenamefont
  {Preskill}}]{dennis2002tqm}%
  \BibitemOpen
  \bibfield  {author} {\bibinfo {author} {\bibfnamefont {E.}~\bibnamefont
  {Dennis}}, \bibinfo {author} {\bibfnamefont {A.}~\bibnamefont {Kitaev}},
  \bibinfo {author} {\bibfnamefont {A.}~\bibnamefont {Landahl}},\ and\ \bibinfo
  {author} {\bibfnamefont {J.}~\bibnamefont {Preskill}},\ }\bibfield  {title}
  {\bibinfo {title} {Topological quantum memory},\ }\href@noop {} {\bibfield
  {journal} {\bibinfo  {journal} {Journal of Mathematical Physics}\ }\textbf
  {\bibinfo {volume} {43}},\ \bibinfo {pages} {4452--4505} (\bibinfo {year}
  {2002})}\BibitemShut {NoStop}%
\bibitem [{\citenamefont {Fowler}\ \emph {et~al.}(2009)\citenamefont {Fowler},
  \citenamefont {Stephens},\ and\ \citenamefont
  {Groszkowski}}]{fowler2009high}%
  \BibitemOpen
  \bibfield  {author} {\bibinfo {author} {\bibfnamefont {A.~G.}\ \bibnamefont
  {Fowler}}, \bibinfo {author} {\bibfnamefont {A.~M.}\ \bibnamefont
  {Stephens}},\ and\ \bibinfo {author} {\bibfnamefont {P.}~\bibnamefont
  {Groszkowski}},\ }\bibfield  {title} {\bibinfo {title} {High-threshold
  universal quantum computation on the surface code},\ }\href@noop {}
  {\bibfield  {journal} {\bibinfo  {journal} {Physical Review A}\ }\textbf
  {\bibinfo {volume} {80}},\ \bibinfo {pages} {052312} (\bibinfo {year}
  {2009})}\BibitemShut {NoStop}%
\bibitem [{\citenamefont {Bravyi}\ \emph {et~al.}(2018)\citenamefont {Bravyi},
  \citenamefont {Englbrecht}, \citenamefont {K{\"o}nig},\ and\ \citenamefont
  {Peard}}]{bravyi2018surface_code_coherent_noise}%
  \BibitemOpen
  \bibfield  {author} {\bibinfo {author} {\bibfnamefont {S.}~\bibnamefont
  {Bravyi}}, \bibinfo {author} {\bibfnamefont {M.}~\bibnamefont {Englbrecht}},
  \bibinfo {author} {\bibfnamefont {R.}~\bibnamefont {K{\"o}nig}},\ and\
  \bibinfo {author} {\bibfnamefont {N.}~\bibnamefont {Peard}},\ }\bibfield
  {title} {\bibinfo {title} {Correcting coherent errors with surface codes},\
  }\href@noop {} {\bibfield  {journal} {\bibinfo  {journal} {npj Quantum
  Information}\ }\textbf {\bibinfo {volume} {4}},\ \bibinfo {pages} {1--6}
  (\bibinfo {year} {2018})}\BibitemShut {NoStop}%
\bibitem [{\citenamefont {Wang}\ \emph {et~al.}(2011)\citenamefont {Wang},
  \citenamefont {Courtade}, \citenamefont {Shankar},\ and\ \citenamefont
  {Wesel}}]{wang2011softLDPC}%
  \BibitemOpen
  \bibfield  {author} {\bibinfo {author} {\bibfnamefont {J.}~\bibnamefont
  {Wang}}, \bibinfo {author} {\bibfnamefont {T.}~\bibnamefont {Courtade}},
  \bibinfo {author} {\bibfnamefont {H.}~\bibnamefont {Shankar}},\ and\ \bibinfo
  {author} {\bibfnamefont {R.~D.}\ \bibnamefont {Wesel}},\ }\bibfield  {title}
  {\bibinfo {title} {Soft information for {LDPC} decoding in flash:
  Mutual-information optimized quantization},\ }in\ \href@noop {} {\emph
  {\bibinfo {booktitle} {2011 IEEE Global Telecommunications
  Conference-GLOBECOM 2011}}}\ (\bibinfo {organization} {IEEE},\ \bibinfo
  {year} {2011})\ pp.\ \bibinfo {pages} {1--6}\BibitemShut {NoStop}%
\bibitem [{\citenamefont {Costello}\ and\ \citenamefont
  {Forney}(2007)}]{costello2007channel}%
  \BibitemOpen
  \bibfield  {author} {\bibinfo {author} {\bibfnamefont {D.~J.}\ \bibnamefont
  {Costello}}\ and\ \bibinfo {author} {\bibfnamefont {G.~D.}\ \bibnamefont
  {Forney}},\ }\bibfield  {title} {\bibinfo {title} {Channel coding: The road
  to channel capacity},\ }\href@noop {} {\bibfield  {journal} {\bibinfo
  {journal} {Proceedings of the IEEE}\ }\textbf {\bibinfo {volume} {95}},\
  \bibinfo {pages} {1150--1177} (\bibinfo {year} {2007})}\BibitemShut {NoStop}%
\bibitem [{\citenamefont {Gottesman}\ \emph {et~al.}(2001)\citenamefont
  {Gottesman}, \citenamefont {Kitaev},\ and\ \citenamefont
  {Preskill}}]{gottesman2001encoding}%
  \BibitemOpen
  \bibfield  {author} {\bibinfo {author} {\bibfnamefont {D.}~\bibnamefont
  {Gottesman}}, \bibinfo {author} {\bibfnamefont {A.}~\bibnamefont {Kitaev}},\
  and\ \bibinfo {author} {\bibfnamefont {J.}~\bibnamefont {Preskill}},\
  }\bibfield  {title} {\bibinfo {title} {Encoding a qubit in an oscillator},\
  }\href@noop {} {\bibfield  {journal} {\bibinfo  {journal} {Physical Review
  A}\ }\textbf {\bibinfo {volume} {64}},\ \bibinfo {pages} {012310} (\bibinfo
  {year} {2001})}\BibitemShut {NoStop}%
\bibitem [{\citenamefont {Fukui}\ \emph {et~al.}(2017)\citenamefont {Fukui},
  \citenamefont {Tomita},\ and\ \citenamefont {Okamoto}}]{fukui2017analog}%
  \BibitemOpen
  \bibfield  {author} {\bibinfo {author} {\bibfnamefont {K.}~\bibnamefont
  {Fukui}}, \bibinfo {author} {\bibfnamefont {A.}~\bibnamefont {Tomita}},\ and\
  \bibinfo {author} {\bibfnamefont {A.}~\bibnamefont {Okamoto}},\ }\bibfield
  {title} {\bibinfo {title} {Analog quantum error correction with encoding a
  qubit into an oscillator},\ }\href@noop {} {\bibfield  {journal} {\bibinfo
  {journal} {Physical Review Letters}\ }\textbf {\bibinfo {volume} {119}},\
  \bibinfo {pages} {180507} (\bibinfo {year} {2017})}\BibitemShut {NoStop}%
\bibitem [{\citenamefont {Vuillot}\ \emph {et~al.}(2019)\citenamefont
  {Vuillot}, \citenamefont {Asasi}, \citenamefont {Wang}, \citenamefont
  {Pryadko},\ and\ \citenamefont {Terhal}}]{vuillotquantum}%
  \BibitemOpen
  \bibfield  {author} {\bibinfo {author} {\bibfnamefont {C.}~\bibnamefont
  {Vuillot}}, \bibinfo {author} {\bibfnamefont {H.}~\bibnamefont {Asasi}},
  \bibinfo {author} {\bibfnamefont {Y.}~\bibnamefont {Wang}}, \bibinfo {author}
  {\bibfnamefont {L.~P.}\ \bibnamefont {Pryadko}},\ and\ \bibinfo {author}
  {\bibfnamefont {B.~M.}\ \bibnamefont {Terhal}},\ }\bibfield  {title}
  {\bibinfo {title} {Quantum error correction with the toric
  {G}ottesman-{K}itaev-{P}reskill code},\ }\href
  {https://doi.org/10.1103/PhysRevA.99.032344} {\bibfield  {journal} {\bibinfo
  {journal} {Phys. Rev. A}\ }\textbf {\bibinfo {volume} {99}},\ \bibinfo
  {pages} {032344} (\bibinfo {year} {2019})}\BibitemShut {NoStop}%
\bibitem [{\citenamefont {Noh}\ and\ \citenamefont
  {Chamberland}(2020)}]{noh2020fault}%
  \BibitemOpen
  \bibfield  {author} {\bibinfo {author} {\bibfnamefont {K.}~\bibnamefont
  {Noh}}\ and\ \bibinfo {author} {\bibfnamefont {C.}~\bibnamefont
  {Chamberland}},\ }\bibfield  {title} {\bibinfo {title} {Fault-tolerant
  bosonic quantum error correction with the
  surface--{Gottesman-Kitaev-Preskill} code},\ }\href@noop {} {\bibfield
  {journal} {\bibinfo  {journal} {Physical Review A}\ }\textbf {\bibinfo
  {volume} {101}},\ \bibinfo {pages} {012316} (\bibinfo {year}
  {2020})}\BibitemShut {NoStop}%
\bibitem [{\citenamefont {Noh}\ \emph {et~al.}(2021)\citenamefont {Noh},
  \citenamefont {Chamberland},\ and\ \citenamefont {Brand{\~a}o}}]{noh2021low}%
  \BibitemOpen
  \bibfield  {author} {\bibinfo {author} {\bibfnamefont {K.}~\bibnamefont
  {Noh}}, \bibinfo {author} {\bibfnamefont {C.}~\bibnamefont {Chamberland}},\
  and\ \bibinfo {author} {\bibfnamefont {F.~G.}\ \bibnamefont {Brand{\~a}o}},\
  }\bibfield  {title} {\bibinfo {title} {Low overhead fault-tolerant quantum
  error correction with the surface-{GKP} code},\ }\href@noop {} {\bibfield
  {journal} {\bibinfo  {journal} {arXiv preprint arXiv:2103.06994}\ } (\bibinfo
  {year} {2021})}\BibitemShut {NoStop}%
\bibitem [{\citenamefont {Chamberland}\ \emph {et~al.}(2020)\citenamefont
  {Chamberland}, \citenamefont {Noh}, \citenamefont {Arrangoiz-Arriola},
  \citenamefont {Campbell}, \citenamefont {Hann}, \citenamefont {Iverson},
  \citenamefont {Putterman}, \citenamefont {Bohdanowicz}, \citenamefont
  {Flammia}, \citenamefont {Keller} \emph {et~al.}}]{chamberland2020building}%
  \BibitemOpen
  \bibfield  {author} {\bibinfo {author} {\bibfnamefont {C.}~\bibnamefont
  {Chamberland}}, \bibinfo {author} {\bibfnamefont {K.}~\bibnamefont {Noh}},
  \bibinfo {author} {\bibfnamefont {P.}~\bibnamefont {Arrangoiz-Arriola}},
  \bibinfo {author} {\bibfnamefont {E.~T.}\ \bibnamefont {Campbell}}, \bibinfo
  {author} {\bibfnamefont {C.~T.}\ \bibnamefont {Hann}}, \bibinfo {author}
  {\bibfnamefont {J.}~\bibnamefont {Iverson}}, \bibinfo {author} {\bibfnamefont
  {H.}~\bibnamefont {Putterman}}, \bibinfo {author} {\bibfnamefont {T.~C.}\
  \bibnamefont {Bohdanowicz}}, \bibinfo {author} {\bibfnamefont {S.~T.}\
  \bibnamefont {Flammia}}, \bibinfo {author} {\bibfnamefont {A.}~\bibnamefont
  {Keller}}, \emph {et~al.},\ }\bibfield  {title} {\bibinfo {title} {Building a
  fault-tolerant quantum computer using concatenated cat codes},\ }\href@noop
  {} {\bibfield  {journal} {\bibinfo  {journal} {arXiv preprint
  arXiv:2012.04108}\ } (\bibinfo {year} {2020})}\BibitemShut {NoStop}%
\bibitem [{\citenamefont {Delfosse}\ and\ \citenamefont
  {Nickerson}(2017)}]{delfosse2017UF}%
  \BibitemOpen
  \bibfield  {author} {\bibinfo {author} {\bibfnamefont {N.}~\bibnamefont
  {Delfosse}}\ and\ \bibinfo {author} {\bibfnamefont {N.~H.}\ \bibnamefont
  {Nickerson}},\ }\bibfield  {title} {\bibinfo {title} {Almost-linear time
  decoding algorithm for topological codes},\ }\href@noop {} {\bibfield
  {journal} {\bibinfo  {journal} {arXiv preprint arXiv:1709.06218}\ } (\bibinfo
  {year} {2017})}\BibitemShut {NoStop}%
\bibitem [{\citenamefont {Bombin}\ and\ \citenamefont
  {Martin-Delgado}(2007)}]{bombin2007homological}%
  \BibitemOpen
  \bibfield  {author} {\bibinfo {author} {\bibfnamefont {H.}~\bibnamefont
  {Bombin}}\ and\ \bibinfo {author} {\bibfnamefont {M.~A.}\ \bibnamefont
  {Martin-Delgado}},\ }\bibfield  {title} {\bibinfo {title} {Homological error
  correction: Classical and quantum codes},\ }\href
  {https://doi.org/10.1063/1.2731356} {\bibfield  {journal} {\bibinfo
  {journal} {Journal of Mathematical Physics}\ }\textbf {\bibinfo {volume}
  {48}},\ \bibinfo {pages} {052105} (\bibinfo {year} {2007})}\BibitemShut
  {NoStop}%
\bibitem [{\citenamefont {Delfosse}(2014)}]{delfosse2014projection}%
  \BibitemOpen
  \bibfield  {author} {\bibinfo {author} {\bibfnamefont {N.}~\bibnamefont
  {Delfosse}},\ }\bibfield  {title} {\bibinfo {title} {Decoding color codes by
  projection onto surface codes},\ }\href@noop {} {\bibfield  {journal}
  {\bibinfo  {journal} {Physical Review A}\ }\textbf {\bibinfo {volume} {89}},\
  \bibinfo {pages} {012317} (\bibinfo {year} {2014})}\BibitemShut {NoStop}%
\bibitem [{\citenamefont {Chao}\ \emph {et~al.}(2020)\citenamefont {Chao},
  \citenamefont {Beverland}, \citenamefont {Delfosse},\ and\ \citenamefont
  {Haah}}]{chao2020optimization}%
  \BibitemOpen
  \bibfield  {author} {\bibinfo {author} {\bibfnamefont {R.}~\bibnamefont
  {Chao}}, \bibinfo {author} {\bibfnamefont {M.~E.}\ \bibnamefont {Beverland}},
  \bibinfo {author} {\bibfnamefont {N.}~\bibnamefont {Delfosse}},\ and\
  \bibinfo {author} {\bibfnamefont {J.}~\bibnamefont {Haah}},\ }\bibfield
  {title} {\bibinfo {title} {Optimization of the surface code design for
  majorana-based qubits},\ }\href@noop {} {\bibfield  {journal} {\bibinfo
  {journal} {Quantum}\ }\textbf {\bibinfo {volume} {4}},\ \bibinfo {pages}
  {352} (\bibinfo {year} {2020})}\BibitemShut {NoStop}%
\bibitem [{\citenamefont {Delfosse}\ \emph
  {et~al.}(2021{\natexlab{a}})\citenamefont {Delfosse}, \citenamefont {Londe},\
  and\ \citenamefont {Beverland}}]{delfosse2021unionfind}%
  \BibitemOpen
  \bibfield  {author} {\bibinfo {author} {\bibfnamefont {N.}~\bibnamefont
  {Delfosse}}, \bibinfo {author} {\bibfnamefont {V.}~\bibnamefont {Londe}},\
  and\ \bibinfo {author} {\bibfnamefont {M.}~\bibnamefont {Beverland}},\
  }\href@noop {} {\bibinfo {title} {Toward a {Union-Find} decoder for quantum
  {LDPC} codes}} (\bibinfo {year} {2021}{\natexlab{a}}),\ \Eprint
  {https://arxiv.org/abs/2103.08049} {arXiv:2103.08049 [quant-ph]} \BibitemShut
  {NoStop}%
\bibitem [{\citenamefont {Hastings}\ and\ \citenamefont
  {Haah}(2021)}]{hastings2021dynamically}%
  \BibitemOpen
  \bibfield  {author} {\bibinfo {author} {\bibfnamefont {M.~B.}\ \bibnamefont
  {Hastings}}\ and\ \bibinfo {author} {\bibfnamefont {J.}~\bibnamefont
  {Haah}},\ }\href@noop {} {\bibinfo {title} {Dynamically generated logical
  qubits}} (\bibinfo {year} {2021}),\ \Eprint
  {https://arxiv.org/abs/2107.02194} {arXiv:2107.02194 [quant-ph]} \BibitemShut
  {NoStop}%
\bibitem [{\citenamefont {Wang}\ \emph {et~al.}(2003)\citenamefont {Wang},
  \citenamefont {Harrington},\ and\ \citenamefont
  {Preskill}}]{wang2003confinement_higgs}%
  \BibitemOpen
  \bibfield  {author} {\bibinfo {author} {\bibfnamefont {C.}~\bibnamefont
  {Wang}}, \bibinfo {author} {\bibfnamefont {J.}~\bibnamefont {Harrington}},\
  and\ \bibinfo {author} {\bibfnamefont {J.}~\bibnamefont {Preskill}},\
  }\bibfield  {title} {\bibinfo {title} {Confinement-{Higgs} transition in a
  disordered gauge theory and the accuracy threshold for quantum memory},\
  }\href@noop {} {\bibfield  {journal} {\bibinfo  {journal} {Annals of
  Physics}\ }\textbf {\bibinfo {volume} {303}},\ \bibinfo {pages} {31--58}
  (\bibinfo {year} {2003})}\BibitemShut {NoStop}%
\bibitem [{\citenamefont {Berge}(1973)}]{berge1973graphs}%
  \BibitemOpen
  \bibfield  {author} {\bibinfo {author} {\bibfnamefont {C.}~\bibnamefont
  {Berge}},\ }\href@noop {} {\emph {\bibinfo {title} {Graphs and
  hypergraphs}}}\ (\bibinfo  {publisher} {North-Holland Pub. Co.},\ \bibinfo
  {year} {1973})\BibitemShut {NoStop}%
\bibitem [{\citenamefont {Fowler}\ \emph {et~al.}(2012)\citenamefont {Fowler},
  \citenamefont {Mariantoni}, \citenamefont {Martinis},\ and\ \citenamefont
  {Cleland}}]{fowler2012surface}%
  \BibitemOpen
  \bibfield  {author} {\bibinfo {author} {\bibfnamefont {A.~G.}\ \bibnamefont
  {Fowler}}, \bibinfo {author} {\bibfnamefont {M.}~\bibnamefont {Mariantoni}},
  \bibinfo {author} {\bibfnamefont {J.~M.}\ \bibnamefont {Martinis}},\ and\
  \bibinfo {author} {\bibfnamefont {A.~N.}\ \bibnamefont {Cleland}},\
  }\bibfield  {title} {\bibinfo {title} {Surface codes: Towards practical
  large-scale quantum computation},\ }\href@noop {} {\bibfield  {journal}
  {\bibinfo  {journal} {Physical Review A}\ }\textbf {\bibinfo {volume} {86}},\
  \bibinfo {pages} {032324} (\bibinfo {year} {2012})}\BibitemShut {NoStop}%
\bibitem [{\citenamefont {Litinski}(2019)}]{litinski2019game}%
  \BibitemOpen
  \bibfield  {author} {\bibinfo {author} {\bibfnamefont {D.}~\bibnamefont
  {Litinski}},\ }\bibfield  {title} {\bibinfo {title} {A game of surface codes:
  Large-scale quantum computing with lattice surgery},\ }\href@noop {}
  {\bibfield  {journal} {\bibinfo  {journal} {Quantum}\ }\textbf {\bibinfo
  {volume} {3}},\ \bibinfo {pages} {128} (\bibinfo {year} {2019})}\BibitemShut
  {NoStop}%
\bibitem [{\citenamefont {Kitaev}(2003{\natexlab{b}})}]{kitaev2003fault}%
  \BibitemOpen
  \bibfield  {author} {\bibinfo {author} {\bibfnamefont {A.~Y.}\ \bibnamefont
  {Kitaev}},\ }\bibfield  {title} {\bibinfo {title} {Fault-tolerant quantum
  computation by anyons},\ }\href@noop {} {\bibfield  {journal} {\bibinfo
  {journal} {Annals of Physics}\ }\textbf {\bibinfo {volume} {303}},\ \bibinfo
  {pages} {2--30} (\bibinfo {year} {2003}{\natexlab{b}})}\BibitemShut {NoStop}%
\bibitem [{\citenamefont {Breuckmann}\ \emph {et~al.}(2017)\citenamefont
  {Breuckmann}, \citenamefont {Vuillot}, \citenamefont {Campbell},
  \citenamefont {Krishna},\ and\ \citenamefont
  {Terhal}}]{breuckmann2017hyperbolic}%
  \BibitemOpen
  \bibfield  {author} {\bibinfo {author} {\bibfnamefont {N.~P.}\ \bibnamefont
  {Breuckmann}}, \bibinfo {author} {\bibfnamefont {C.}~\bibnamefont {Vuillot}},
  \bibinfo {author} {\bibfnamefont {E.}~\bibnamefont {Campbell}}, \bibinfo
  {author} {\bibfnamefont {A.}~\bibnamefont {Krishna}},\ and\ \bibinfo {author}
  {\bibfnamefont {B.~M.}\ \bibnamefont {Terhal}},\ }\bibfield  {title}
  {\bibinfo {title} {Hyperbolic and semi-hyperbolic surface codes for quantum
  storage},\ }\href@noop {} {\bibfield  {journal} {\bibinfo  {journal} {Quantum
  Science and Technology}\ }\textbf {\bibinfo {volume} {2}},\ \bibinfo {pages}
  {035007} (\bibinfo {year} {2017})}\BibitemShut {NoStop}%
\bibitem [{\citenamefont {Gambetta}\ \emph {et~al.}(2007)\citenamefont
  {Gambetta}, \citenamefont {Braff}, \citenamefont {Wallraff}, \citenamefont
  {Girvin},\ and\ \citenamefont {Schoelkopf}}]{Gambetta2007measurement}%
  \BibitemOpen
  \bibfield  {author} {\bibinfo {author} {\bibfnamefont {J.}~\bibnamefont
  {Gambetta}}, \bibinfo {author} {\bibfnamefont {W.~A.}\ \bibnamefont {Braff}},
  \bibinfo {author} {\bibfnamefont {A.}~\bibnamefont {Wallraff}}, \bibinfo
  {author} {\bibfnamefont {S.~M.}\ \bibnamefont {Girvin}},\ and\ \bibinfo
  {author} {\bibfnamefont {R.~J.}\ \bibnamefont {Schoelkopf}},\ }\bibfield
  {title} {\bibinfo {title} {Protocols for optimal readout of qubits using a
  continuous quantum nondemolition measurement},\ }\href
  {https://doi.org/10.1103/PhysRevA.76.012325} {\bibfield  {journal} {\bibinfo
  {journal} {Phys. Rev. A}\ }\textbf {\bibinfo {volume} {76}},\ \bibinfo
  {pages} {012325} (\bibinfo {year} {2007})}\BibitemShut {NoStop}%
\bibitem [{\citenamefont {Nielsen}\ and\ \citenamefont
  {Chuang}(2011)}]{nielsenChuang}%
  \BibitemOpen
  \bibfield  {author} {\bibinfo {author} {\bibfnamefont {M.}~\bibnamefont
  {Nielsen}}\ and\ \bibinfo {author} {\bibfnamefont {I.}~\bibnamefont
  {Chuang}},\ }\href@noop {} {\emph {\bibinfo {title} {Quantum Computation and
  Quantum Information: 10th Anniversary Edition}}}\ (\bibinfo  {publisher}
  {Cambridge University Press},\ \bibinfo {year} {2011})\BibitemShut {NoStop}%
\bibitem [{\citenamefont {Blais}\ \emph {et~al.}(2004)\citenamefont {Blais},
  \citenamefont {Huang}, \citenamefont {Wallraff}, \citenamefont {Girvin},\
  and\ \citenamefont {Schoelkopf}}]{Blais2004cqed}%
  \BibitemOpen
  \bibfield  {author} {\bibinfo {author} {\bibfnamefont {A.}~\bibnamefont
  {Blais}}, \bibinfo {author} {\bibfnamefont {R.-S.}\ \bibnamefont {Huang}},
  \bibinfo {author} {\bibfnamefont {A.}~\bibnamefont {Wallraff}}, \bibinfo
  {author} {\bibfnamefont {S.~M.}\ \bibnamefont {Girvin}},\ and\ \bibinfo
  {author} {\bibfnamefont {R.~J.}\ \bibnamefont {Schoelkopf}},\ }\bibfield
  {title} {\bibinfo {title} {Cavity quantum electrodynamics for superconducting
  electrical circuits: An architecture for quantum computation},\ }\href
  {https://doi.org/10.1103/PhysRevA.69.062320} {\bibfield  {journal} {\bibinfo
  {journal} {Phys. Rev. A}\ }\textbf {\bibinfo {volume} {69}},\ \bibinfo
  {pages} {062320} (\bibinfo {year} {2004})}\BibitemShut {NoStop}%
\bibitem [{\citenamefont {Colless}\ \emph {et~al.}(2013)\citenamefont
  {Colless}, \citenamefont {Mahoney}, \citenamefont {Hornibrook}, \citenamefont
  {Doherty}, \citenamefont {Lu}, \citenamefont {Gossard},\ and\ \citenamefont
  {Reilly}}]{Colless2013dotdispersive}%
  \BibitemOpen
  \bibfield  {author} {\bibinfo {author} {\bibfnamefont {J.~I.}\ \bibnamefont
  {Colless}}, \bibinfo {author} {\bibfnamefont {A.~C.}\ \bibnamefont
  {Mahoney}}, \bibinfo {author} {\bibfnamefont {J.~M.}\ \bibnamefont
  {Hornibrook}}, \bibinfo {author} {\bibfnamefont {A.~C.}\ \bibnamefont
  {Doherty}}, \bibinfo {author} {\bibfnamefont {H.}~\bibnamefont {Lu}},
  \bibinfo {author} {\bibfnamefont {A.~C.}\ \bibnamefont {Gossard}},\ and\
  \bibinfo {author} {\bibfnamefont {D.~J.}\ \bibnamefont {Reilly}},\ }\bibfield
   {title} {\bibinfo {title} {Dispersive readout of a few-electron double
  quantum dot with fast rf gate sensors},\ }\href
  {https://doi.org/10.1103/PhysRevLett.110.046805} {\bibfield  {journal}
  {\bibinfo  {journal} {Phys. Rev. Lett.}\ }\textbf {\bibinfo {volume} {110}},\
  \bibinfo {pages} {046805} (\bibinfo {year} {2013})}\BibitemShut {NoStop}%
\bibitem [{\citenamefont {Tomita}\ and\ \citenamefont
  {Svore}(2014)}]{tomita2014low}%
  \BibitemOpen
  \bibfield  {author} {\bibinfo {author} {\bibfnamefont {Y.}~\bibnamefont
  {Tomita}}\ and\ \bibinfo {author} {\bibfnamefont {K.~M.}\ \bibnamefont
  {Svore}},\ }\bibfield  {title} {\bibinfo {title} {Low-distance surface codes
  under realistic quantum noise},\ }\href@noop {} {\bibfield  {journal}
  {\bibinfo  {journal} {Physical Review A}\ }\textbf {\bibinfo {volume} {90}},\
  \bibinfo {pages} {062320} (\bibinfo {year} {2014})}\BibitemShut {NoStop}%
\bibitem [{\citenamefont {Garey}\ and\ \citenamefont
  {Johnson}(1979)}]{garey1979NP}%
  \BibitemOpen
  \bibfield  {author} {\bibinfo {author} {\bibfnamefont {M.~R.}\ \bibnamefont
  {Garey}}\ and\ \bibinfo {author} {\bibfnamefont {D.~S.}\ \bibnamefont
  {Johnson}},\ }\href@noop {} {\emph {\bibinfo {title} {Computers and
  Intractability: A Guide to the Theory of NP-Completeness}}}\ (\bibinfo
  {publisher} {W. H. Freeman \& Co.},\ \bibinfo {address} {USA},\ \bibinfo
  {year} {1979})\BibitemShut {NoStop}%
\bibitem [{\citenamefont {Edmonds}(1965)}]{edmonds1965mwpm}%
  \BibitemOpen
  \bibfield  {author} {\bibinfo {author} {\bibfnamefont {J.}~\bibnamefont
  {Edmonds}},\ }\bibfield  {title} {\bibinfo {title} {Paths, trees, and
  flowers},\ }\href@noop {} {\bibfield  {journal} {\bibinfo  {journal}
  {Canadian Journal of mathematics}\ }\textbf {\bibinfo {volume} {17}},\
  \bibinfo {pages} {449--467} (\bibinfo {year} {1965})}\BibitemShut {NoStop}%
\bibitem [{\citenamefont {Kolmogorov}(2009)}]{kolmogorov2009blossom}%
  \BibitemOpen
  \bibfield  {author} {\bibinfo {author} {\bibfnamefont {V.}~\bibnamefont
  {Kolmogorov}},\ }\bibfield  {title} {\bibinfo {title} {Blossom {V}: a new
  implementation of a minimum cost perfect matching algorithm},\ }\href@noop {}
  {\bibfield  {journal} {\bibinfo  {journal} {Mathematical Programming
  Computation}\ }\textbf {\bibinfo {volume} {1}},\ \bibinfo {pages} {43--67}
  (\bibinfo {year} {2009})}\BibitemShut {NoStop}%
\bibitem [{\citenamefont {Cook}\ and\ \citenamefont
  {Rohe}(1999)}]{cook1999blossom}%
  \BibitemOpen
  \bibfield  {author} {\bibinfo {author} {\bibfnamefont {W.}~\bibnamefont
  {Cook}}\ and\ \bibinfo {author} {\bibfnamefont {A.}~\bibnamefont {Rohe}},\
  }\bibfield  {title} {\bibinfo {title} {Computing minimum-weight perfect
  matchings},\ }\href@noop {} {\bibfield  {journal} {\bibinfo  {journal}
  {INFORMS journal on computing}\ }\textbf {\bibinfo {volume} {11}},\ \bibinfo
  {pages} {138--148} (\bibinfo {year} {1999})}\BibitemShut {NoStop}%
\bibitem [{\citenamefont {Vazirani}(1994)}]{vazirani1994mwpm}%
  \BibitemOpen
  \bibfield  {author} {\bibinfo {author} {\bibfnamefont {V.~V.}\ \bibnamefont
  {Vazirani}},\ }\bibfield  {title} {\bibinfo {title} {A theory of alternating
  paths and blossoms for proving correctness of the ${O} (\sqrt{V}{E})$ general
  graph maximum matching algorithm},\ }\href@noop {} {\bibfield  {journal}
  {\bibinfo  {journal} {Combinatorica}\ }\textbf {\bibinfo {volume} {14}},\
  \bibinfo {pages} {71--109} (\bibinfo {year} {1994})}\BibitemShut {NoStop}%
\bibitem [{\citenamefont {Huang}\ \emph {et~al.}(2020)\citenamefont {Huang},
  \citenamefont {Newman},\ and\ \citenamefont {Brown}}]{Huang2020UF}%
  \BibitemOpen
  \bibfield  {author} {\bibinfo {author} {\bibfnamefont {S.}~\bibnamefont
  {Huang}}, \bibinfo {author} {\bibfnamefont {M.}~\bibnamefont {Newman}},\ and\
  \bibinfo {author} {\bibfnamefont {K.~R.}\ \bibnamefont {Brown}},\ }\bibfield
  {title} {\bibinfo {title} {Fault-tolerant weighted {Union-Find} decoding on
  the toric code},\ }\href {https://doi.org/10.1103/PhysRevA.102.012419}
  {\bibfield  {journal} {\bibinfo  {journal} {Phys. Rev. A}\ }\textbf {\bibinfo
  {volume} {102}},\ \bibinfo {pages} {012419} (\bibinfo {year}
  {2020})}\BibitemShut {NoStop}%
\bibitem [{\citenamefont {Delfosse}\ and\ \citenamefont
  {Z{\'e}mor}(2020)}]{delfosse2020peeling}%
  \BibitemOpen
  \bibfield  {author} {\bibinfo {author} {\bibfnamefont {N.}~\bibnamefont
  {Delfosse}}\ and\ \bibinfo {author} {\bibfnamefont {G.}~\bibnamefont
  {Z{\'e}mor}},\ }\bibfield  {title} {\bibinfo {title} {Linear-time maximum
  likelihood decoding of surface codes over the quantum erasure channel},\
  }\href@noop {} {\bibfield  {journal} {\bibinfo  {journal} {Physical Review
  Research}\ }\textbf {\bibinfo {volume} {2}},\ \bibinfo {pages} {033042}
  (\bibinfo {year} {2020})}\BibitemShut {NoStop}%
\bibitem [{\citenamefont {Chen}\ \emph {et~al.}(2021)\citenamefont {Chen},
  \citenamefont {Satzinger}, \citenamefont {Atalaya}, \citenamefont {Korotkov},
  \citenamefont {Dunsworth}, \citenamefont {Sank}, \citenamefont {Quintana},
  \citenamefont {McEwen}, \citenamefont {Barends}, \citenamefont {Klimov},
  \citenamefont {Hong}, \citenamefont {Jones}, \citenamefont {Petukhov},
  \citenamefont {Kafri}, \citenamefont {Demura}, \citenamefont {Burkett},
  \citenamefont {Gidney}, \citenamefont {Fowler}, \citenamefont {Putterman},
  \citenamefont {Aleiner}, \citenamefont {Arute}, \citenamefont {Arya},
  \citenamefont {Babbush}, \citenamefont {Bardin}, \citenamefont {Bengtsson},
  \citenamefont {Bourassa}, \citenamefont {Broughton}, \citenamefont {Buckley},
  \citenamefont {Buell}, \citenamefont {Bushnell}, \citenamefont {Chiaro},
  \citenamefont {Collins}, \citenamefont {Courtney}, \citenamefont {Derk},
  \citenamefont {Eppens}, \citenamefont {Erickson}, \citenamefont {Farhi},
  \citenamefont {Foxen}, \citenamefont {Giustina}, \citenamefont {Gross},
  \citenamefont {Harrigan}, \citenamefont {Harrington}, \citenamefont {Hilton},
  \citenamefont {Ho}, \citenamefont {Huang}, \citenamefont {Huggins},
  \citenamefont {Ioffe}, \citenamefont {Isakov}, \citenamefont {Jeffrey},
  \citenamefont {Jiang}, \citenamefont {Kechedzhi}, \citenamefont {Kim},
  \citenamefont {Kostritsa}, \citenamefont {Landhuis}, \citenamefont {Laptev},
  \citenamefont {Lucero}, \citenamefont {Martin}, \citenamefont {McClean},
  \citenamefont {McCourt}, \citenamefont {Mi}, \citenamefont {Miao},
  \citenamefont {Mohseni}, \citenamefont {Mruczkiewicz}, \citenamefont {Mutus},
  \citenamefont {Naaman}, \citenamefont {Neeley}, \citenamefont {Neill},
  \citenamefont {Newman}, \citenamefont {Niu}, \citenamefont {O'Brien},
  \citenamefont {Opremcak}, \citenamefont {Ostby}, \citenamefont {Pató},
  \citenamefont {Redd}, \citenamefont {Roushan}, \citenamefont {Rubin},
  \citenamefont {Shvarts}, \citenamefont {Strain}, \citenamefont {Szalay},
  \citenamefont {Trevithick}, \citenamefont {Villalonga}, \citenamefont
  {White}, \citenamefont {Yao}, \citenamefont {Yeh}, \citenamefont {Zalcman},
  \citenamefont {Neven}, \citenamefont {Boixo}, \citenamefont {Smelyanskiy},
  \citenamefont {Chen}, \citenamefont {Megrant},\ and\ \citenamefont
  {Kelly}}]{fowler2021repetition}%
  \BibitemOpen
  \bibfield  {author} {\bibinfo {author} {\bibfnamefont {Z.}~\bibnamefont
  {Chen}}, \bibinfo {author} {\bibfnamefont {K.~J.}\ \bibnamefont {Satzinger}},
  \bibinfo {author} {\bibfnamefont {J.}~\bibnamefont {Atalaya}}, \bibinfo
  {author} {\bibfnamefont {A.~N.}\ \bibnamefont {Korotkov}}, \bibinfo {author}
  {\bibfnamefont {A.}~\bibnamefont {Dunsworth}}, \bibinfo {author}
  {\bibfnamefont {D.}~\bibnamefont {Sank}}, \bibinfo {author} {\bibfnamefont
  {C.}~\bibnamefont {Quintana}}, \bibinfo {author} {\bibfnamefont
  {M.}~\bibnamefont {McEwen}}, \bibinfo {author} {\bibfnamefont
  {R.}~\bibnamefont {Barends}}, \bibinfo {author} {\bibfnamefont {P.~V.}\
  \bibnamefont {Klimov}}, \bibinfo {author} {\bibfnamefont {S.}~\bibnamefont
  {Hong}}, \bibinfo {author} {\bibfnamefont {C.}~\bibnamefont {Jones}},
  \bibinfo {author} {\bibfnamefont {A.}~\bibnamefont {Petukhov}}, \bibinfo
  {author} {\bibfnamefont {D.}~\bibnamefont {Kafri}}, \bibinfo {author}
  {\bibfnamefont {S.}~\bibnamefont {Demura}}, \bibinfo {author} {\bibfnamefont
  {B.}~\bibnamefont {Burkett}}, \bibinfo {author} {\bibfnamefont
  {C.}~\bibnamefont {Gidney}}, \bibinfo {author} {\bibfnamefont {A.~G.}\
  \bibnamefont {Fowler}}, \bibinfo {author} {\bibfnamefont {H.}~\bibnamefont
  {Putterman}}, \bibinfo {author} {\bibfnamefont {I.}~\bibnamefont {Aleiner}},
  \bibinfo {author} {\bibfnamefont {F.}~\bibnamefont {Arute}}, \bibinfo
  {author} {\bibfnamefont {K.}~\bibnamefont {Arya}}, \bibinfo {author}
  {\bibfnamefont {R.}~\bibnamefont {Babbush}}, \bibinfo {author} {\bibfnamefont
  {J.~C.}\ \bibnamefont {Bardin}}, \bibinfo {author} {\bibfnamefont
  {A.}~\bibnamefont {Bengtsson}}, \bibinfo {author} {\bibfnamefont
  {A.}~\bibnamefont {Bourassa}}, \bibinfo {author} {\bibfnamefont
  {M.}~\bibnamefont {Broughton}}, \bibinfo {author} {\bibfnamefont {B.~B.}\
  \bibnamefont {Buckley}}, \bibinfo {author} {\bibfnamefont {D.~A.}\
  \bibnamefont {Buell}}, \bibinfo {author} {\bibfnamefont {N.}~\bibnamefont
  {Bushnell}}, \bibinfo {author} {\bibfnamefont {B.}~\bibnamefont {Chiaro}},
  \bibinfo {author} {\bibfnamefont {R.}~\bibnamefont {Collins}}, \bibinfo
  {author} {\bibfnamefont {W.}~\bibnamefont {Courtney}}, \bibinfo {author}
  {\bibfnamefont {A.~R.}\ \bibnamefont {Derk}}, \bibinfo {author}
  {\bibfnamefont {D.}~\bibnamefont {Eppens}}, \bibinfo {author} {\bibfnamefont
  {C.}~\bibnamefont {Erickson}}, \bibinfo {author} {\bibfnamefont
  {E.}~\bibnamefont {Farhi}}, \bibinfo {author} {\bibfnamefont
  {B.}~\bibnamefont {Foxen}}, \bibinfo {author} {\bibfnamefont
  {M.}~\bibnamefont {Giustina}}, \bibinfo {author} {\bibfnamefont {J.~A.}\
  \bibnamefont {Gross}}, \bibinfo {author} {\bibfnamefont {M.~P.}\ \bibnamefont
  {Harrigan}}, \bibinfo {author} {\bibfnamefont {S.~D.}\ \bibnamefont
  {Harrington}}, \bibinfo {author} {\bibfnamefont {J.}~\bibnamefont {Hilton}},
  \bibinfo {author} {\bibfnamefont {A.}~\bibnamefont {Ho}}, \bibinfo {author}
  {\bibfnamefont {T.}~\bibnamefont {Huang}}, \bibinfo {author} {\bibfnamefont
  {W.~J.}\ \bibnamefont {Huggins}}, \bibinfo {author} {\bibfnamefont {L.~B.}\
  \bibnamefont {Ioffe}}, \bibinfo {author} {\bibfnamefont {S.~V.}\ \bibnamefont
  {Isakov}}, \bibinfo {author} {\bibfnamefont {E.}~\bibnamefont {Jeffrey}},
  \bibinfo {author} {\bibfnamefont {Z.}~\bibnamefont {Jiang}}, \bibinfo
  {author} {\bibfnamefont {K.}~\bibnamefont {Kechedzhi}}, \bibinfo {author}
  {\bibfnamefont {S.}~\bibnamefont {Kim}}, \bibinfo {author} {\bibfnamefont
  {F.}~\bibnamefont {Kostritsa}}, \bibinfo {author} {\bibfnamefont
  {D.}~\bibnamefont {Landhuis}}, \bibinfo {author} {\bibfnamefont
  {P.}~\bibnamefont {Laptev}}, \bibinfo {author} {\bibfnamefont
  {E.}~\bibnamefont {Lucero}}, \bibinfo {author} {\bibfnamefont
  {O.}~\bibnamefont {Martin}}, \bibinfo {author} {\bibfnamefont {J.~R.}\
  \bibnamefont {McClean}}, \bibinfo {author} {\bibfnamefont {T.}~\bibnamefont
  {McCourt}}, \bibinfo {author} {\bibfnamefont {X.}~\bibnamefont {Mi}},
  \bibinfo {author} {\bibfnamefont {K.~C.}\ \bibnamefont {Miao}}, \bibinfo
  {author} {\bibfnamefont {M.}~\bibnamefont {Mohseni}}, \bibinfo {author}
  {\bibfnamefont {W.}~\bibnamefont {Mruczkiewicz}}, \bibinfo {author}
  {\bibfnamefont {J.}~\bibnamefont {Mutus}}, \bibinfo {author} {\bibfnamefont
  {O.}~\bibnamefont {Naaman}}, \bibinfo {author} {\bibfnamefont
  {M.}~\bibnamefont {Neeley}}, \bibinfo {author} {\bibfnamefont
  {C.}~\bibnamefont {Neill}}, \bibinfo {author} {\bibfnamefont
  {M.}~\bibnamefont {Newman}}, \bibinfo {author} {\bibfnamefont {M.~Y.}\
  \bibnamefont {Niu}}, \bibinfo {author} {\bibfnamefont {T.~E.}\ \bibnamefont
  {O'Brien}}, \bibinfo {author} {\bibfnamefont {A.}~\bibnamefont {Opremcak}},
  \bibinfo {author} {\bibfnamefont {E.}~\bibnamefont {Ostby}}, \bibinfo
  {author} {\bibfnamefont {B.}~\bibnamefont {Pató}}, \bibinfo {author}
  {\bibfnamefont {N.}~\bibnamefont {Redd}}, \bibinfo {author} {\bibfnamefont
  {P.}~\bibnamefont {Roushan}}, \bibinfo {author} {\bibfnamefont {N.~C.}\
  \bibnamefont {Rubin}}, \bibinfo {author} {\bibfnamefont {V.}~\bibnamefont
  {Shvarts}}, \bibinfo {author} {\bibfnamefont {D.}~\bibnamefont {Strain}},
  \bibinfo {author} {\bibfnamefont {M.}~\bibnamefont {Szalay}}, \bibinfo
  {author} {\bibfnamefont {M.~D.}\ \bibnamefont {Trevithick}}, \bibinfo
  {author} {\bibfnamefont {B.}~\bibnamefont {Villalonga}}, \bibinfo {author}
  {\bibfnamefont {T.}~\bibnamefont {White}}, \bibinfo {author} {\bibfnamefont
  {Z.~J.}\ \bibnamefont {Yao}}, \bibinfo {author} {\bibfnamefont
  {P.}~\bibnamefont {Yeh}}, \bibinfo {author} {\bibfnamefont {A.}~\bibnamefont
  {Zalcman}}, \bibinfo {author} {\bibfnamefont {H.}~\bibnamefont {Neven}},
  \bibinfo {author} {\bibfnamefont {S.}~\bibnamefont {Boixo}}, \bibinfo
  {author} {\bibfnamefont {V.}~\bibnamefont {Smelyanskiy}}, \bibinfo {author}
  {\bibfnamefont {Y.}~\bibnamefont {Chen}}, \bibinfo {author} {\bibfnamefont
  {A.}~\bibnamefont {Megrant}},\ and\ \bibinfo {author} {\bibfnamefont
  {J.}~\bibnamefont {Kelly}},\ }\href@noop {} {\bibinfo {title} {Exponential
  suppression of bit or phase flip errors with repetitive error correction}}
  (\bibinfo {year} {2021}),\ \Eprint {https://arxiv.org/abs/arXiv:2102.06132}
  {arXiv:arXiv:2102.06132 [quant-ph]} \BibitemShut {NoStop}%
\bibitem [{\citenamefont {Svore}\ \emph {et~al.}(2007)\citenamefont {Svore},
  \citenamefont {Divincenzo},\ and\ \citenamefont
  {Terhal}}]{svore2007architecture2d}%
  \BibitemOpen
  \bibfield  {author} {\bibinfo {author} {\bibfnamefont {K.~M.}\ \bibnamefont
  {Svore}}, \bibinfo {author} {\bibfnamefont {D.~P.}\ \bibnamefont
  {Divincenzo}},\ and\ \bibinfo {author} {\bibfnamefont {B.~M.}\ \bibnamefont
  {Terhal}},\ }\bibfield  {title} {\bibinfo {title} {Noise threshold for a
  fault-tolerant two-dimensional lattice architecture},\ }\href@noop {}
  {\bibfield  {journal} {\bibinfo  {journal} {Quantum Info. Comput.}\ }\textbf
  {\bibinfo {volume} {7}},\ \bibinfo {pages} {297–318} (\bibinfo {year}
  {2007})}\BibitemShut {NoStop}%
\bibitem [{\citenamefont {Beverland}\ \emph {et~al.}(2021)\citenamefont
  {Beverland}, \citenamefont {Kubica},\ and\ \citenamefont
  {Svore}}]{beverland2021codeswitching}%
  \BibitemOpen
  \bibfield  {author} {\bibinfo {author} {\bibfnamefont {M.~E.}\ \bibnamefont
  {Beverland}}, \bibinfo {author} {\bibfnamefont {A.}~\bibnamefont {Kubica}},\
  and\ \bibinfo {author} {\bibfnamefont {K.~M.}\ \bibnamefont {Svore}},\
  }\bibfield  {title} {\bibinfo {title} {Cost of universality: A comparative
  study of the overhead of state distillation and code switching with color
  codes},\ }\href {https://doi.org/10.1103/PRXQuantum.2.020341} {\bibfield
  {journal} {\bibinfo  {journal} {PRX Quantum}\ }\textbf {\bibinfo {volume}
  {2}},\ \bibinfo {pages} {020341} (\bibinfo {year} {2021})}\BibitemShut
  {NoStop}%
\bibitem [{\citenamefont {Arute}\ \emph {et~al.}(2019)\citenamefont {Arute},
  \citenamefont {Arya}, \citenamefont {Babbush}, \citenamefont {Bacon},
  \citenamefont {Bardin}, \citenamefont {Barends}, \citenamefont {Biswas},
  \citenamefont {Boixo}, \citenamefont {Brandao}, \citenamefont {Buell},
  \citenamefont {Burkett}, \citenamefont {Chen}, \citenamefont {Chen},
  \citenamefont {Chiaro}, \citenamefont {Collins}, \citenamefont {Courtney},
  \citenamefont {Dunsworth}, \citenamefont {Farhi}, \citenamefont {Foxen},
  \citenamefont {Fowler}, \citenamefont {Gidney}, \citenamefont {Giustina},
  \citenamefont {Graff}, \citenamefont {Guerin}, \citenamefont {Habegger},
  \citenamefont {Harrigan}, \citenamefont {Hartmann}, \citenamefont {Ho},
  \citenamefont {Hoffmann}, \citenamefont {Huang}, \citenamefont {Humble},
  \citenamefont {Isakov}, \citenamefont {Jeffrey}, \citenamefont {Jiang},
  \citenamefont {Kafri}, \citenamefont {Kechedzhi}, \citenamefont {Kelly},
  \citenamefont {Klimov}, \citenamefont {Knysh}, \citenamefont {Korotkov},
  \citenamefont {Kostritsa}, \citenamefont {Landhuis}, \citenamefont
  {Lindmark}, \citenamefont {Lucero}, \citenamefont {Lyakh}, \citenamefont
  {Mandr{\`a}}, \citenamefont {McClean}, \citenamefont {McEwen}, \citenamefont
  {Megrant}, \citenamefont {Mi}, \citenamefont {Michielsen}, \citenamefont
  {Mohseni}, \citenamefont {Mutus}, \citenamefont {Naaman}, \citenamefont
  {Neeley}, \citenamefont {Neill}, \citenamefont {Niu}, \citenamefont {Ostby},
  \citenamefont {Petukhov}, \citenamefont {Platt}, \citenamefont {Quintana},
  \citenamefont {Rieffel}, \citenamefont {Roushan}, \citenamefont {Rubin},
  \citenamefont {Sank}, \citenamefont {Satzinger}, \citenamefont {Smelyanskiy},
  \citenamefont {Sung}, \citenamefont {Trevithick}, \citenamefont
  {Vainsencher}, \citenamefont {Villalonga}, \citenamefont {White},
  \citenamefont {Yao}, \citenamefont {Yeh}, \citenamefont {Zalcman},
  \citenamefont {Neven},\ and\ \citenamefont {Martinis}}]{arute2019supremacy}%
  \BibitemOpen
  \bibfield  {author} {\bibinfo {author} {\bibfnamefont {F.}~\bibnamefont
  {Arute}}, \bibinfo {author} {\bibfnamefont {K.}~\bibnamefont {Arya}},
  \bibinfo {author} {\bibfnamefont {R.}~\bibnamefont {Babbush}}, \bibinfo
  {author} {\bibfnamefont {D.}~\bibnamefont {Bacon}}, \bibinfo {author}
  {\bibfnamefont {J.~C.}\ \bibnamefont {Bardin}}, \bibinfo {author}
  {\bibfnamefont {R.}~\bibnamefont {Barends}}, \bibinfo {author} {\bibfnamefont
  {R.}~\bibnamefont {Biswas}}, \bibinfo {author} {\bibfnamefont
  {S.}~\bibnamefont {Boixo}}, \bibinfo {author} {\bibfnamefont {F.~G. S.~L.}\
  \bibnamefont {Brandao}}, \bibinfo {author} {\bibfnamefont {D.~A.}\
  \bibnamefont {Buell}}, \bibinfo {author} {\bibfnamefont {B.}~\bibnamefont
  {Burkett}}, \bibinfo {author} {\bibfnamefont {Y.}~\bibnamefont {Chen}},
  \bibinfo {author} {\bibfnamefont {Z.}~\bibnamefont {Chen}}, \bibinfo {author}
  {\bibfnamefont {B.}~\bibnamefont {Chiaro}}, \bibinfo {author} {\bibfnamefont
  {R.}~\bibnamefont {Collins}}, \bibinfo {author} {\bibfnamefont
  {W.}~\bibnamefont {Courtney}}, \bibinfo {author} {\bibfnamefont
  {A.}~\bibnamefont {Dunsworth}}, \bibinfo {author} {\bibfnamefont
  {E.}~\bibnamefont {Farhi}}, \bibinfo {author} {\bibfnamefont
  {B.}~\bibnamefont {Foxen}}, \bibinfo {author} {\bibfnamefont
  {A.}~\bibnamefont {Fowler}}, \bibinfo {author} {\bibfnamefont
  {C.}~\bibnamefont {Gidney}}, \bibinfo {author} {\bibfnamefont
  {M.}~\bibnamefont {Giustina}}, \bibinfo {author} {\bibfnamefont
  {R.}~\bibnamefont {Graff}}, \bibinfo {author} {\bibfnamefont
  {K.}~\bibnamefont {Guerin}}, \bibinfo {author} {\bibfnamefont
  {S.}~\bibnamefont {Habegger}}, \bibinfo {author} {\bibfnamefont {M.~P.}\
  \bibnamefont {Harrigan}}, \bibinfo {author} {\bibfnamefont {M.~J.}\
  \bibnamefont {Hartmann}}, \bibinfo {author} {\bibfnamefont {A.}~\bibnamefont
  {Ho}}, \bibinfo {author} {\bibfnamefont {M.}~\bibnamefont {Hoffmann}},
  \bibinfo {author} {\bibfnamefont {T.}~\bibnamefont {Huang}}, \bibinfo
  {author} {\bibfnamefont {T.~S.}\ \bibnamefont {Humble}}, \bibinfo {author}
  {\bibfnamefont {S.~V.}\ \bibnamefont {Isakov}}, \bibinfo {author}
  {\bibfnamefont {E.}~\bibnamefont {Jeffrey}}, \bibinfo {author} {\bibfnamefont
  {Z.}~\bibnamefont {Jiang}}, \bibinfo {author} {\bibfnamefont
  {D.}~\bibnamefont {Kafri}}, \bibinfo {author} {\bibfnamefont
  {K.}~\bibnamefont {Kechedzhi}}, \bibinfo {author} {\bibfnamefont
  {J.}~\bibnamefont {Kelly}}, \bibinfo {author} {\bibfnamefont {P.~V.}\
  \bibnamefont {Klimov}}, \bibinfo {author} {\bibfnamefont {S.}~\bibnamefont
  {Knysh}}, \bibinfo {author} {\bibfnamefont {A.}~\bibnamefont {Korotkov}},
  \bibinfo {author} {\bibfnamefont {F.}~\bibnamefont {Kostritsa}}, \bibinfo
  {author} {\bibfnamefont {D.}~\bibnamefont {Landhuis}}, \bibinfo {author}
  {\bibfnamefont {M.}~\bibnamefont {Lindmark}}, \bibinfo {author}
  {\bibfnamefont {E.}~\bibnamefont {Lucero}}, \bibinfo {author} {\bibfnamefont
  {D.}~\bibnamefont {Lyakh}}, \bibinfo {author} {\bibfnamefont
  {S.}~\bibnamefont {Mandr{\`a}}}, \bibinfo {author} {\bibfnamefont {J.~R.}\
  \bibnamefont {McClean}}, \bibinfo {author} {\bibfnamefont {M.}~\bibnamefont
  {McEwen}}, \bibinfo {author} {\bibfnamefont {A.}~\bibnamefont {Megrant}},
  \bibinfo {author} {\bibfnamefont {X.}~\bibnamefont {Mi}}, \bibinfo {author}
  {\bibfnamefont {K.}~\bibnamefont {Michielsen}}, \bibinfo {author}
  {\bibfnamefont {M.}~\bibnamefont {Mohseni}}, \bibinfo {author} {\bibfnamefont
  {J.}~\bibnamefont {Mutus}}, \bibinfo {author} {\bibfnamefont
  {O.}~\bibnamefont {Naaman}}, \bibinfo {author} {\bibfnamefont
  {M.}~\bibnamefont {Neeley}}, \bibinfo {author} {\bibfnamefont
  {C.}~\bibnamefont {Neill}}, \bibinfo {author} {\bibfnamefont {M.~Y.}\
  \bibnamefont {Niu}}, \bibinfo {author} {\bibfnamefont {E.}~\bibnamefont
  {Ostby}}, \bibinfo {author} {\bibfnamefont {A.}~\bibnamefont {Petukhov}},
  \bibinfo {author} {\bibfnamefont {J.~C.}\ \bibnamefont {Platt}}, \bibinfo
  {author} {\bibfnamefont {C.}~\bibnamefont {Quintana}}, \bibinfo {author}
  {\bibfnamefont {E.~G.}\ \bibnamefont {Rieffel}}, \bibinfo {author}
  {\bibfnamefont {P.}~\bibnamefont {Roushan}}, \bibinfo {author} {\bibfnamefont
  {N.~C.}\ \bibnamefont {Rubin}}, \bibinfo {author} {\bibfnamefont
  {D.}~\bibnamefont {Sank}}, \bibinfo {author} {\bibfnamefont {K.~J.}\
  \bibnamefont {Satzinger}}, \bibinfo {author} {\bibfnamefont {V.}~\bibnamefont
  {Smelyanskiy}}, \bibinfo {author} {\bibfnamefont {K.~J.}\ \bibnamefont
  {Sung}}, \bibinfo {author} {\bibfnamefont {M.~D.}\ \bibnamefont
  {Trevithick}}, \bibinfo {author} {\bibfnamefont {A.}~\bibnamefont
  {Vainsencher}}, \bibinfo {author} {\bibfnamefont {B.}~\bibnamefont
  {Villalonga}}, \bibinfo {author} {\bibfnamefont {T.}~\bibnamefont {White}},
  \bibinfo {author} {\bibfnamefont {Z.~J.}\ \bibnamefont {Yao}}, \bibinfo
  {author} {\bibfnamefont {P.}~\bibnamefont {Yeh}}, \bibinfo {author}
  {\bibfnamefont {A.}~\bibnamefont {Zalcman}}, \bibinfo {author} {\bibfnamefont
  {H.}~\bibnamefont {Neven}},\ and\ \bibinfo {author} {\bibfnamefont {J.~M.}\
  \bibnamefont {Martinis}},\ }\bibfield  {title} {\bibinfo {title} {Quantum
  supremacy using a programmable superconducting processor},\ }\href
  {https://doi.org/10.1038/s41586-019-1666-5} {\bibfield  {journal} {\bibinfo
  {journal} {Nature}\ }\textbf {\bibinfo {volume} {574}},\ \bibinfo {pages}
  {505--510} (\bibinfo {year} {2019})}\BibitemShut {NoStop}%
\bibitem [{\citenamefont {Neg\^{\i}rneac}\ \emph {et~al.}(2021)\citenamefont
  {Neg\^{\i}rneac}, \citenamefont {Ali}, \citenamefont {Muthusubramanian},
  \citenamefont {Battistel}, \citenamefont {Sagastizabal}, \citenamefont
  {Moreira}, \citenamefont {Marques}, \citenamefont {Vlothuizen}, \citenamefont
  {Beekman}, \citenamefont {Zachariadis}, \citenamefont {Haider}, \citenamefont
  {Bruno},\ and\ \citenamefont {DiCarlo}}]{nagirneac2019fastcz}%
  \BibitemOpen
  \bibfield  {author} {\bibinfo {author} {\bibfnamefont {V.}~\bibnamefont
  {Neg\^{\i}rneac}}, \bibinfo {author} {\bibfnamefont {H.}~\bibnamefont {Ali}},
  \bibinfo {author} {\bibfnamefont {N.}~\bibnamefont {Muthusubramanian}},
  \bibinfo {author} {\bibfnamefont {F.}~\bibnamefont {Battistel}}, \bibinfo
  {author} {\bibfnamefont {R.}~\bibnamefont {Sagastizabal}}, \bibinfo {author}
  {\bibfnamefont {M.~S.}\ \bibnamefont {Moreira}}, \bibinfo {author}
  {\bibfnamefont {J.~F.}\ \bibnamefont {Marques}}, \bibinfo {author}
  {\bibfnamefont {W.~J.}\ \bibnamefont {Vlothuizen}}, \bibinfo {author}
  {\bibfnamefont {M.}~\bibnamefont {Beekman}}, \bibinfo {author} {\bibfnamefont
  {C.}~\bibnamefont {Zachariadis}}, \bibinfo {author} {\bibfnamefont
  {N.}~\bibnamefont {Haider}}, \bibinfo {author} {\bibfnamefont
  {A.}~\bibnamefont {Bruno}},\ and\ \bibinfo {author} {\bibfnamefont
  {L.}~\bibnamefont {DiCarlo}},\ }\bibfield  {title} {\bibinfo {title}
  {High-fidelity controlled-$z$ gate with maximal intermediate leakage
  operating at the speed limit in a superconducting quantum processor},\ }\href
  {https://doi.org/10.1103/PhysRevLett.126.220502} {\bibfield  {journal}
  {\bibinfo  {journal} {Phys. Rev. Lett.}\ }\textbf {\bibinfo {volume} {126}},\
  \bibinfo {pages} {220502} (\bibinfo {year} {2021})}\BibitemShut {NoStop}%
\bibitem [{\citenamefont {Crippa}\ \emph {et~al.}(2019)\citenamefont {Crippa},
  \citenamefont {Ezzouch}, \citenamefont {Apr{\'a}}, \citenamefont {Amisse},
  \citenamefont {Lavi{\'e}ville}, \citenamefont {Hutin}, \citenamefont
  {Bertrand}, \citenamefont {Vinet}, \citenamefont {Urdampilleta},
  \citenamefont {Meunier}, \citenamefont {Sanquer}, \citenamefont {Jehl},
  \citenamefont {Maurand},\ and\ \citenamefont
  {De~Franceschi}}]{crippa99spinsgatesensing}%
  \BibitemOpen
  \bibfield  {author} {\bibinfo {author} {\bibfnamefont {A.}~\bibnamefont
  {Crippa}}, \bibinfo {author} {\bibfnamefont {R.}~\bibnamefont {Ezzouch}},
  \bibinfo {author} {\bibfnamefont {A.}~\bibnamefont {Apr{\'a}}}, \bibinfo
  {author} {\bibfnamefont {A.}~\bibnamefont {Amisse}}, \bibinfo {author}
  {\bibfnamefont {R.}~\bibnamefont {Lavi{\'e}ville}}, \bibinfo {author}
  {\bibfnamefont {L.}~\bibnamefont {Hutin}}, \bibinfo {author} {\bibfnamefont
  {B.}~\bibnamefont {Bertrand}}, \bibinfo {author} {\bibfnamefont
  {M.}~\bibnamefont {Vinet}}, \bibinfo {author} {\bibfnamefont
  {M.}~\bibnamefont {Urdampilleta}}, \bibinfo {author} {\bibfnamefont
  {T.}~\bibnamefont {Meunier}}, \bibinfo {author} {\bibfnamefont
  {M.}~\bibnamefont {Sanquer}}, \bibinfo {author} {\bibfnamefont
  {X.}~\bibnamefont {Jehl}}, \bibinfo {author} {\bibfnamefont {R.}~\bibnamefont
  {Maurand}},\ and\ \bibinfo {author} {\bibfnamefont {S.}~\bibnamefont
  {De~Franceschi}},\ }\bibfield  {title} {\bibinfo {title} {Gate-reflectometry
  dispersive readout and coherent control of a spin qubit in silicon},\ }\href
  {https://doi.org/10.1038/s41467-019-10848-z} {\bibfield  {journal} {\bibinfo
  {journal} {Nature Communications}\ }\textbf {\bibinfo {volume} {10}},\
  \bibinfo {pages} {2776} (\bibinfo {year} {2019})}\BibitemShut {NoStop}%
\bibitem [{\citenamefont {Jurcevic}\ \emph {et~al.}(2021)\citenamefont
  {Jurcevic}, \citenamefont {Javadi-Abhari}, \citenamefont {Bishop},
  \citenamefont {Lauer}, \citenamefont {Bogorin}, \citenamefont {Brink},
  \citenamefont {Capelluto}, \citenamefont {Günlük}, \citenamefont {Itoko},
  \citenamefont {Kanazawa}, \citenamefont {Kandala}, \citenamefont {Keefe},
  \citenamefont {Krsulich}, \citenamefont {Landers}, \citenamefont
  {Lewandowski}, \citenamefont {McClure}, \citenamefont {Nannicini},
  \citenamefont {Narasgond}, \citenamefont {Nayfeh}, \citenamefont {Pritchett},
  \citenamefont {Rothwell}, \citenamefont {Srinivasan}, \citenamefont
  {Sundaresan}, \citenamefont {Wang}, \citenamefont {Wei}, \citenamefont
  {Wood}, \citenamefont {Yau}, \citenamefont {Zhang}, \citenamefont {Dial},
  \citenamefont {Chow},\ and\ \citenamefont {Gambetta}}]{jurcevic2021volume64}%
  \BibitemOpen
  \bibfield  {author} {\bibinfo {author} {\bibfnamefont {P.}~\bibnamefont
  {Jurcevic}}, \bibinfo {author} {\bibfnamefont {A.}~\bibnamefont
  {Javadi-Abhari}}, \bibinfo {author} {\bibfnamefont {L.~S.}\ \bibnamefont
  {Bishop}}, \bibinfo {author} {\bibfnamefont {I.}~\bibnamefont {Lauer}},
  \bibinfo {author} {\bibfnamefont {D.~F.}\ \bibnamefont {Bogorin}}, \bibinfo
  {author} {\bibfnamefont {M.}~\bibnamefont {Brink}}, \bibinfo {author}
  {\bibfnamefont {L.}~\bibnamefont {Capelluto}}, \bibinfo {author}
  {\bibfnamefont {O.}~\bibnamefont {Günlük}}, \bibinfo {author}
  {\bibfnamefont {T.}~\bibnamefont {Itoko}}, \bibinfo {author} {\bibfnamefont
  {N.}~\bibnamefont {Kanazawa}}, \bibinfo {author} {\bibfnamefont
  {A.}~\bibnamefont {Kandala}}, \bibinfo {author} {\bibfnamefont {G.~A.}\
  \bibnamefont {Keefe}}, \bibinfo {author} {\bibfnamefont {K.}~\bibnamefont
  {Krsulich}}, \bibinfo {author} {\bibfnamefont {W.}~\bibnamefont {Landers}},
  \bibinfo {author} {\bibfnamefont {E.~P.}\ \bibnamefont {Lewandowski}},
  \bibinfo {author} {\bibfnamefont {D.~T.}\ \bibnamefont {McClure}}, \bibinfo
  {author} {\bibfnamefont {G.}~\bibnamefont {Nannicini}}, \bibinfo {author}
  {\bibfnamefont {A.}~\bibnamefont {Narasgond}}, \bibinfo {author}
  {\bibfnamefont {H.~M.}\ \bibnamefont {Nayfeh}}, \bibinfo {author}
  {\bibfnamefont {E.}~\bibnamefont {Pritchett}}, \bibinfo {author}
  {\bibfnamefont {M.~B.}\ \bibnamefont {Rothwell}}, \bibinfo {author}
  {\bibfnamefont {S.}~\bibnamefont {Srinivasan}}, \bibinfo {author}
  {\bibfnamefont {N.}~\bibnamefont {Sundaresan}}, \bibinfo {author}
  {\bibfnamefont {C.}~\bibnamefont {Wang}}, \bibinfo {author} {\bibfnamefont
  {K.~X.}\ \bibnamefont {Wei}}, \bibinfo {author} {\bibfnamefont {C.~J.}\
  \bibnamefont {Wood}}, \bibinfo {author} {\bibfnamefont {J.-B.}\ \bibnamefont
  {Yau}}, \bibinfo {author} {\bibfnamefont {E.~J.}\ \bibnamefont {Zhang}},
  \bibinfo {author} {\bibfnamefont {O.~E.}\ \bibnamefont {Dial}}, \bibinfo
  {author} {\bibfnamefont {J.~M.}\ \bibnamefont {Chow}},\ and\ \bibinfo
  {author} {\bibfnamefont {J.~M.}\ \bibnamefont {Gambetta}},\ }\bibfield
  {title} {\bibinfo {title} {Demonstration of quantum volume 64 on a
  superconducting quantum computing system},\ }\href
  {https://doi.org/10.1088/2058-9565/abe519} {\bibfield  {journal} {\bibinfo
  {journal} {Quantum Science and Technology}\ }\textbf {\bibinfo {volume}
  {6}},\ \bibinfo {pages} {025020} (\bibinfo {year} {2021})}\BibitemShut
  {NoStop}%
\bibitem [{\citenamefont {Xue}\ \emph {et~al.}(2021)\citenamefont {Xue},
  \citenamefont {Russ}, \citenamefont {Samkharadze}, \citenamefont {Undseth},
  \citenamefont {Sammak}, \citenamefont {Scappucci},\ and\ \citenamefont
  {Vandersypen}}]{xue2021spins99}%
  \BibitemOpen
  \bibfield  {author} {\bibinfo {author} {\bibfnamefont {X.}~\bibnamefont
  {Xue}}, \bibinfo {author} {\bibfnamefont {M.}~\bibnamefont {Russ}}, \bibinfo
  {author} {\bibfnamefont {N.}~\bibnamefont {Samkharadze}}, \bibinfo {author}
  {\bibfnamefont {B.}~\bibnamefont {Undseth}}, \bibinfo {author} {\bibfnamefont
  {A.}~\bibnamefont {Sammak}}, \bibinfo {author} {\bibfnamefont
  {G.}~\bibnamefont {Scappucci}},\ and\ \bibinfo {author} {\bibfnamefont
  {L.~M.~K.}\ \bibnamefont {Vandersypen}},\ }\href@noop {} {\bibinfo {title}
  {Computing with spin qubits at the surface code error threshold}} (\bibinfo
  {year} {2021}),\ \Eprint {https://arxiv.org/abs/2107.00628} {arXiv:2107.00628
  [quant-ph]} \BibitemShut {NoStop}%
\bibitem [{\citenamefont {Bombin}\ and\ \citenamefont
  {Martin-Delgado}(2006)}]{bombin2006topological}%
  \BibitemOpen
  \bibfield  {author} {\bibinfo {author} {\bibfnamefont {H.}~\bibnamefont
  {Bombin}}\ and\ \bibinfo {author} {\bibfnamefont {M.~A.}\ \bibnamefont
  {Martin-Delgado}},\ }\bibfield  {title} {\bibinfo {title} {Topological
  quantum distillation},\ }\href@noop {} {\bibfield  {journal} {\bibinfo
  {journal} {Physical Review Letters}\ }\textbf {\bibinfo {volume} {97}},\
  \bibinfo {pages} {180501} (\bibinfo {year} {2006})}\BibitemShut {NoStop}%
\bibitem [{\citenamefont {Ataides}\ \emph {et~al.}(2021)\citenamefont
  {Ataides}, \citenamefont {Tuckett}, \citenamefont {Bartlett}, \citenamefont
  {Flammia},\ and\ \citenamefont {Brown}}]{ataides2021xzzx}%
  \BibitemOpen
  \bibfield  {author} {\bibinfo {author} {\bibfnamefont {J.~P.~B.}\
  \bibnamefont {Ataides}}, \bibinfo {author} {\bibfnamefont {D.~K.}\
  \bibnamefont {Tuckett}}, \bibinfo {author} {\bibfnamefont {S.~D.}\
  \bibnamefont {Bartlett}}, \bibinfo {author} {\bibfnamefont {S.~T.}\
  \bibnamefont {Flammia}},\ and\ \bibinfo {author} {\bibfnamefont {B.~J.}\
  \bibnamefont {Brown}},\ }\bibfield  {title} {\bibinfo {title} {The xzzx
  surface code},\ }\href@noop {} {\bibfield  {journal} {\bibinfo  {journal}
  {Nature communications}\ }\textbf {\bibinfo {volume} {12}},\ \bibinfo {pages}
  {1--12} (\bibinfo {year} {2021})}\BibitemShut {NoStop}%
\bibitem [{\citenamefont {Delfosse}\ \emph
  {et~al.}(2021{\natexlab{b}})\citenamefont {Delfosse}, \citenamefont {Londe},\
  and\ \citenamefont {Beverland}}]{delfosse2021toward}%
  \BibitemOpen
  \bibfield  {author} {\bibinfo {author} {\bibfnamefont {N.}~\bibnamefont
  {Delfosse}}, \bibinfo {author} {\bibfnamefont {V.}~\bibnamefont {Londe}},\
  and\ \bibinfo {author} {\bibfnamefont {M.}~\bibnamefont {Beverland}},\
  }\bibfield  {title} {\bibinfo {title} {Toward a union-find decoder for
  quantum ldpc codes},\ }\href@noop {} {\bibfield  {journal} {\bibinfo
  {journal} {arXiv preprint arXiv:2103.08049}\ } (\bibinfo {year}
  {2021}{\natexlab{b}})}\BibitemShut {NoStop}%
\bibitem [{\citenamefont {Gambetta}\ \emph {et~al.}(2008)\citenamefont
  {Gambetta}, \citenamefont {Blais}, \citenamefont {Boissonneault},
  \citenamefont {Houck}, \citenamefont {Schuster},\ and\ \citenamefont
  {Girvin}}]{gambetta2008measurement2}%
  \BibitemOpen
  \bibfield  {author} {\bibinfo {author} {\bibfnamefont {J.}~\bibnamefont
  {Gambetta}}, \bibinfo {author} {\bibfnamefont {A.}~\bibnamefont {Blais}},
  \bibinfo {author} {\bibfnamefont {M.}~\bibnamefont {Boissonneault}}, \bibinfo
  {author} {\bibfnamefont {A.~A.}\ \bibnamefont {Houck}}, \bibinfo {author}
  {\bibfnamefont {D.~I.}\ \bibnamefont {Schuster}},\ and\ \bibinfo {author}
  {\bibfnamefont {S.~M.}\ \bibnamefont {Girvin}},\ }\bibfield  {title}
  {\bibinfo {title} {Quantum trajectory approach to circuit qed: Quantum jumps
  and the zeno effect},\ }\href {https://doi.org/10.1103/PhysRevA.77.012112}
  {\bibfield  {journal} {\bibinfo  {journal} {Phys. Rev. A}\ }\textbf {\bibinfo
  {volume} {77}},\ \bibinfo {pages} {012112} (\bibinfo {year}
  {2008})}\BibitemShut {NoStop}%
\bibitem [{\citenamefont {Gidney}(2020)}]{gidneyDecorrelatedDepBlog}%
  \BibitemOpen
  \bibfield  {author} {\bibinfo {author} {\bibfnamefont {C.}~\bibnamefont
  {Gidney}},\ }\href@noop {} {\bibinfo {title} {Decorrelated depolarization}},\
  \bibinfo {howpublished} {\url{https://algassert.com/post/2001}} (\bibinfo
  {year} {2020}),\ \bibinfo {note} {accessed: 2021-07-26}\BibitemShut {NoStop}%
\bibitem [{\citenamefont {Gelman}\ \emph {et~al.}(2013)\citenamefont {Gelman},
  \citenamefont {Carlin}, \citenamefont {Stern}, \citenamefont {Dunson},
  \citenamefont {Vehtari},\ and\ \citenamefont {Rubin}}]{gelman2013bayesian}%
  \BibitemOpen
  \bibfield  {author} {\bibinfo {author} {\bibfnamefont {A.}~\bibnamefont
  {Gelman}}, \bibinfo {author} {\bibfnamefont {J.~B.}\ \bibnamefont {Carlin}},
  \bibinfo {author} {\bibfnamefont {H.~S.}\ \bibnamefont {Stern}}, \bibinfo
  {author} {\bibfnamefont {D.~B.}\ \bibnamefont {Dunson}}, \bibinfo {author}
  {\bibfnamefont {A.}~\bibnamefont {Vehtari}},\ and\ \bibinfo {author}
  {\bibfnamefont {D.~B.}\ \bibnamefont {Rubin}},\ }\href@noop {} {\emph
  {\bibinfo {title} {Bayesian data analysis}}}\ (\bibinfo  {publisher} {CRC
  press},\ \bibinfo {year} {2013})\BibitemShut {NoStop}%
\end{thebibliography}
\end{document}